\documentclass[twocolumn,showpacs,amsmath,amssymb,aps,prx,footinbib,floatfix,superscriptaddress,longbibliography]{revtex4-1}
\usepackage{graphicx}
\usepackage{braket}
\usepackage{savesym}
\usepackage{algpseudocode}
\usepackage{float} 
\makeatletter
\let\newfloat\newfloat@ltx
\makeatother
\usepackage{algorithm}
\usepackage{amsfonts}
\usepackage{bm}
\usepackage{color}
\usepackage{soul}
\usepackage{subfigure}
\usepackage{wasysym}
\usepackage{upgreek}
\usepackage{booktabs}
\usepackage{tikz}
\usepackage{dsfont}
\usepackage{mathtools}
\usepackage[usenames,dvipsnames,svgnames,table]{xcolor}
\usepackage{hyperref}
\usepackage{xurl} 
\definecolor{revgreen}{rgb}{0,0.45,0.15}
\newtheorem{proposition}{Proposition}
\newenvironment{proof}[1][Proof]{\par\noindent\emph{#1.}\ }{\hfill$\square$\par\medskip}
\usepackage{hhline} 
\usepackage{multirow} 
\usepackage[normalem]{ulem}
\usepackage{siunitx}   
\usepackage{caption}
\usepackage{makecell}
\usepackage[english]{babel}

\hypersetup{
	colorlinks=true,       
	linkcolor=VioletRed,  
    citecolor=JungleGreen,        
    filecolor=Orchid,      
    urlcolor=black           
}
\newcommand{\be}{\begin{eqnarray}}
\newcommand{\ee}{\end{eqnarray}}

\def\mh#1{\textcolor{blue}{#1}}

\def\aaa#1{\textcolor{red}{#1}}
\definecolor{adilcolor}{rgb}{0.3,0,0.5}
\def\cld#1{\textcolor{red}{#1}}

\begin{document}

\newcommand{\crit}[1]{\par\noindent\aaa{\textbf{[AMA:} #1\textbf{]}}\par}

\newcommand{\aacld}[1]{\noindent\cld{\textbf{[AAA:} #1\textbf{]}}}

\title{Quantum Geometric Tensor Preconditioning for Stable Training of Recurrent Neural Quantum States}


\date{\today}

\author{Adil Attar}
\affiliation{Department of Applied Mathematics, University of Waterloo, Waterloo, ON N2L 3G1, Canada}
\author{Amine M. Aboussalah}
\affiliation{NYU Tandon School of Engineering, 6 MetroTech Center, Brooklyn, New York 11201, United States}
\author{Mohamed Hibat-Allah}
\email{mohamed.hibatallah@uwaterloo.ca}
\affiliation{Department of Applied Mathematics, University of Waterloo, Waterloo, ON N2L 3G1, Canada}
\affiliation{Vector Institute,  Toronto,  Ontario,  M5G 0C6,  Canada}

\begin{abstract}
Neural Quantum States (NQS) provide a powerful neural network–based variational framework for representing many-body wave functions and solving for ground states. Recurrent Neural Networks (RNNs) are particularly promising owing to their relatively low computational cost and their autoregressive property, which enables perfect sampling. Recently, RNNs have been reported to be unstable under curvature-based optimizers such as the minimum-step stochastic reconfiguration (minSR) method. In this paper, we address this perceived limitation and show that minSR can be stabilized through simple regularization techniques, enabling robust training of RNN-based NQS with only a few samples. Our approach outperforms the Adam optimizer on the one-dimensional transverse-field Ising model and the one-dimensional cluster state, and provides competitive results on the two-dimensional Heisenberg and $J_1-J_2$ models. This work offers a promising pathway for using modern optimization techniques with autoregressive NQS to address open questions in quantum simulation.
\end{abstract}

\maketitle


\section{Introduction}
Simulating quantum matter with neural-network quantum states (NQS) is an emerging field that has delivered promising results in the quantum matter community~\cite{Carleo_original_2017, NQS-rev, Dawid_2025}. This achievement has been enabled by different neural network architectures ranging from restricted Boltzmann machines (RBMs)~\cite{Carleo_original_2017,Choo2020,v2024restrictedboltzmannmachinenetwork, PhysRevB.96.205152, PhysRevX.11.031034}, feedforward neural networks~\cite{Di_Luo, zi2018}, convolutional neural networks (CNNs)~\cite{Choo_2019, PhysRevB.108.054410, Chen_2024}, to recurrent neural networks (RNNs)~\cite{RNN-Wavefunc, hibatallah2022supplementingrecurrentneuralnetwork,Wu_2023,roth2020iterativeretrainingquantumspin, luo2021gauge, Schuyler, nh89-6jmf, Hibat_Allah_2025, merali2026parallelscanrecurrentneural} and Transformer-based wave functions~\cite{Zhang_2023, PhysRevLett.130.236401, minSR-lianlg, 7xwp-25y9, Sprague_2024} for benchmarks beyond ground state physics, ranging from real-time quantum dynamics of pure states~\cite{Carleo_original_2017,PhysRevLett.125.100503, kqvx-dl54, Van_de_Walle_2025, Sinibaldi2023unbiasingtime}, dynamics of open quantum systems~\cite{PhysRevLett.122.250503, PhysRevLett.128.090501, PhysRevLett.127.230501}, to thermal quantum many-body physics~\cite{Nomura_2021, nys2024realtimequantumdynamicsthermal, PhysRevLett.120.240503, PhysRevB.106.165111, KUMAR2026170362}. A cornerstone principle enabling these advances is the variational Monte Carlo (VMC) framework~\cite{Becca_Sorella_2017}, where one minimizes the variational energy, i.e., the expectation value of a Hamiltonian $\hat{H}$ in a variational wave function $|\Psi_{\boldsymbol{\theta}} \rangle$ with parameters $\boldsymbol{\theta}$, in the particular task of ground state approximation. This minimization has largely been performed in the NQS community using gradient descent through either first-order or second-order optimization techniques. One of the most effective methods is stochastic reconfiguration (SR)~\cite{PhysRevB.61.2599, Becca_Sorella_2017}, a second-order optimization method known in the machine-learning literature as natural gradient~\cite{amariNG}. This technique is well-principled thanks to its correspondence to imaginary time evolution~\cite{PhysRevB.61.2599,Becca_Sorella_2017, PhysRevResearch.2.023232, Dash2025}. This method has been extended to its minimal counterpart: minimum-step stochastic reconfiguration (minSR)~\cite{Chen_2024,minSR-lianlg}, which provides a significant speedup when the number of samples is small compared to the total number of variational parameters.

Second-order optimization has remained largely unexplored for RNNs due to its reported ill-conditioning~\cite{donatella23, Lange2024, KUMAR2026170362, duque2026timerevisitingneuralquantum}. This challenge is mainly related to the rapid decay of the spectrum of the Fisher matrix, which makes its inversion an ill-defined problem. For this architecture, the first-order Adam optimizer~\cite{AdamOriginal} is currently the most widely used choice, with promising results in ground state benchmarks~\cite{hibatallah2022supplementingrecurrentneuralnetwork, Hibat_Allah_2025}. Recently, the natural gradient optimization approach with autoregressive models has seen some success in the context of classical statistical physics~\cite{Liu_2025} and ab-initio quantum chemistry~\cite{autoregressiveChem}. Thus, a natural question is whether it is possible to stabilize the optimization of RNN wave functions using minSR despite previous reports of instability in the NQS literature. A second important question is whether a subsampled Fisher matrix in minSR with few samples still provides a useful preconditioner during optimization. In this work, we address both questions and develop a reproducible stability framework for minSR applied to RNN wave functions on prototypical one- and two-dimensional quantum benchmarks, supported by conditioning analysis, Fisher spectrum diagnostics, and controlled ablations. In particular, we demonstrate a significant advantage over the Adam optimizer in one spatial dimension on the ferromagnetic transverse-field Ising model (TFIM) and cluster state benchmarks. In two spatial dimensions, we show that minSR can be made stable and comparable to the Adam optimizer on the square lattice Heisenberg and $J_1-J_2$ models. Our results also show a speed-up in convergence for minSR optimization. Additionally, all of our investigations were conducted using only a few hundred Monte Carlo samples on a single GPU, highlighting the computational efficiency of RNN wave functions.


\section{Methods}

\subsection{Recurrent Neural Network wave function}

\begin{figure}
 \centering   \includegraphics[width=\columnwidth]{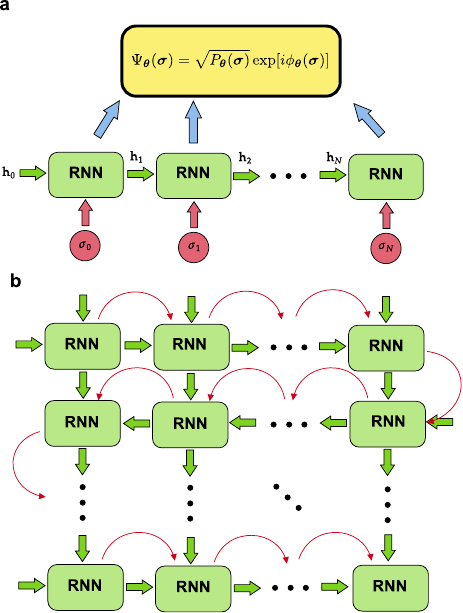}
 \caption{\textbf{Illustration of RNN wave functions.}  a) Schematic of an RNN wave function. A series of spins $\sigma_i$ are sequentially provided as an input into the RNN cell (shown in green), which produces a list of outputs  $\bm{y}_i$ (shown in blue).  Subsequently, the Softmax and/or Softsign layers are applied to compute the amplitude and phase of the wave function, respectively. The outputs of the RNN can be used to calculate the wave function through its amplitude and phase. b) Schematic of a 2D RNN. Each RNN cell (shown in green) takes a spin in a 2D array as an input and outputs two identical hidden states, drawn as green arrows, to each of its nearest neighbors. The autoregressive property can be applied to 2D RNNs through sampling in a zigzag pattern, as illustrated above using red arrows.}
 \label{fig:RNNWF}
\end{figure}

We use RNNs as a variational quantum many-body wave function \cite{RNN-Wavefunc} over a lattice of spins $\bm{\sigma}$, written as
\begin{align}
    \ket{\Psi_{\bm{\theta}}} = \sum_{\bm{\sigma}} \Psi_{\bm{\theta}}(\bm{\sigma})\ket{\bm{\sigma}}.
\end{align}
In the computational basis, an RNN wave function can be expressed as
\begin{equation}
    \Psi_{\boldsymbol{\theta}}(\boldsymbol{\sigma}) = \sqrt{P_{\boldsymbol{\theta}}(\boldsymbol{\sigma})}\exp{\bigl[i\phi_{\boldsymbol{\theta}}(\boldsymbol{\sigma})\bigr]},
\end{equation}
where $\boldsymbol{\sigma} = (\sigma_1, \dots, \sigma_N)$ denotes the lattice-spin configuration and each $\sigma_i$ can take one of $d_{\sigma}$ possible values, with $d_{\sigma}=2$ for spin-$\frac{1}{2}$ systems. $\boldsymbol{\theta}$ is the set of variational parameters, which includes the weights and biases of our RNN. $P_{\boldsymbol{\theta}}(\boldsymbol{\sigma})$ is the probability distribution over the spins, and $\phi_{\boldsymbol{\theta}}(\boldsymbol{\sigma})$ is the phase of the wave function.

RNN wave functions stand out for their ability to autoregressively sample $P_{\boldsymbol{\theta}}$. This procedure is performed by decomposing $P_{\bm{\theta}}$ using the chain rule of probabilities as
    \begin{equation}
        P_{\bm{\theta}}(\boldsymbol{\sigma}) = P_{\bm{\theta}}(\sigma_1)P_{\bm{\theta}}(\sigma_2|\sigma_1)\dots P_{\bm{\theta}}(\sigma_N|\sigma_{1}, \dots ,\sigma_{N-1}),
    \end{equation}
and directly modeling these conditional probabilities, which are univariate multinomial distributions that can be directly sampled. Through the process of sampling each conditional probability $P_{\bm{\theta}}(\sigma_i|\sigma_{<i})$, an RNN wave function can efficiently produce uncorrelated samples for estimating observables and avoid Markov Chain Monte Carlo (MCMC) schemes~\cite{RNN-Wavefunc,Sharir_2020}. The phases of the wave function can be decomposed similarly
    \begin{equation}
        \phi(\boldsymbol{\sigma}) = \sum_i\phi_i(\sigma_i|\sigma_{<i}),
    \end{equation}
where the conditional phases $\phi_i$ and the conditional probabilities are calculated using the RNN output (shown in Fig.~\ref{fig:RNNWF}(a)), and their exact form depends on the specific RNN architecture. RNNs make the core assumption that a hidden state $\bm{h}_i$ can effectively summarize all previously seen spins, such that
\begin{align}
    P_{\bm{\theta}}(\sigma_i|\sigma_{<i})=P_{\bm{\theta}}(\sigma_i|\bm{h}_i),\\
    \phi(\boldsymbol{\sigma}) = \sum_i\phi_i(\sigma_i|\bm{h}_i).
\end{align}
We depict the RNN wave function generally in Fig.~\ref{fig:RNNWF}(a). Here we can define two types of RNN wave functions~\cite{RNN-Wavefunc}, positive RNN (pRNN) wave functions, which do not include a phase term and are applicable in the case of stoquastic Hamiltonians~\cite{RNN-Wavefunc, Hibat_Allah_2025}. A second class is the complex RNN (cRNN) wave function, which includes both a probability and a phase term and can model ground states of non-stoquastic Hamiltonians~\cite{RNN-Wavefunc, nh89-6jmf}. Note that the conditional probabilities are given by
\begin{equation}
    P_{\bm{\theta}}(\sigma_i|\sigma_{<i}) = \textrm{softmax}(U\bm{h}_i+\bm{c}) \cdot\bm{\sigma}_i,
\end{equation}
where $\bm{\sigma}_i$ is the one-hot encoding of $\sigma_i$, and $\bm{h}_i \in \mathbb{R}^{d_h}$ is the hidden state, with dimension $d_h$. $U\in\mathbb{R}^{d_{\sigma}\times d_h}$ and $\bm{c}\in\mathbb{R}^{d_{\sigma}}$ are learnable parameters, and softmax is given by
\begin{align}
    \operatorname{softmax}(\bm{x})_i  =\frac{e^{x_i}}{\sum_k e^{x_k}}.
\end{align}
By applying a softmax activation, our conditional probabilities become normalized, and by extension our wave function becomes normalized~\cite{RNN-Wavefunc}. Furthermore, the conditional phases are given by
\begin{align}
    \phi_i = \pi\textrm{Softsign}(U_{\phi}\bm{h}_i + \bm{c}_{\phi})\cdot \bm{\sigma}_i,
\end{align}
where the Softsign function is given by
\begin{align}
    \operatorname{Softsign}(\bm{x})_i = \frac{x_i}{1+|x_i|}\in(-1,1).
\end{align}
For the simplest RNN cell, known as a vanilla RNN cell~\cite{Vanilla}, the hidden state is updated as
\begin{equation}
    \bm{h}_i = f(W[\bm{h}_{i-1};\bm{\sigma}_{i-1}]+\bm{b}) \mh{,
} \end{equation}
where $W\in \mathbb{R}^{d_h\times(d_h+d_{\sigma})}$ and $\bm{b}\in \mathbb{R}^{d_h}$ are all learnable parameters, $f$ is a non-linear activation function, and $[\cdot;\cdot]$ represents vector concatenation. We initialize the hidden state $\bm{h}_0$ and the first input $\bm{\sigma}_0$ as zero vectors. The same parameters are shared across every spin in the lattice, making RNNs ideal for transfer learning across different sizes~\cite{roth2020iterativeretrainingquantumspin,hibatallah2022supplementingrecurrentneuralnetwork,Schuyler, nh89-6jmf}.

For simplicity, we present the vanilla RNN cell implementation, and we refer to App.~\ref{app:GRU} for a detailed implementation of the Gated Recurrent Unit (GRU) cell used in this work, which can better capture long-range dependencies in our lattice compared to vanilla RNNs~\cite{hibatallah2022supplementingrecurrentneuralnetwork}. Dilated recurrent connections have also been proposed to inject a long-range inductive bias into RNN wave functions and to help reproduce power-law correlations~\cite{DilatedRNN}.

\subsection{Two-dimensional RNNs}
RNN wave functions can be extended to two-dimensional systems~\cite{RNN-Wavefunc} as depicted in Fig.~\ref{fig:RNNWF}(b). In this setup, each RNN cell receives inputs from the horizontal and vertical neighbors. To maintain the autoregressive property, we define a one-dimensional path $\pi(k)$, with $k\in[0,N_xN_y-1]$, that goes through all the spins in our lattice. We choose a zigzag path
\begin{equation}
    \pi(k) = (i_k , j_k),\\
\end{equation}
where $j_k = \lfloor k/N_x\rfloor$ and
\begin{equation}
         i_k=\begin{cases}
        k \mod{N_x} & \textrm{if $j_k$ is even}   \\
        (N_x-1) - (k \mod{N_x})& \textrm{if $j_k$ is odd}
     \end{cases}.
\end{equation}
The conditional probabilities are given by
\begin{equation}
    P_{\boldsymbol{\theta}}(\boldsymbol{\sigma}_{\pi(k)}|\boldsymbol{\sigma}_{\pi(<k)}) = \textrm{softmax}(U\bm{h}_{i,j}+\bm{c})\cdot \boldsymbol{\sigma}_{i,j},
\end{equation}
with $\pi(k) = (i,j)$. $\pi(<k)$ corresponds to all spins generated before $\bm{\sigma}_{i,j}$ with respect to the zigzag path as depicted in Fig.~\ref{fig:RNNWF}(b).  Note that the bias introduced by the zigzag sampling path can be mitigated by applying symmetries, which results in an improved accuracy in a variational calculation~\cite{RNN-Wavefunc,hibatallah2022supplementingrecurrentneuralnetwork}.
The conditional phases for the two-dimensional RNN can be defined similarly to the one-dimensional case;
\begin{align}
    \phi_{i,j} = \pi\textrm{Softsign}(U_{\phi} \bm{h}_{i,j} + \bm{c}_{\phi})\cdot \bm{\sigma}_{i,j}.
\end{align}
In its simplest version, the 2D RNN wave function's hidden states are given by
\begin{align}
    \bm{h}_{ij}  = f(W[\bm{\sigma}_{i-(-1)^j,j};\bm{\sigma}_{i,j-1};\bm{h}_{i-(-1)^j,j};\bm{h}_{i,j-1}]+\mathbf{b}) .
\end{align}
Here $f$ is a non-linear activation function, and $W\in\mathbb{R}^{d_h\times 2(d_\sigma+d_h)}$ and $\bm{b}\in \mathbb{R}^{d_h}$ are learnable parameters. Boundary spins and hidden states are initialized with the zero vector. In this work, we use a GRU variant of this two-dimensional architecture, which we describe in detail in App.~\ref{app:GRU}.

\subsection{Optimization}
To optimize the parameters $\bm{\theta}$ of our normalized RNN wave function, we apply the variational principle and minimize the energy expectation value $\braket{\hat{H}}$:
\begin{equation}
     \min_\theta\braket{\Psi_{\bm{\theta}}|\hat{H}|\Psi_{\bm{\theta}}} \gtrsim E_{GS},
\end{equation}
where $E_{GS}$ is the true ground state energy for the Hamiltonian $\hat{H}$. We define the local energy $\epsilon(\bm\sigma)$ of our wave function as
\begin{align}
\epsilon(\bm\sigma) = \sum_{\bm{\sigma}'}
\frac{H_{\bm{\sigma}\bm{\sigma}'}\Psi_{\bm{\theta}}(\bm{\sigma}')}{\Psi_{\bm{\theta}}(\bm{\sigma})},
\end{align} where $H_{\bm{\sigma}\bm{\sigma'}}=\braket{\bm{\sigma}|\hat{H}|\bm{\sigma'}}$. We can rewrite the energy expectation value as an average over the probability distribution $P_{\bm{\theta}}(\bm{\sigma})=|\Psi_{\bm{\theta}}(\bm{\sigma})|^2$~\cite{Becca_Sorella_2017}
\begin{align}
    \braket{\hat{H}} =& \sum_{\bm{\sigma}}|\Psi_{\bm{\theta}}(\bm{\sigma})|^2\sum_{\bm{\sigma}'}\frac{H_{\bm{\sigma}\bm{\sigma}'}\Psi_{\bm{\theta}}(\bm{\sigma}')}{\Psi_{\bm{\theta}}(\bm{\sigma})}\\
    =&\sum_{\bm{\sigma}}P_{\bm{\theta}}(\bm{\sigma})\epsilon(\bm{\bm{\sigma}}).
\end{align}
The summation is over all possible spin configurations, whose number grows exponentially with the number of spins. Notably, our RNN wave function has the autoregressive property, which allows us to perfectly sample a list of configurations $\bm{\sigma}^{(1)}\ldots\bm{\sigma}^{(N_s)}$ efficiently~\cite{RNN-Wavefunc}. As a result, our energy expectation value can be estimated as
\begin{align}
    E_{\bm \theta} =\frac{1}{N_s}\sum_{i=1}^{N_s}\epsilon(\bm{\sigma}^{(i)}).
    \label{eq:energy}
\end{align}
The parameters of $\Psi_{\bm{\theta}}$ can be updated as $\bm{\theta}\rightarrow\bm{\theta}-\delta\bm{\theta}$ using stochastic gradient descent~(SGD), with $\delta\bm{\theta}$ being the direction of steepest descent
\begin{align}
    \delta\bm{\theta} &= \eta\partial_{\bm\theta}E_{\bm \theta},\\
    &\approx 2\eta\frac{1}{N_s}\sum_{i=1}^{N_s}\operatorname{Re}\Bigr[\partial_{\bm\theta}\log\Psi(\bm{\sigma}^{(i)})\epsilon(\bm{\sigma}^{(i)})^*\Bigl].
\end{align}
Here $\eta$ is the step-size for our optimization. We define $\mathcal O_{ks} \equiv \frac{1}{\sqrt{N_s}}\partial_{\bm\theta_{k}}\log\Psi(\bm\sigma^{(s)})$. We can further reduce the noise in the gradients by centering $\epsilon$, with $\bar\epsilon_i = \frac{1}{\sqrt N_s}(\epsilon(\bm\sigma^{(i)}) - E)^{*}$~\cite{RNN-Wavefunc}.
\begin{align}
\label{eq:SGD_update}
&\delta\bm\theta=2\eta\operatorname{Re}\Bigr[\mathcal{O}\bar{\bm{\epsilon}}\Bigl].
\end{align}
Both the local energy $\bm{\epsilon}\in\mathbb C^{N_s}$ and the log-derivative matrix $\mathcal O\in\mathbb C^{N_p\times N_s}$ are estimated by sampling the variational wave function $N_s$ times in parallel at every training iteration.

Several variants of SGD have been developed to help accelerate convergence. Momentum can add information from previous gradient steps, making optimization more robust against noise \cite{POLYAK19641}. Several sophisticated momentum-inspired versions of SGD have been developed, such as Adam~\cite{AdamOriginal}, which is among the state of the art in first-order optimization. Although first-order methods can perform well, they may make slow progress through flat regions of the loss landscape~\cite{amari_information_2016}. These flat plateaus are ubiquitous in neural networks, as they arise from their highly degenerate parameterization~\cite{berner2023degenerateparametrizationneuralnetworks}. To overcome these issues, we explore Natural Gradient Descent as a viable optimization alternative.

\subsection{Natural Gradients}
Natural Gradient, otherwise known as SR in the VMC community~\cite{Becca_Sorella_2017}, is a curvature-based method to optimize a variational ansatz. Its main advantage is the ability to escape flat plateaus in the loss landscape in finite time \cite{amari_information_2016}. For statistical models, the Fisher information matrix is the Hessian of the Kullback--Leibler
divergence, i.e.\ the Bregman divergence generated by the negative entropy, and the natural gradient is the ordinary gradient preconditioned by its inverse. In the context of VMC, the natural gradient can be viewed in two equivalent ways: as approximating imaginary-time evolution~\cite{NQS-rev}, or using a geometric point of view by minimizing the Fubini-Study metric during updates~\cite{Becca_Sorella_2017}. We proceed by deriving Natural Gradients from the imaginary-time evolution view. Given a trial wave function $\Psi_{\bm{\theta}}$, the natural gradient update approximates the imaginary-time evolution
\begin{equation}
   \ket{\Psi_{\bm{\theta}-\delta\bm{\theta}}} = \exp{(-\hat{H}\delta \tau)}\ket{\Psi_{\bm{\theta}}},
\end{equation}
which exponentially suppresses all excited states relative to the ground state. We Taylor expand both the wave function and the imaginary-time evolution to first order
\begin{align}
    |\Psi_{\bm{\theta}} \rangle -\delta\bm{\theta}\cdot \nabla_{\bm{\theta}}|\Psi_{\bm{\theta}} \rangle = |\Psi_{\bm{\theta}}\rangle -\delta\tau \hat{H} |\Psi_{\bm{\theta}}\rangle.
\end{align}
Normalized imaginary-time evolution gives \[ |\Psi_{\bm{\theta}}\rangle \longrightarrow |\Psi_{\bm{\theta}}\rangle -\delta\tau(\hat H-E_{\bm{\theta}}) |\Psi_{\bm{\theta}}\rangle . \] Matching this tangent vector to the normalized version of the variational displacement $-\delta\bm{\theta}\cdot\nabla_{\bm{\theta}} |\Psi_{\bm{\theta}}\rangle$ in a Fubini--Study metric minimization task leads to 
\begin{align}
   \delta\bm{\theta} =  \operatorname*{argmin}_{\delta\bm{\theta}\in\mathbb R^{N_p}} \|\bar{\mathcal{O}} ^{\mathsf T} \delta\bm{\theta}-2\eta\bar{\bm{\epsilon}}^{*}\|_2^2,
\end{align}
where  $\bar{\mathcal{O}}_{ks}=\mathcal O_{ks}-\frac{1}{N_s}\sum_{s'}\mathcal O_{ks'}$ and $\eta=\delta\tau/2$ serves the role of the learning rate~\cite{Nomura_2023}. Since this system is under-determined, we add Tikhonov regularization,
\begin{equation}
\operatorname*{argmin}_{\delta\bm\theta\in\mathbb R^{N_p}}\|\bar{\mathcal{O}}^{\mathsf T}\delta\bm\theta-2\eta\bar{\bm{\epsilon}}^{*}\|_2^2 + \lambda\|\delta \bm\theta\|_2^2.\label{eq:lstsq}
\end{equation}
For a complex-valued ansatz, this least-squares  problem can be written exactly as a real one by defining
\begin{align}
\label{eq:real-valued}
    \bar{\mathcal{O}}' &= \Bigl(\,\operatorname{Re}\bar{\mathcal O}\;\;\operatorname{Im}\bar{\mathcal O}\,\Bigr), \\
    \bm{\bar\epsilon}' &= 2\begin{pmatrix}
       \operatorname{Re}\bm{\bar\epsilon}\\
       -\operatorname{Im}\bm{\bar\epsilon}
    \end{pmatrix}.
\end{align}
The normal equations are therefore
\begin{equation}
    \delta\bm\theta = \eta\big( S +\lambda \mathbb{I} \big) ^{-1}\bar{\mathcal{O}} ' \bar{\bm{\epsilon}} ',
    \qquad
    S\equiv\bar{\mathcal O}'\bar{\mathcal O^{'\mathsf T}}=\operatorname{Re}\big(\bar{\mathcal O}\bar{\mathcal O}^{\dagger}\big) ,
    \label{eq:SR}
\end{equation}
where $\lambda$ is a Tikhonov parameter to mitigate ill-conditioning, and can be interpreted as an inverse trust region radius. When $\lambda\gg \|S\|$ the updates approach the SGD updates in Eq.~\eqref{eq:SGD_update} with centered log-derivatives. The complex matrix $Q \equiv \bar{\mathcal O}\bar{\mathcal O}^{\dagger}$ is the quantum geometric tensor (QGT), while for real parameters the Fubini--Study metric entering SR is its real part $S$. For a positive real ansatz, the imaginary blocks vanish, and the two definitions coincide. SR is therefore a natural-gradient method: parameter updates are measured by the distance between quantum states rather than by the Euclidean distance between parameters.

Using the identity $(AA^{\mathsf T}+\lambda I)^{-1}A=A(A^{\mathsf T}A+\lambda I)^{-1}$ gives the minSR update
\begin{equation}
    \delta\bm\theta = \eta\bar{\mathcal{O}} ' \big( \mathcal T +\lambda\mathbb{I} \big) ^{-1}\bar{\bm{\epsilon}} ',
    \qquad
    \mathcal T\equiv\bar{\mathcal O}^{'\mathsf T}\bar{\mathcal O}',
    \label{eq:minsrR}
 \end{equation}
which is the real-parameter form of the minSR identity used in our implementation. For a complex-valued ansatz, $\mathcal T$ is $2N_s\times2N_s$; for a positive real ansatz, the zero imaginary block can be omitted, and the reduced matrix is $N_s\times N_s$. In either case, minSR replaces the parameter-space inversion by a sample-space inversion, reducing the leading inversion cost from cubic in $N_p$ to $\mathcal{O}(N_s^2N_p + N_s^3)$ up to constant factors. Geometrically, minSR is the natural-gradient update restricted to the sample-spanned tangent subspace. In this study, we work in the $N_s\ll N_p$ regime, and we find the Jacobian calculation to be the most memory-consuming step, with leading time complexity $O(d_h^2NN_s)$. We present the minSR algorithm in Algorithm~\ref{alg:minsr} and apply momentum as described in App.~\ref{app:momentum}.

\begin{algorithm*}
\caption{Minimum-step stochastic reconfiguration (minSR) for RNN wave functions}
\label{alg:minsr}
\begin{algorithmic}[1]
\Require Initial parameters $\bm{\theta}_0$; sample count $N_s$; learning rate $\eta$; regularization $\lambda$; momentum coefficient $\mu$; number of iterations $T$
\Ensure Optimized parameters $\bm{\theta}_T$ approximating the ground state
\State Initialize the momentum buffer $\bm{m}_0 \gets \bm{0}$
\For{$t = 0$ \textbf{to} $T-1$}
    \State Draw $N_s$ configurations $\{\bm{\sigma}^{(i)}\}_{i=1}^{N_s}$ autoregressively from $|\Psi_{\bm{\theta}_t}(\bm{\sigma})|^2$
    \State Evaluate the local energies $\epsilon(\bm{\sigma}^{(i)})$ and the energy estimate $E \gets \tfrac{1}{N_s}\sum_{i} \epsilon(\bm{\sigma}^{(i)})$
    \State Assemble the log-derivative matrix $\mathcal{O}_{ki} \gets \tfrac{1}{\sqrt{N_s}}\,\partial_{\theta_k}\!\log\Psi_{\bm{\theta}_t}(\bm{\sigma}^{(i)})$
    \State Center over samples: $\bar{\mathcal{O}}_{ki} \gets \mathcal{O}_{ki} - \tfrac{1}{N_s}\textstyle\sum_{j}\mathcal{O}_{kj}$ and $\bar{\epsilon}_i \gets \tfrac{1}{\sqrt{N_s}}\big(\epsilon(\bm{\sigma}^{(i)}) - E\big)^*$
    \State Form the real representations $\bar{\mathcal{O}}'$, $\bar{\bm{\epsilon}}'$ and the minSR matrix $\mathcal{T} \gets \bar{\mathcal{O}}'^{\top}\bar{\mathcal{O}}'$ \Comment{$2N_s \times 2N_s$}
    \State Compute the update $\delta\bm{\theta}_t \gets \eta\,\bar{\mathcal{O}}'(\mathcal{T} + \lambda\mathbb{I})^{-1}\bar{\bm{\epsilon}}'$ \Comment{Eq.~\eqref{eq:minsrR}}
    \State Apply momentum $\bm{m}_{t+1} \gets \mu\,\bm{m}_t + (1-\mu)\,\delta\bm{\theta}_t$ \Comment{App.~\ref{app:momentum}}
    \State Update the parameters $\bm{\theta}_{t+1} \gets \bm{\theta}_t - \bm{m}_{t+1}$
\EndFor
\State \Return $\bm{\theta}_T$
\end{algorithmic}
\end{algorithm*}

\subsection{RNN Fisher Spectrum}
\label{subsec:rnn_fisher_spectrum}
Natural-gradient optimization of RNN wave functions can be challenging due to two main considerations. First, different RNN parameter values can describe the same wave function. Directions that leave the wave function itself unchanged have zero Fisher norm and therefore appear as exact zero modes of the Fisher matrix. Second, the RNN recurrence can amplify small parameter changes, which can produce very large Fisher eigenvalues. Together, these two effects can spread the Fisher spectrum over a wide range and make the SR matrix severely ill-conditioned. 

Fisher singularities from redundant parameters are not unique to RNN wave functions: in transformer manifolds, low Fisher rank similarly identifies locally inactive parameter directions and motivates intrinsic-dimension-aware adaptation~\cite{GeLoRA}. A related interaction between architecture and representation geometry appears in message-passing networks, where local graph topology can determine whether repeated propagation yields expressive representations or oversmoothing~\cite{GNNTopology}.

We now state both mechanisms as propositions. The proofs can be found in App.~\ref{app:illcond_detailed} for a simplified model: a real, positive, autoregressive RNN wave function over $N$ sites, whose conditional distributions are produced by a linear RNN cell with a softmax head. At site $i$, the hidden state $\mathbf h_i\in\mathbb R^{d_h}$ and the logits $\mathbf z_i\in\mathbb R^{d_{\sigma}}$ obey $\mathbf h_i=W\mathbf h_{i-1}+V\boldsymbol{\sigma}_i+\mathbf b$ and $\mathbf z_i=U\mathbf h_{i-1}+\mathbf c$. Here $\boldsymbol{\sigma}_i\in\mathbb R^{d_{\sigma}}$ is the one-hot encoding of the local spin, $W\in\mathbb R^{d_h\times d_h}$ is the recurrent matrix, $V\in\mathbb R^{d_h\times d_{\sigma}}$ is the input matrix, $U\in\mathbb R^{d_{\sigma}\times d_h}$ is the readout matrix, and $\mathbf b\in\mathbb R^{d_h}$, $\mathbf c\in\mathbb R^{d_{\sigma}}$ are biases. The initial hidden state is fixed at $\mathbf h_0=\mathbf 0$. The conditional probability vector at site $i$ is $\mathbf p_i=\operatorname{softmax}(\mathbf z_i)$. These exact $GL(d_h)$ gauge statements are proved for this linear model; a gated nonlinear GRU does not, in general, possess the same full hidden-basis symmetry exactly.

\begin{proposition}[Exact gauge degeneracy]
\label{prop:gauge_zero_modes}
Consider the previous linear RNN wave function, with parameters $\boldsymbol{\theta}=(W,V,\mathbf b,U,\mathbf c)$ and fixed initial hidden state $\mathbf h_0=\mathbf 0$. Fix any matrix $X\in\mathbb R^{d_h\times d_h}$ and define the parameter velocities
\begin{equation}
\dot{\boldsymbol{\theta}}_X
=
\big([X,W],\,XV,\,X\mathbf b,\,-UX,\,\mathbf 0\big),
\label{eq:prop_gauge_direction}
\end{equation}
whose five entries are the increments of $W,V,\mathbf b,U,\mathbf c$, and where $[X,W]=XW-WX$ is the commutator, or Lie bracket, of $X$ and $W$. Then $d\boldsymbol{\theta}_X = \dot{\boldsymbol{\theta}}_X \text{d}t$, given a real-valued parameter $t$, is a symmetry direction of the model, with two consequences:
\begin{enumerate}
\item[(i)] It leaves the wave function unchanged. The directional derivative
of $\log\Psi_{\boldsymbol{\theta}}(\boldsymbol{\sigma})$ along
$d\boldsymbol{\theta}_X$ is zero for every configuration $\boldsymbol{\sigma}$,
not just on average.
\item[(ii)] It is therefore a null direction of the SR matrix
$S=\bar{\mathcal{O}}\bar{\mathcal{O}}^{\dagger}$. It lies in $\ker S$ both for
the finite-sample matrix built from any set of configurations and for its population limit ($N_s\to\infty$) under sampling from
$|\Psi_{\boldsymbol{\theta}}|^2$.
\end{enumerate}
The directions $d\boldsymbol{\theta}_X$ are the infinitesimal generators of the gauge group $GL(d_h)$ of invertible real $d_h\times d_h$ matrices, which act on the parameters through Eq.~\eqref{eq:gauge_map_H}. Equivalently, the tangent space to the gauge orbit at $\boldsymbol{\theta}$ is contained in $\ker S$.
\end{proposition}
The proof is given in App.~\ref{app:illcond_detailed}.

\begin{proposition}[Recurrent amplification]
\label{prop:spectral_divergence}
Consider the model of Proposition~\ref{prop:gauge_zero_modes}, now with a diagonal recurrent matrix $W=\operatorname{diag}(r_1,\dots,r_{d_h})$, so that
the hidden modes evolve independently with multipliers $r_1,\dots,r_{d_h}$. Let $S$ be the population SR matrix, i.e.\ the $N_s\to\infty$ limit of $\bar{\mathcal{O}}\bar{\mathcal{O}}^{\dagger}$ under configurations drawn i.i.d.\ from $|\Psi_{\boldsymbol{\theta}}|^2$. Since $S$ is Hermitian and positive semidefinite, its eigenvalues are real and nonnegative; we write $\lambda_{\max}(S)$ for the largest one.

Fix a hidden mode $a\in\{1,\dots,d_h\}$, and let $\mathbf u_a=U\mathbf e_a$ be its readout vector, with $\mathbf e_a$ the $a$-th standard basis vector of $\mathbb R^{d_h}$. Define the nonnegative quantity
\begin{equation}
\chi_{a,i}
=
\big\langle \mathbf u_a^{\mathsf T} C_i\,\mathbf u_a \big\rangle
\;\ge\; 0,
\label{eq:prop_chi_def}
\end{equation}
which measures the conditional softmax covariance
$C_i=\operatorname{diag}(\mathbf p_i)-\mathbf p_i\mathbf p_i^{\mathsf T}$ at
site $i$ along the readout direction $\mathbf u_a$, averaged over histories
drawn from $|\Psi_{\boldsymbol{\theta}}|^2$. Assume \emph{uniform
non-saturation}: there is a constant $\chi_0>0$ with $\chi_{a,i}\ge\chi_0$ for
all sites $i\le N$ and all system sizes $N$.

Then $\lambda_{\max}(S)$ diverges with $N$, at a rate set by the multiplier
$r_a$:
\begin{enumerate}
\item[(i)] \emph{Marginal mode ($r_a=1$).}
\begin{equation}
\lambda_{\max}(S)
\;\ge\;
\frac{\chi_0}{4}\,
\frac{(N-1)N(2N-1)}{6},
\label{eq:prop_marginal_bound}
\end{equation}
which grows as $\chi_0 N^{3}/12$: polynomial in the system size $N$.
\item[(ii)] \emph{Expanding mode ($|r_a|>1$).}
\begin{equation}
\lambda_{\max}(S)
\;\ge\;
\frac{1}{4}\,
\big|m_{N-1}(r_a)\big|^{2}\,
\chi_{a,N},
\label{eq:prop_expanding_bound}
\end{equation}
with $m_{N-1}(r_a)=(1-r_a^{\,N-1})/(1-r_a)$. The right-hand side grows as
$|r_a|^{2(N-1)}\chi_{a,N}/(4|r_a-1|^{2})$: exponential in $N$.
\end{enumerate}
\end{proposition}
The proof is given in App.~\ref{app:illcond_detailed}.

In this work, we demonstrate that careful preconditioning can mitigate the challenges posed by the ill-conditioned RNN Fisher spectrum, resulting in stable natural-gradient optimization, as discussed in the results section.

\section{Results}
\label{sec:results}
We perform extensive hyperparameter searches for both Adam and minSR. We describe our hyperparameter experiments in App.~\ref{app:hyperparameter}. We initialize training using the same initial RNN parameters for both minSR and Adam.
\subsection{1D Transverse-Field Ising Model}

We benchmark minSR against Adam in finding the ground-state wave function on a variety of models, starting with the simplest case of a one-dimensional stoquastic Hamiltonian, namely the one-dimensional transverse-field Ising model (TFIM).
The Hamiltonian of the 1D TFIM is given by
\begin{equation}
    \hat{\mathcal{H}}_{\textrm{TFIM}} = -J\sum_{i=1}^{N-1}\hat{\sigma}_i^z\hat{\sigma}_{i+1}^z - h\sum _{i=1}^N\hat{\sigma}_i^x,
\end{equation}
where $J$ and $h$ are coupling constants, and $\hat{\sigma}^{x,y,z}$ are the Pauli matrices. We work with open boundary conditions (OBC) and choose $J=1=h$, which corresponds to the critical point~\cite{Sachdev_2011}. The stoquasticity of this Hamiltonian~\cite{bravyi2007complexitystoquasticlocalhamiltonian} allows us to represent the ground state with a pRNN wave function.

For a fair comparison, we evaluate both optimizers over the same training time, allowing Adam to take more steps than minSR. In particular, over the course of a $500$-second training session, we show that using minSR we can achieve several orders of magnitude improvement over the Adam optimizer. In Fig.~\ref{fig:TFIM}, we plot the performance of both minSR and Adam on the TFIM with $N = 200$ spins. We compare the relative error of both optimizers
\begin{equation}
    \label{eq:rele}
    \operatorname{Rel. Err.} =  \frac{E_{\textrm{RNN}}-E_{\textrm{ref}}}{|E_{\textrm{ref}}|} ,
\end{equation}
where $E_{\textrm{ref}}$ is the lowest reference energy available from the literature or the exact result from theory. We obtain this performance using a single fixed value of $\lambda$. The poor conditioning of the SR matrix comes from the structure of the RNN itself. App.~\ref{app:illcond_detailed} illustrates mechanisms by which recurrent architectures can generate ill-conditioning in a simplified model; this motivates Tikhonov regularization in the GRU experiments, since it makes the minSR preconditioner well conditioned. We further analyze the training instabilities associated with ill-conditioning in App.~\ref{app:instabilities}. Hyperparameters of our simulation are reported in App.~\ref{app:hyperparameter}.
\begin{figure}
 \centering
   \includegraphics[width=\columnwidth]{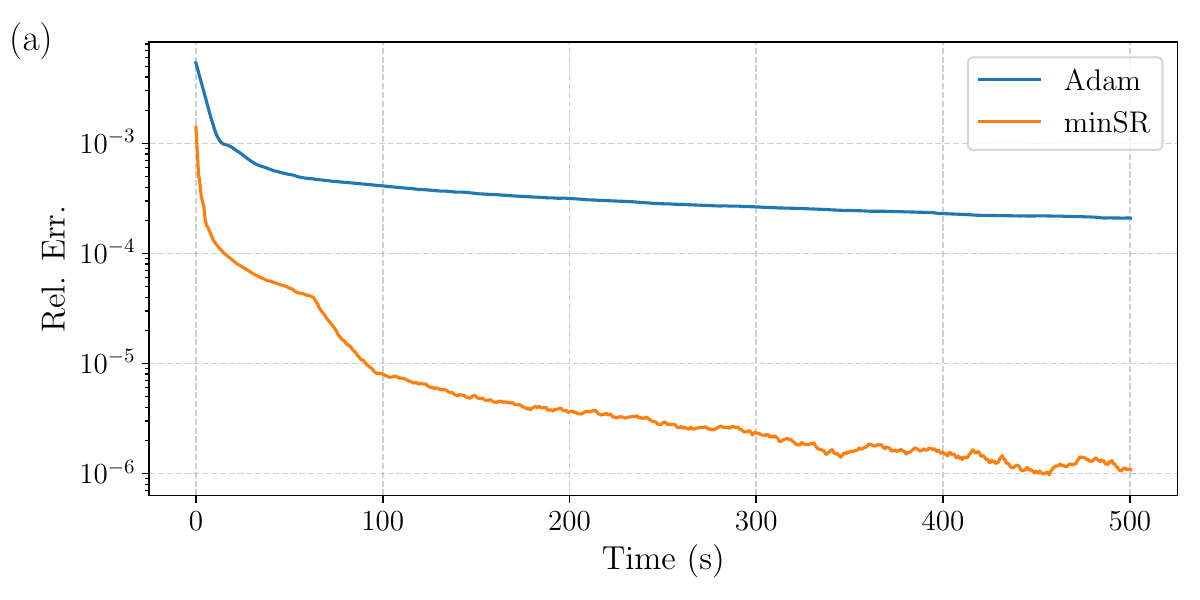}
   \includegraphics[width=\columnwidth]{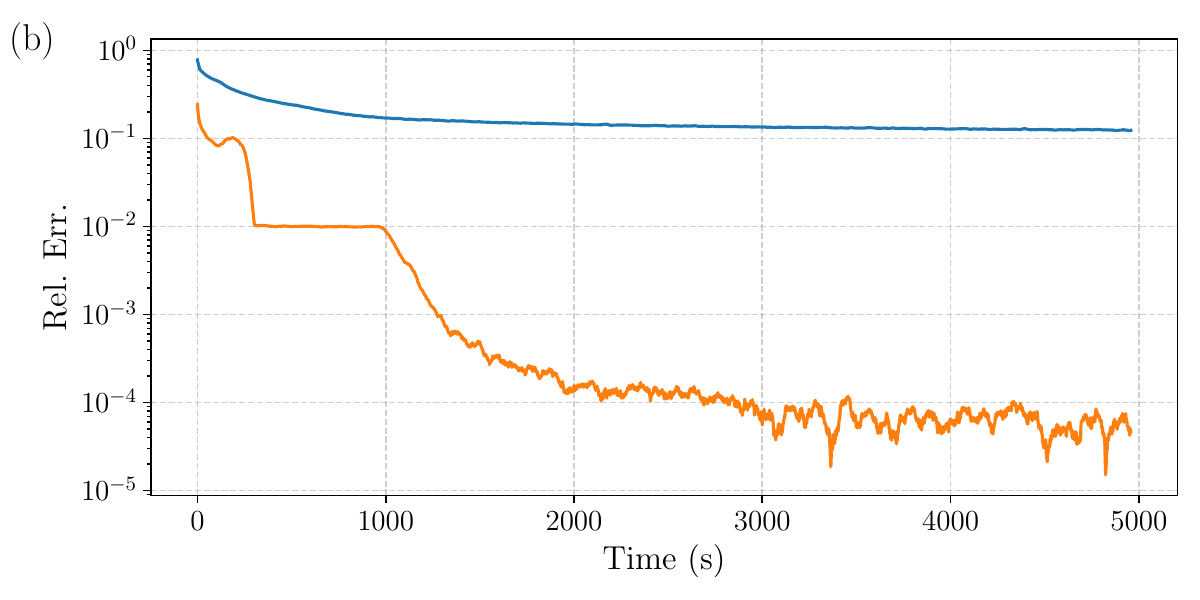}
 \caption{\textbf{One-dimensional benchmark comparison} between Adam and minSR optimizers on two testbeds a) the 1D TFIM and b) the Cluster state. a) We find the ground state energy of the TFIM Hamiltonian with open boundary conditions with 200 spins at the critical point $h=J=1$. The underlying model is a pRNN with $d_h = 32$ hidden units and $N_s=100$ Monte Carlo samples. The relative error with respect to the exact theoretical value~\cite{ising} is plotted against time in seconds. To minimize statistical noise, the relative errors are averaged over a $300$ training-step window. b) A comparison between Adam and minSR optimizers for the one-dimensional cluster state for $N=100$ qubits. The relative error with respect to the exact value of $E_{\mathrm{exact}}=-N$ is plotted against time in seconds. We use a moving-average window size of 20 to reduce statistical noise due to small sample size in plots. Hyperparameters are reported in Tab.~\ref{tab:combined_hyperparams} and final energies are reported in Tab.~\ref{tab:energies}.}
 \label{fig:cluster}
 \label{fig:TFIM}
\end{figure}

\subsection{1D Cluster State}
We now turn to the 1D Cluster state benchmark, which corresponds to a non-stoquastic Hamiltonian. The Cluster state is a key model for measurement-based quantum computation~\cite{cluster-preskill}. It is known to have long-range conditional mutual information, making it challenging to learn using NQS~\cite{cluster-preskill}. The Cluster state corresponds to the ground state of the following Hamiltonian:
\begin{align*}
 \hat{\mathcal{H}}_{\textrm{CS}} =& -\sum_{k=2}^{N-2}\hat{\sigma}^x_{k-1}\hat{\sigma}^z_k\hat{\sigma}^x_{k+1} -\hat{\sigma}^z_1\hat{\sigma}^x_2\\& -\hat{\sigma}^x_{N-1}\hat{\sigma}^x_{N}-\hat{\sigma}^x_{N-2}\hat{\sigma}^z_{N-1}\hat{\sigma}^z_N.
\end{align*}
Since this Hamiltonian is non-stoquastic, we use a cRNN wave function as our ansatz. Note that RNN wave functions~\cite{cluster-preskill} and Dilated RNN wave functions~\cite{DilatedRNN} have already been applied to this problem and only optimized using the Adam optimizer. In Ref.~\cite{adaptive}, RNN wave functions optimized with Adam achieve a relative error of about $10^{-2}$ with a system size $N = 64$. In Fig.~\ref{fig:cluster}(b), we achieve very low errors, far below those of the Adam optimizer. Using minSR on a larger system size of $N = 100$ spins, our final variational energy $E=-99.99(4)$ is very close to the theoretical value of $E=-100$. Here too, minSR is stabilized by a single constant Tikhonov shift $\lambda$, with no schedule; the values used are listed in App.~\ref{app:hyperparameter}. Overall, we demonstrate that RNN wave functions implemented with the traditional GRU cell, optimized with minSR, can numerically capture long-range conditional correlations of the Cluster state without the need for advanced RNN models such as Dilated RNNs~\cite{DilatedRNN}.

\subsection{Square Lattice $J_1-J_2$ Model}
We now move our benchmark comparison towards two-dimensional systems. Here we focus on the square lattice $J_1 - J_2$ model given by
\begin{align*}
    \hat{\mathcal{H}}_{J_1J_2} &= \frac{J_1}{4}\sum_{\langle ij\rangle}\bigl(\hat{\sigma}_i^x\hat{\sigma}_j^x + \hat{\sigma}_i^y\hat{\sigma}_j^y+\hat{\sigma}_i^z\hat{\sigma}_j^z\bigr)\\
    &+\frac{J_2}{4}\sum_{\langle\langle ij\rangle\rangle}\bigl(\hat{\sigma}_i^x\hat{\sigma}_j^x + \hat{\sigma}_i^y\hat{\sigma}_j^y+\hat{\sigma}_i^z\hat{\sigma}_j^z\bigr),
\end{align*}
where $\hat{\sigma}^{x,y,z}$ are Pauli operators and $J_1$, $J_2$ are coupling constants, and we set $J_1 = 1$. Also note that $\langle ij\rangle$ ($\langle\langle ij\rangle\rangle$) correspond to nearest (next-nearest) neighbours. The two main cases of interest are: $J_2=0$ corresponding to the antiferromagnetic Heisenberg model, and $J_2=0.5$ corresponding to a frustrated regime. This frustration point is often used to benchmark the strengths and weaknesses of NQS methods~\cite{Chen_2024, minSR-lianlg, PhysRevX.11.031034}.

We begin with the unfrustrated point $J_2 = 0$, where we use a 2D pRNN wave function as our ansatz. Even though the Hamiltonian is originally non-stoquastic, a simple Marshall sign rotation can be applied to make it stoquastic~\cite{Marshall1955, shamim2026graphtheoreticanalysisphaseoptimization}.
Similar to the relative error metric used for the one-dimensional benchmarks, the accuracy of an NQS can also be assessed by using a standard metric, which combines the energy and variance into a single dimensionless parameter known as the V-score~\cite{vscore}:
\begin{equation}
    \textbf{V-score} = \frac{N \cdot \operatorname{Var}[\bm{\epsilon}]}{(E-E_{\infty})^2}.\label{eq:vscore}
\end{equation}
Here $N$ is the number of spins and $E_{\infty}$ is the infinite-temperature energy, which is equal to zero for this model~\cite{vscore}. This metric eliminates the need for an accurate reference energy while providing the same scaling behavior as the relative error near convergence~\cite{vscore}. In Fig.~\ref{fig:J1J2}(a), we show that minSR achieves a lower V-score than Adam over the same time period, as well as lower energy and energy variance as reported in Tab.~\ref{tab:energies}. However, we note that the advantage of minSR is less pronounced relative to the previous one-dimensional benchmarks. In our hyperparameter tuning experiments, we choose an increasing value of the Tikhonov regularization parameter, $\lambda =3\times 10^{-4}\cdot\|\mathcal{T}\|^{2/3}$, where $\|\mathcal{T}\|$ generally starts with a small value at the start of training, then increases several orders of magnitude and then finally settles at a large value that is usually in the range $\|\mathcal{T}\|\sim10^{5-8}$. Using this power of $\frac{2}{3}$ was empirically found to give superior results for $d_h=100$ compared to a constant $\lambda$. This can be explained by the model preferring natural gradient updates, which correspond to a small $\lambda$, during the initial period of training and SGD, which corresponds to having a large value of $\lambda$, during the later stages of training.

\begin{figure}
 \centering
\includegraphics[width=0.9\columnwidth]{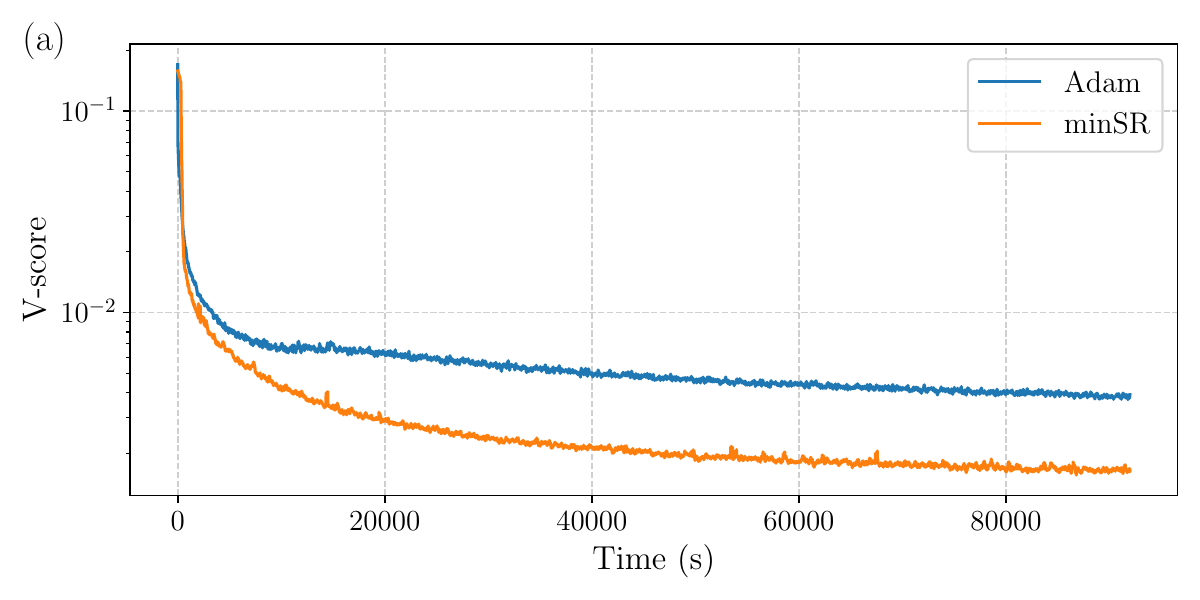}
\includegraphics[width=0.9\columnwidth]{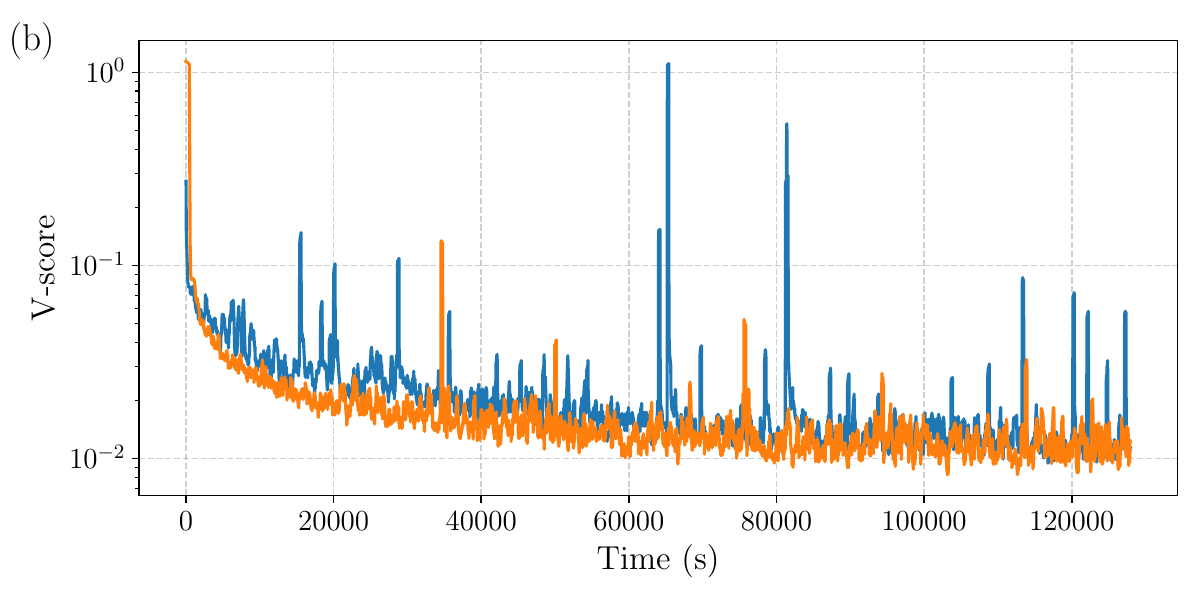}
 \caption{\textbf{2D $J_1 - J_2$ model benchmark.}  We compare the Adam and minSR optimizers by showing the evolution of the V-score~\eqref{eq:vscore} versus time in seconds at (a) the Heisenberg point $J_2 = 0$, and (b) the frustrated regime of $J_2=0.5$. We apply a moving average with a window size of $200$ iterations to minimize statistical noise. Our results show an enhancement of accuracy when applying minSR at $J_2 = 0$, but an accuracy comparable to that of Adam at $J_2 = 0.5$.}
 \label{fig:J1J2}
\end{figure}

In the frustrated regime ($J_2 = 0.5$), we apply a Marshall sign rule~\cite{Marshall1955, shamim2026graphtheoreticanalysisphaseoptimization}
\begin{equation}
    \Psi_{\bm{\theta}}(\bm{\sigma}) = (-1)^{\sigma_A}\tilde{\Psi}_{\bm{\theta}}(\bm{\sigma}),
\end{equation}
where $\sigma_A = \sum_{i\in A}\sigma_i$ and the sum over $A$ runs over every other spin on the square lattice. Here the amplitudes $\tilde{\Psi}_{\bm{\theta}}(\bm{\sigma})$ are modeled by a two-dimensional cRNN wave function. This rule corresponds to the exact sign structure for $J_2=0$ and can be used to improve NQS performance for relatively small $J_2$~\cite{hibatallah2022supplementingrecurrentneuralnetwork,RNN-Wavefunc, Choo_2019}.

For minSR regularization, we use the trust region learning rate $\eta_*$ from Eq.~\eqref{eq:trust-region-lr}, which adapts the step-size in both parameter space and in distribution space to be within a trust region (see App.~\ref{app:trust_region}). In Fig.~\ref{fig:J1J2}(b), we observe that minSR and Adam have comparable V-scores, though minSR converges slightly faster initially. In particular, the results show significant instability, suffering from strong fluctuations in the V-score throughout training. We find that a constant value of $\lambda = 10^{-1}$ provides the best results.

In both unfrustrated and frustrated models, we find training to be very sensitive to the damping scheme. When applying the Levenberg-Marquardt (LM) scheme (described in App.~\ref{app:levenberg_marquardt}), we find that $\lambda$ tends to diverge as a result of the poor conditioning of the landscape. Because larger $\lambda$ corresponds to stronger damping, the observed growth of $\lambda$ indicates that the local quadratic model is trusted only over increasingly small parameter displacements, at least with a small number of samples. This can substantially reduce the benefit of the minSR preconditioner. Certain parameters become very sensitive to even small step sizes, making a uniform trust region for all parameter directions a poor choice.

Naive schemes such as using $\lambda \times \textrm{diag}(\mathcal{T})$ as a preconditioner instead of $\lambda \mathbb{I}$ can result in very small updates and poor final energies. Structural Damping is a more sophisticated damping scheme that has been successful in second-order RNN optimization~\cite{hf-rnn} by penalizing sudden changes in hidden state trajectories. However, it is too computationally expensive to implement for larger-scale RNNs.

We conjecture that the 2D RNN wave function architecture, inspired by the multi-dimensional RNN (MDRNN) introduced in Ref.~\cite{mdrnn}, may be the cause of the large discrepancy between the performance of natural gradients in 1D and 2D. In one dimension, each cell passes its hidden state to a single neighbor. In two dimensions, each cell receives hidden states from two neighbors. As a result, the number of paths along which gradients travel grows combinatorially with system size, leading to training instabilities~\cite{GridLSTM}. Additionally, in Ref.~\cite{leakyMDRNN}, it is shown that Long-Short Term Memory (LSTM) based MDRNNs can face exploding gradients, similarly to 1D Vanilla RNNs. More stable multi-dimensional RNNs have been proposed, such as the GridLSTM~\cite{GridLSTM} or replacing the GRU cell with a Leaky Lowpass (LeakyLP) cell~\cite{leakyMDRNN}. Designing a more stable 2D RNN architecture may reduce ill-conditioning of the loss landscape, enabling a more efficient natural gradient optimization, without the need for more sophisticated damping schemes. 

\section{Loss landscape}
\label{sec:loss_landscape}
In the previous section, we observed that minSR was significantly more successful in 1D and hypothesized that the 2D loss landscapes may be ill-conditioned. In this section, we visualize these landscapes along certain parameter directions. In Ref.~\cite{VizualizingRNNs}, plotting along two random parameter directions, with appropriate parameter normalization, has been shown to provide useful insights into the geometry of the loss landscape.
\begin{figure}
    \centering
    \includegraphics[width=
    \columnwidth]{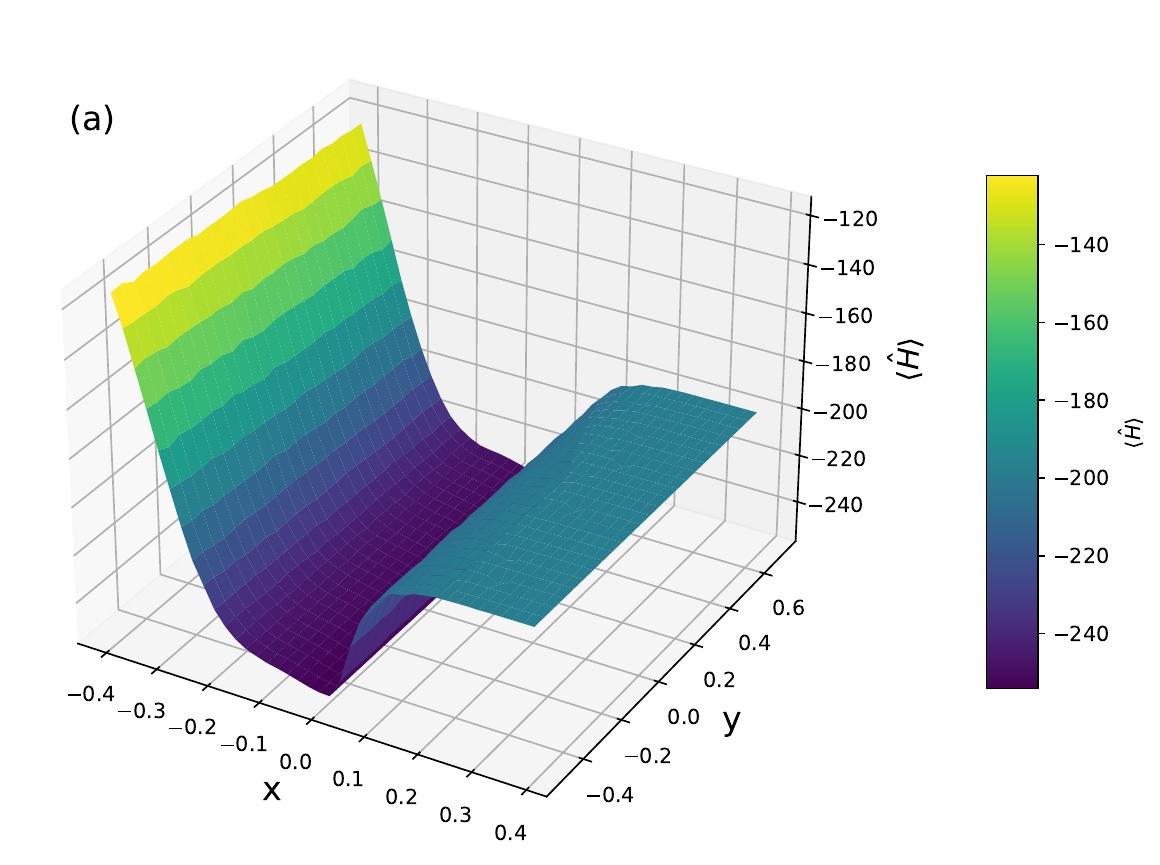}
\includegraphics[width=\columnwidth]{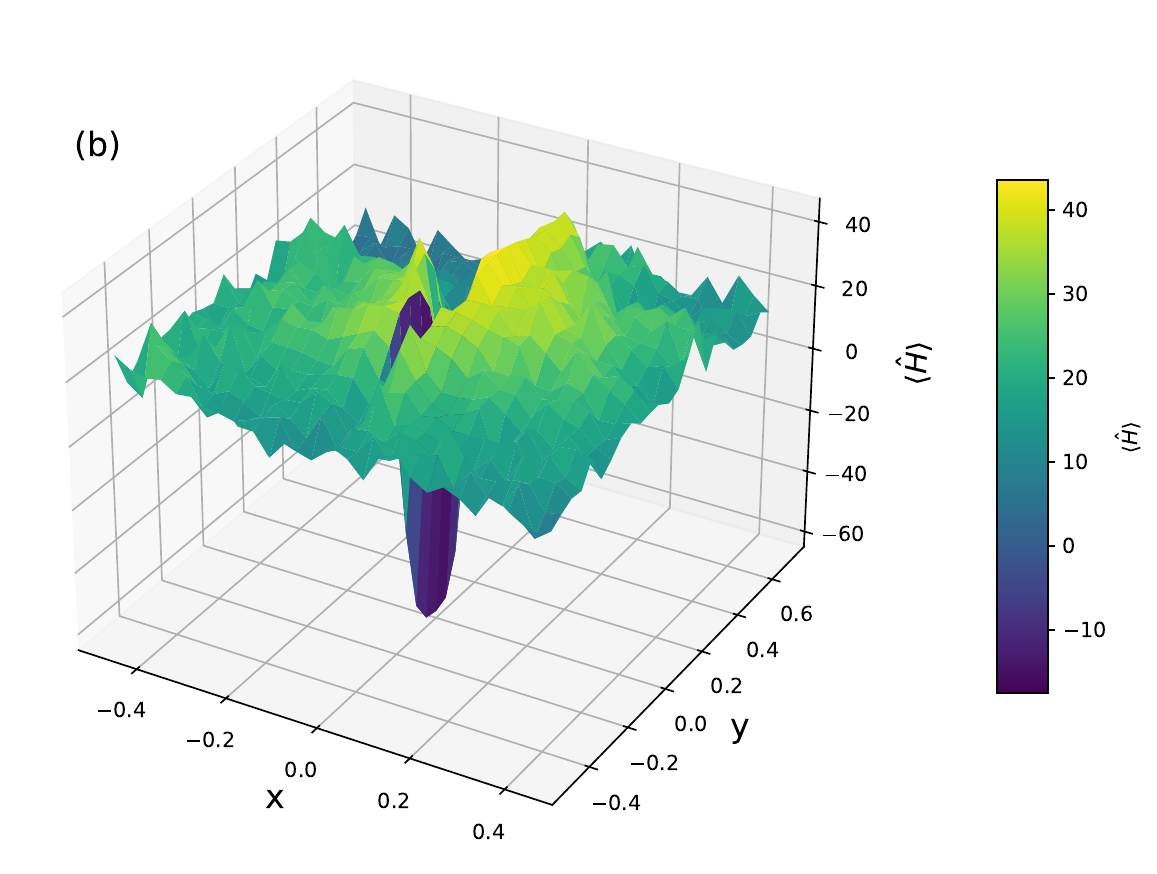}
    \caption{\textbf{Loss Landscape in 1D and 2D.} We plot the loss landscapes of the a) 1D TFIM and b) 2D Heisenberg model around their optimal parameters, from section \ref{sec:results}, on a $60\times 60$ grid. Each energy is calculated using $10^{3}$ samples, which we found sufficient to discern important features of the loss landscape.  }
    \label{fig:loss-landscape}
\end{figure}
We plot the energy $E(\bm\theta_* + x\bm v_{\mathrm{max}} +y\bm v_{\min})$, where $\bm\theta_*$ denotes the optimized variational parameters obtained after training. $\bm v_{\max}$ and $\bm v_{\min} $ are parameter-space directions corresponding to finite-sample approximations to the eigenvectors of $S$ with eigenvalues approximating $\lambda_{\max}$ and $\lambda_{\min}$, respectively. A similar choice of $v_{\max}$ and $v_{\min}$ has been proposed to accurately capture saddle points~\cite{bottcher2024visualizing}, and it may allow us to focus on ill-conditioned portions of the loss landscape. To ensure our comparison of the different landscapes is not biased by the size of the weights and biases, we ensure $\|\bm v_{\max} \|=\|\bm v_{\min} \|=\|\bm\theta_*\|$ for each of our models \cite{VizualizingRNNs}.

In Fig.~\ref{fig:loss-landscape}, we observe that the loss landscape of the 2D Heisenberg model is very ill-conditioned in comparison to the smooth loss landscape in the 1D TFIM model. Its convex region in the center is 10 times smaller in the $x$ direction. This indicates that the variational energy becomes more sensitive to small changes in certain parameters, reminiscent of other models that face exploding gradients. This observation suggests that using a uniform $\lambda$ across all parameter directions may be a suboptimal choice, and that parameter-dependent damping should be introduced.

These plots support the hypothesis that loss-landscape ill-conditioning contributes to minSR's reduced advantage in two spatial dimensions. More stable architectures or damping schemes may therefore improve optimization in two dimensions.
\section{Scaling study}
\label{sec:samples}

\begin{figure}

    \includegraphics[width=\columnwidth]{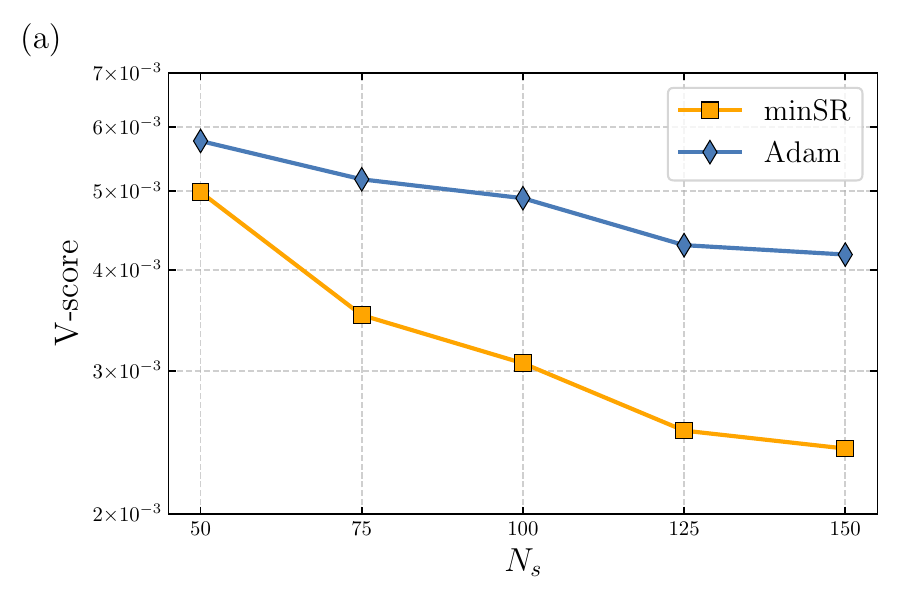}
    \includegraphics[width=\columnwidth]{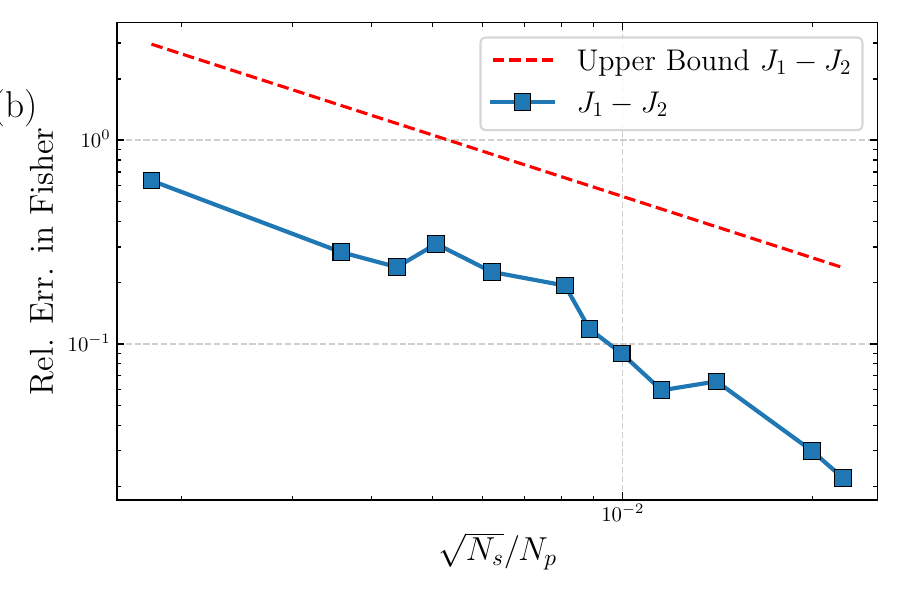}
    \caption{\textbf{Scaling Study of minSR.} (a) We plot the V-score (see Eq.~\eqref{eq:vscore}) of Adam and minSR on a $10\times 10$ antiferromagnetic Heisenberg model with OBC and $d_h=100$. We increase the number of samples $N_s$, plotted on the x-axis, and observe the effect on the V-score for both Adam and minSR. We keep an equal training time of 24 hours for all data points. We plot Adam as blue diamonds, and minSR as orange squares. (b) We plot the relative error of the subsampled SR metric $\|S-\tilde S\|_F/\|S\|_F$ against $\frac{\sqrt{N_s}}{N_p}$ for the $6\times 6$ $J_1-J_2$ model with $d_h=10$, shown as blue squares. Here $\|\cdot\|_F$ is the Frobenius norm. We plot the theoretical upper bound given by Eq.~\eqref{eq:UB} with $\epsilon=10^{-6}$ in the red dashed line as a guide to the eye.}
    \label{fig:scaling}
\end{figure}
In this section, we explore the effect of $N_s$ on the performance of minSR. We find that accurate variational energies can be obtained even with $N_s\ll N_p$. In Fig.~\ref{fig:scaling}(a), we plot both the V-score scaling of Adam and minSR with increased $N_s$, while holding the hidden dimension constant. All experiments are run for a training time of 24 hours as a control, leading to fewer training iterations for minSR and for larger $N_s$. Adam achieves nearly twice as many training iterations as minSR for $N_s=150$. Despite this speed advantage, minSR outperforms Adam for all values of $N_s$. Additionally, the advantage of minSR over Adam increases with $N_s$, while their performance is similar for low sample sizes.

The sensitivity of minSR to $N_s$ can be explained by the accuracy of our finite-sample SR metric. In particular, a good approximation of the SR metric requires $N_s\gg N_p$, where $N_p$ is the total number of parameters, e.g., $N_s\approx 10 N_p$, which is recommended in the literature~\cite{Becca_Sorella_2017}. Theoretical error bounds for the subsampled SR matrix follow from a Markov inequality argument~\cite{chebyshev-inequality, fisher-var}. We state the result as Proposition~\ref{prop:subsampled_fisher} and prove it in App.~\ref{app:proof_subsampled}.

\begin{proposition}[Subsampled QGT error bound]
\label{prop:subsampled_fisher}
Let the configurations
$\bm\sigma^{(1)},\dots,\bm\sigma^{(N_s)}$
be drawn i.i.d.\ from $|\Psi_{\bm\theta}|^2$, which exact
autoregressive sampling guarantees, and write
\[
O_i(\bm\sigma)
=
\partial_{\theta_i}\log\Psi_{\bm\theta}(\bm\sigma),
\qquad i=1,\dots,N_p,
\]
for the log-derivatives, which may be complex. Define
\begin{equation}
\mu_i\equiv\mathbb E[O_i],
\qquad
X_i\equiv O_i-\mu_i.
\end{equation}
Let $Q$ denote the population QGT, with entries
\begin{equation}
Q_{ij}
=
\mathbb E[X_iX_j^*],
\end{equation}
and let $\widetilde Q$ be the corresponding population-centered
$N_s$-sample estimator,
\begin{equation}
\widetilde Q_{ij}
=
\frac{1}{N_s}
\sum_{s=1}^{N_s}
X_i\big(\bm\sigma^{(s)}\big)
X_j\big(\bm\sigma^{(s)}\big)^*.
\label{eq:Qtilde_est}
\end{equation}
Then $\mathbb E[\widetilde Q]=Q$. Let
\begin{equation}
V_{\mathrm{avg}}^{(c)}
\equiv
\frac{1}{N_p^2}
\sum_{i,j=1}^{N_p}
\operatorname{Var}(X_iX_j^*)
\label{eq:vars1}
\end{equation}
be the average variance of a single-sample matrix element, where
$\operatorname{Var}(Z)
\equiv
\mathbb E[|Z-\mathbb E[Z]|^2]$,
and assume $V_{\mathrm{avg}}^{(c)}<\infty$.
Then, for every $\epsilon\in(0,1)$, with probability at least
$1-\epsilon$,
\begin{equation}
\|Q-\widetilde Q\|_F
\le
\frac{N_p}{\sqrt{N_s}}
\sqrt{\frac{V_{\mathrm{avg}}^{(c)}}{\epsilon}}.
\end{equation}
For real variational parameters, the population SR metric is
$S=\operatorname{Re}Q$. Defining
$\widetilde S=\operatorname{Re}\widetilde Q$, the same bound implies
\begin{equation}
\|S-\widetilde S\|_F
\le
\alpha^{-1}
\sqrt{\frac{V_{\mathrm{avg}}^{(c)}}{\epsilon}},
\label{eq:UB}
\end{equation}
where $\alpha=\sqrt{N_s}/N_p$. The same right-hand side also bounds the corresponding spectral norms,
since $\|A\|_2\le\|A\|_F$.
\end{proposition}
The proof is provided in App.~\ref{app:proof_subsampled}. Note that computing $\widetilde S$ with an empirical mean $\hat{\mu}_i$, as we do in our numerical simulations, still results in an upper bound proportional to $\alpha^{-1}$ as highlighted in App.~\ref{app:proof_subsampled}.

In Fig.~\ref{fig:scaling}(b), we plot the relative error between the subsampled Fisher $\tilde S$ and a high-sample reference estimate $S$ of the Fisher matrix for the two-dimensional $J_1-J_2$ model. We use $N_s\gtrsim N_p$ to construct this reference Fisher estimate. We observe that as $\alpha$ grows, the decay rate of the subsampling error qualitatively matches the inverse relationship with $\alpha$ theoretically predicted by the upper bound~\eqref{eq:UB}.


\section{Conclusions}

In this work, we demonstrate a stable and reproducible application of minSR for optimizing RNN wave functions within the VMC framework. Our methodology is based on the standard regularization of the Fisher matrix. Our results show a significant improvement over the Adam optimizer on the 1D TFIM and 1D cluster state benchmarks. In contrast, we obtain comparable performance to the Adam optimizer in two spatial dimensions, namely on the square lattice Heisenberg and $J_1-J_2$ model benchmarks. In most of our experiments, we observe faster convergence of the minSR setup compared to the first-order optimizer using only a few hundred Monte Carlo samples. Our scaling study further shows that a subsampled SR matrix with $N_s \ll N_p$ remains a useful preconditioner, with an estimation error consistent with the theoretical bound in Eq.~\eqref{eq:UB}. 

An important investigation path for future work is to check whether the combinatorial growth of backpropagation paths in the two-dimensional RNN architecture is the main cause behind the ill-conditioned loss landscape~\cite{GridLSTM}. More sophisticated damping techniques may enable minSR to optimize on ill-conditioned loss landscapes. Alternatively, GridLSTM cells are known to overcome this ill-conditioning and merit further exploration~\cite{GridLSTM}. Additionally, the LeakyLP~\cite{leakyMDRNN} cell could yield a two-dimensional RNN wave function that delivers superior results under minSR optimization.  Our results also highlight the memory advantage provided by first-order optimizers, which allow for scalability to very large system sizes~\cite{merali2026parallelscanrecurrentneural}. Deploying second-order optimization on small system sizes for faster convergence, and then switching to first-order optimizers on large system sizes under an iterative training procedure~\cite{roth2020iterativeretrainingquantumspin, hibatallah2022supplementingrecurrentneuralnetwork, Schuyler, nh89-6jmf, merali2026parallelscanrecurrentneural}, offers a practical route to balancing accuracy and scalability against the memory limits of available hardware. Additionally, we note that our variational energies are calculated without applying point group symmetries on our RNN ansatz, as our main goal is to compare minSR and Adam optimizers on equal footing. We believe that imposing physical symmetries on our architecture will further improve our variational accuracies through large-scale simulations, which we plan to pursue in future work.

An open question raised by our results is whether the geometry of the variational manifold can predict when minSR outperforms first-order optimizers. In particular, the spectrum of the QGT and the effective dimension of the sample-spanned tangent space may quantify how much curvature information is captured by a finite number of samples. Relating these information-geometric quantities to the observed performance gap between minSR and Adam across models and spatial dimensions is therefore a promising direction for future work~\cite{Dash2025}. More broadly, geometric and curvature-aware methods have proved useful for characterizing and improving learning dynamics in other machine-learning settings~\cite{GeoHNN,DMDGEN,GRATIN}, suggesting that a similar systematic study of variational geometry could provide useful diagnostics for predicting optimization performance in neural quantum states.

\section*{Code availability}
Our code is available at \texttt{https://github.com/werro40/RNN-Natural-Gradients}.

\section*{Acknowledgments}
We thank Juan Carrasquilla, Ejaaz Merali, Adam Smith, and Omar Masri for insightful discussions. Computer simulations were made possible thanks to the Digital Research Alliance of Canada and the Math Faculty Computing Facility at the University of Waterloo. M.H. acknowledges support from the Natural Sciences and Engineering Research Council of Canada (NSERC) and the Digital Research Alliance of Canada.

\appendix

\clearpage
\bibliography{Biblio}

\begin{thebibliography}{90}%
\makeatletter
\providecommand \@ifxundefined [1]{%
 \@ifx{#1\undefined}
}%
\providecommand \@ifnum [1]{%
 \ifnum #1\expandafter \@firstoftwo
 \else \expandafter \@secondoftwo
 \fi
}%
\providecommand \@ifx [1]{%
 \ifx #1\expandafter \@firstoftwo
 \else \expandafter \@secondoftwo
 \fi
}%
\providecommand \natexlab [1]{#1}%
\providecommand \enquote  [1]{``#1''}%
\providecommand \bibnamefont  [1]{#1}%
\providecommand \bibfnamefont [1]{#1}%
\providecommand \citenamefont [1]{#1}%
\providecommand \href@noop [0]{\@secondoftwo}%
\providecommand \href [0]{\begingroup \@sanitize@url \@href}%
\providecommand \@href[1]{\@@startlink{#1}\@@href}%
\providecommand \@@href[1]{\endgroup#1\@@endlink}%
\providecommand \@sanitize@url [0]{\catcode `\\12\catcode `\$12\catcode `\&12\catcode `\#12\catcode `\^12\catcode `\_12\catcode `\%12\relax}%
\providecommand \@@startlink[1]{}%
\providecommand \@@endlink[0]{}%
\providecommand \url  [0]{\begingroup\@sanitize@url \@url }%
\providecommand \@url [1]{\endgroup\@href {#1}{\urlprefix }}%
\providecommand \urlprefix  [0]{URL }%
\providecommand \Eprint [0]{\href }%
\providecommand \doibase [0]{http://dx.doi.org/}%
\providecommand \selectlanguage [0]{\@gobble}%
\providecommand \bibinfo  [0]{\@secondoftwo}%
\providecommand \bibfield  [0]{\@secondoftwo}%
\providecommand \translation [1]{[#1]}%
\providecommand \BibitemOpen [0]{}%
\providecommand \bibitemStop [0]{}%
\providecommand \bibitemNoStop [0]{.\EOS\space}%
\providecommand \EOS [0]{\spacefactor3000\relax}%
\providecommand \BibitemShut  [1]{\csname bibitem#1\endcsname}%
\let\auto@bib@innerbib\@empty
\bibitem [{\citenamefont {Carleo}\ and\ \citenamefont {Troyer}(2017)}]{Carleo_original_2017}%
  \BibitemOpen
  \bibfield  {author} {\bibinfo {author} {\bibfnamefont {Giuseppe}\ \bibnamefont {Carleo}}\ and\ \bibinfo {author} {\bibfnamefont {Matthias}\ \bibnamefont {Troyer}},\ }\bibfield  {title} {\enquote {\bibinfo {title} {Solving the quantum many-body problem with artificial neural networks},}\ }\href {\doibase 10.1126/science.aag2302} {\bibfield  {journal} {\bibinfo  {journal} {Science}\ }\textbf {\bibinfo {volume} {355}},\ \bibinfo {pages} {602–606} (\bibinfo {year} {2017})}\BibitemShut {NoStop}%
\bibitem [{\citenamefont {Lange}\ \emph {et~al.}(2024{\natexlab{a}})\citenamefont {Lange}, \citenamefont {Van~de Walle}, \citenamefont {Abedinnia},\ and\ \citenamefont {Bohrdt}}]{NQS-rev}%
  \BibitemOpen
  \bibfield  {author} {\bibinfo {author} {\bibfnamefont {Hannah}\ \bibnamefont {Lange}}, \bibinfo {author} {\bibfnamefont {Anka}\ \bibnamefont {Van~de Walle}}, \bibinfo {author} {\bibfnamefont {Atiye}\ \bibnamefont {Abedinnia}}, \ and\ \bibinfo {author} {\bibfnamefont {Annabelle}\ \bibnamefont {Bohrdt}},\ }\bibfield  {title} {\enquote {\bibinfo {title} {From architectures to applications: a review of neural quantum states},}\ }\href {\doibase 10.1088/2058-9565/ad7168} {\bibfield  {journal} {\bibinfo  {journal} {Quantum Science and Technology}\ }\textbf {\bibinfo {volume} {9}},\ \bibinfo {pages} {040501} (\bibinfo {year} {2024}{\natexlab{a}})}\BibitemShut {NoStop}%
\bibitem [{\citenamefont {Dawid}\ \emph {et~al.}(2025)\citenamefont {Dawid}, \citenamefont {Arnold}, \citenamefont {Requena}, \citenamefont {Gresch}, \citenamefont {Płodzień}, \citenamefont {Donatella}, \citenamefont {Nicoli}, \citenamefont {Stornati}, \citenamefont {Koch}, \citenamefont {Büttner}, \citenamefont {Okuła}, \citenamefont {Muñoz-Gil}, \citenamefont {Vargas-Hernández}, \citenamefont {Cervera-Lierta}, \citenamefont {Carrasquilla}, \citenamefont {Dunjko}, \citenamefont {Gabrié}, \citenamefont {Huembeli}, \citenamefont {van Nieuwenburg}, \citenamefont {Vicentini}, \citenamefont {Wang}, \citenamefont {Wetzel}, \citenamefont {Carleo}, \citenamefont {Greplová}, \citenamefont {Krems}, \citenamefont {Marquardt}, \citenamefont {Tomza}, \citenamefont {Lewenstein},\ and\ \citenamefont {Dauphin}}]{Dawid_2025}%
  \BibitemOpen
  \bibfield  {author} {\bibinfo {author} {\bibfnamefont {Anna}\ \bibnamefont {Dawid}}, \bibinfo {author} {\bibfnamefont {Julian}\ \bibnamefont {Arnold}}, \bibinfo {author} {\bibfnamefont {Borja}\ \bibnamefont {Requena}}, \bibinfo {author} {\bibfnamefont {Alexander}\ \bibnamefont {Gresch}}, \bibinfo {author} {\bibfnamefont {Marcin}\ \bibnamefont {Płodzień}}, \bibinfo {author} {\bibfnamefont {Kaelan}\ \bibnamefont {Donatella}}, \bibinfo {author} {\bibfnamefont {Kim~A.}\ \bibnamefont {Nicoli}}, \bibinfo {author} {\bibfnamefont {Paolo}\ \bibnamefont {Stornati}}, \bibinfo {author} {\bibfnamefont {Rouven}\ \bibnamefont {Koch}}, \bibinfo {author} {\bibfnamefont {Miriam}\ \bibnamefont {Büttner}}, \bibinfo {author} {\bibfnamefont {Robert}\ \bibnamefont {Okuła}}, \bibinfo {author} {\bibfnamefont {Gorka}\ \bibnamefont {Muñoz-Gil}}, \bibinfo {author} {\bibfnamefont {Rodrigo~A.}\ \bibnamefont {Vargas-Hernández}}, \bibinfo {author} {\bibfnamefont {Alba}\ \bibnamefont {Cervera-Lierta}}, \bibinfo {author}
  {\bibfnamefont {Juan}\ \bibnamefont {Carrasquilla}}, \bibinfo {author} {\bibfnamefont {Vedran}\ \bibnamefont {Dunjko}}, \bibinfo {author} {\bibfnamefont {Marylou}\ \bibnamefont {Gabrié}}, \bibinfo {author} {\bibfnamefont {Patrick}\ \bibnamefont {Huembeli}}, \bibinfo {author} {\bibfnamefont {Evert}\ \bibnamefont {van Nieuwenburg}}, \bibinfo {author} {\bibfnamefont {Filippo}\ \bibnamefont {Vicentini}}, \bibinfo {author} {\bibfnamefont {Lei}\ \bibnamefont {Wang}}, \bibinfo {author} {\bibfnamefont {Sebastian~J.}\ \bibnamefont {Wetzel}}, \bibinfo {author} {\bibfnamefont {Giuseppe}\ \bibnamefont {Carleo}}, \bibinfo {author} {\bibfnamefont {Eliška}\ \bibnamefont {Greplová}}, \bibinfo {author} {\bibfnamefont {Roman}\ \bibnamefont {Krems}}, \bibinfo {author} {\bibfnamefont {Florian}\ \bibnamefont {Marquardt}}, \bibinfo {author} {\bibfnamefont {Michał}\ \bibnamefont {Tomza}}, \bibinfo {author} {\bibfnamefont {Maciej}\ \bibnamefont {Lewenstein}}, \ and\ \bibinfo {author} {\bibfnamefont {Alexandre}\ \bibnamefont
  {Dauphin}},\ }\href {\doibase 10.1017/9781009504942} {\emph {\bibinfo {title} {Machine Learning in Quantum Sciences}}}\ (\bibinfo  {publisher} {Cambridge University Press},\ \bibinfo {year} {2025})\BibitemShut {NoStop}%
\bibitem [{\citenamefont {Choo}\ \emph {et~al.}(2020)\citenamefont {Choo}, \citenamefont {Mezzacapo},\ and\ \citenamefont {Carleo}}]{Choo2020}%
  \BibitemOpen
  \bibfield  {author} {\bibinfo {author} {\bibfnamefont {Kenny}\ \bibnamefont {Choo}}, \bibinfo {author} {\bibfnamefont {Antonio}\ \bibnamefont {Mezzacapo}}, \ and\ \bibinfo {author} {\bibfnamefont {Giuseppe}\ \bibnamefont {Carleo}},\ }\bibfield  {title} {\enquote {\bibinfo {title} {Fermionic neural-network states for ab-initio electronic structure},}\ }\href {\doibase 10.1038/s41467-020-15724-9} {\bibfield  {journal} {\bibinfo  {journal} {Nature Communications}\ }\textbf {\bibinfo {volume} {11}},\ \bibinfo {pages} {2368} (\bibinfo {year} {2020})}\BibitemShut {NoStop}%
\bibitem [{\citenamefont {V}\ and\ \citenamefont {Medhi}(2024)}]{v2024restrictedboltzmannmachinenetwork}%
  \BibitemOpen
  \bibfield  {author} {\bibinfo {author} {\bibfnamefont {Karthik}\ \bibnamefont {V}}\ and\ \bibinfo {author} {\bibfnamefont {Amal}\ \bibnamefont {Medhi}},\ }\href {https://arxiv.org/abs/2412.04103} {\enquote {\bibinfo {title} {Restricted boltzmann machine network versus jastrow correlated wave function for the two-dimensional hubbard model},}\ } (\bibinfo {year} {2024}),\ \Eprint {http://arxiv.org/abs/2412.04103} {arXiv:2412.04103 [cond-mat.str-el]} \BibitemShut {NoStop}%
\bibitem [{\citenamefont {Nomura}\ \emph {et~al.}(2017)\citenamefont {Nomura}, \citenamefont {Darmawan}, \citenamefont {Yamaji},\ and\ \citenamefont {Imada}}]{PhysRevB.96.205152}%
  \BibitemOpen
  \bibfield  {author} {\bibinfo {author} {\bibfnamefont {Yusuke}\ \bibnamefont {Nomura}}, \bibinfo {author} {\bibfnamefont {Andrew~S.}\ \bibnamefont {Darmawan}}, \bibinfo {author} {\bibfnamefont {Youhei}\ \bibnamefont {Yamaji}}, \ and\ \bibinfo {author} {\bibfnamefont {Masatoshi}\ \bibnamefont {Imada}},\ }\bibfield  {title} {\enquote {\bibinfo {title} {Restricted boltzmann machine learning for solving strongly correlated quantum systems},}\ }\href {\doibase 10.1103/PhysRevB.96.205152} {\bibfield  {journal} {\bibinfo  {journal} {Phys. Rev. B}\ }\textbf {\bibinfo {volume} {96}},\ \bibinfo {pages} {205152} (\bibinfo {year} {2017})}\BibitemShut {NoStop}%
\bibitem [{\citenamefont {Nomura}\ and\ \citenamefont {Imada}(2021)}]{PhysRevX.11.031034}%
  \BibitemOpen
  \bibfield  {author} {\bibinfo {author} {\bibfnamefont {Yusuke}\ \bibnamefont {Nomura}}\ and\ \bibinfo {author} {\bibfnamefont {Masatoshi}\ \bibnamefont {Imada}},\ }\bibfield  {title} {\enquote {\bibinfo {title} {Dirac-type nodal spin liquid revealed by refined quantum many-body solver using neural-network wave function, correlation ratio, and level spectroscopy},}\ }\href {\doibase 10.1103/PhysRevX.11.031034} {\bibfield  {journal} {\bibinfo  {journal} {Phys. Rev. X}\ }\textbf {\bibinfo {volume} {11}},\ \bibinfo {pages} {031034} (\bibinfo {year} {2021})}\BibitemShut {NoStop}%
\bibitem [{\citenamefont {Luo}\ and\ \citenamefont {Clark}(2019)}]{Di_Luo}%
  \BibitemOpen
  \bibfield  {author} {\bibinfo {author} {\bibfnamefont {Di}~\bibnamefont {Luo}}\ and\ \bibinfo {author} {\bibfnamefont {Bryan~K.}\ \bibnamefont {Clark}},\ }\bibfield  {title} {\enquote {\bibinfo {title} {Backflow transformations via neural networks for quantum many-body wave functions},}\ }\href {\doibase 10.1103/PhysRevLett.122.226401} {\bibfield  {journal} {\bibinfo  {journal} {Phys. Rev. Lett.}\ }\textbf {\bibinfo {volume} {122}},\ \bibinfo {pages} {226401} (\bibinfo {year} {2019})}\BibitemShut {NoStop}%
\bibitem [{\citenamefont {Cai}\ and\ \citenamefont {Liu}(2018)}]{zi2018}%
  \BibitemOpen
  \bibfield  {author} {\bibinfo {author} {\bibfnamefont {Zi}~\bibnamefont {Cai}}\ and\ \bibinfo {author} {\bibfnamefont {Jinguo}\ \bibnamefont {Liu}},\ }\bibfield  {title} {\enquote {\bibinfo {title} {Approximating quantum many-body wave functions using artificial neural networks},}\ }\href {\doibase 10.1103/PhysRevB.97.035116} {\bibfield  {journal} {\bibinfo  {journal} {Phys. Rev. B}\ }\textbf {\bibinfo {volume} {97}},\ \bibinfo {pages} {035116} (\bibinfo {year} {2018})}\BibitemShut {NoStop}%
\bibitem [{\citenamefont {Choo}\ \emph {et~al.}(2019)\citenamefont {Choo}, \citenamefont {Neupert},\ and\ \citenamefont {Carleo}}]{Choo_2019}%
  \BibitemOpen
  \bibfield  {author} {\bibinfo {author} {\bibfnamefont {Kenny}\ \bibnamefont {Choo}}, \bibinfo {author} {\bibfnamefont {Titus}\ \bibnamefont {Neupert}}, \ and\ \bibinfo {author} {\bibfnamefont {Giuseppe}\ \bibnamefont {Carleo}},\ }\bibfield  {title} {\enquote {\bibinfo {title} {Two-dimensional frustrated j1-j2 model studied with neural network quantum states},}\ }\href {\doibase 10.1103/physrevb.100.125124} {\bibfield  {journal} {\bibinfo  {journal} {Physical Review B}\ }\textbf {\bibinfo {volume} {100}} (\bibinfo {year} {2019}),\ 10.1103/physrevb.100.125124}\BibitemShut {NoStop}%
\bibitem [{\citenamefont {Roth}\ \emph {et~al.}(2023)\citenamefont {Roth}, \citenamefont {Szab\'o},\ and\ \citenamefont {MacDonald}}]{PhysRevB.108.054410}%
  \BibitemOpen
  \bibfield  {author} {\bibinfo {author} {\bibfnamefont {Christopher}\ \bibnamefont {Roth}}, \bibinfo {author} {\bibfnamefont {Attila}\ \bibnamefont {Szab\'o}}, \ and\ \bibinfo {author} {\bibfnamefont {Allan~H.}\ \bibnamefont {MacDonald}},\ }\bibfield  {title} {\enquote {\bibinfo {title} {High-accuracy variational monte carlo for frustrated magnets with deep neural networks},}\ }\href {\doibase 10.1103/PhysRevB.108.054410} {\bibfield  {journal} {\bibinfo  {journal} {Phys. Rev. B}\ }\textbf {\bibinfo {volume} {108}},\ \bibinfo {pages} {054410} (\bibinfo {year} {2023})}\BibitemShut {NoStop}%
\bibitem [{\citenamefont {Chen}\ and\ \citenamefont {Heyl}(2024)}]{Chen_2024}%
  \BibitemOpen
  \bibfield  {author} {\bibinfo {author} {\bibfnamefont {Ao}~\bibnamefont {Chen}}\ and\ \bibinfo {author} {\bibfnamefont {Markus}\ \bibnamefont {Heyl}},\ }\bibfield  {title} {\enquote {\bibinfo {title} {Empowering deep neural quantum states through efficient optimization},}\ }\href {\doibase 10.1038/s41567-024-02566-1} {\bibfield  {journal} {\bibinfo  {journal} {Nature Physics}\ }\textbf {\bibinfo {volume} {20}},\ \bibinfo {pages} {1476–1481} (\bibinfo {year} {2024})}\BibitemShut {NoStop}%
\bibitem [{\citenamefont {Hibat-Allah}\ \emph {et~al.}(2020)\citenamefont {Hibat-Allah}, \citenamefont {Ganahl}, \citenamefont {Hayward}, \citenamefont {Melko},\ and\ \citenamefont {Carrasquilla}}]{RNN-Wavefunc}%
  \BibitemOpen
  \bibfield  {author} {\bibinfo {author} {\bibfnamefont {Mohamed}\ \bibnamefont {Hibat-Allah}}, \bibinfo {author} {\bibfnamefont {Martin}\ \bibnamefont {Ganahl}}, \bibinfo {author} {\bibfnamefont {Lauren~E.}\ \bibnamefont {Hayward}}, \bibinfo {author} {\bibfnamefont {Roger~G.}\ \bibnamefont {Melko}}, \ and\ \bibinfo {author} {\bibfnamefont {Juan}\ \bibnamefont {Carrasquilla}},\ }\bibfield  {title} {\enquote {\bibinfo {title} {Recurrent neural network wave functions},}\ }\href {\doibase 10.1103/PhysRevResearch.2.023358} {\bibfield  {journal} {\bibinfo  {journal} {Phys. Rev. Res.}\ }\textbf {\bibinfo {volume} {2}},\ \bibinfo {pages} {023358} (\bibinfo {year} {2020})}\BibitemShut {NoStop}%
\bibitem [{\citenamefont {Hibat-Allah}\ \emph {et~al.}(2022)\citenamefont {Hibat-Allah}, \citenamefont {Melko},\ and\ \citenamefont {Carrasquilla}}]{hibatallah2022supplementingrecurrentneuralnetwork}%
  \BibitemOpen
  \bibfield  {author} {\bibinfo {author} {\bibfnamefont {Mohamed}\ \bibnamefont {Hibat-Allah}}, \bibinfo {author} {\bibfnamefont {Roger~G.}\ \bibnamefont {Melko}}, \ and\ \bibinfo {author} {\bibfnamefont {Juan}\ \bibnamefont {Carrasquilla}},\ }\href {https://arxiv.org/abs/2207.14314} {\enquote {\bibinfo {title} {Supplementing recurrent neural network wave functions with symmetry and annealing to improve accuracy},}\ } (\bibinfo {year} {2022}),\ \Eprint {http://arxiv.org/abs/2207.14314} {arXiv:2207.14314 [cond-mat.dis-nn]} \BibitemShut {NoStop}%
\bibitem [{\citenamefont {Wu}\ \emph {et~al.}(2023)\citenamefont {Wu}, \citenamefont {Rossi}, \citenamefont {Vicentini},\ and\ \citenamefont {Carleo}}]{Wu_2023}%
  \BibitemOpen
  \bibfield  {author} {\bibinfo {author} {\bibfnamefont {Dian}\ \bibnamefont {Wu}}, \bibinfo {author} {\bibfnamefont {Riccardo}\ \bibnamefont {Rossi}}, \bibinfo {author} {\bibfnamefont {Filippo}\ \bibnamefont {Vicentini}}, \ and\ \bibinfo {author} {\bibfnamefont {Giuseppe}\ \bibnamefont {Carleo}},\ }\bibfield  {title} {\enquote {\bibinfo {title} {From tensor-network quantum states to tensorial recurrent neural networks},}\ }\href {\doibase 10.1103/physrevresearch.5.l032001} {\bibfield  {journal} {\bibinfo  {journal} {Physical Review Research}\ }\textbf {\bibinfo {volume} {5}} (\bibinfo {year} {2023}),\ 10.1103/physrevresearch.5.l032001}\BibitemShut {NoStop}%
\bibitem [{\citenamefont {Roth}(2020)}]{roth2020iterativeretrainingquantumspin}%
  \BibitemOpen
  \bibfield  {author} {\bibinfo {author} {\bibfnamefont {Christopher}\ \bibnamefont {Roth}},\ }\href {https://arxiv.org/abs/2003.06228} {\enquote {\bibinfo {title} {Iterative retraining of quantum spin models using recurrent neural networks},}\ } (\bibinfo {year} {2020}),\ \Eprint {http://arxiv.org/abs/2003.06228} {arXiv:2003.06228 [physics.comp-ph]} \BibitemShut {NoStop}%
\bibitem [{\citenamefont {Luo}\ \emph {et~al.}(2021)\citenamefont {Luo}, \citenamefont {Chen}, \citenamefont {Hu}, \citenamefont {Zhao}, \citenamefont {Hur},\ and\ \citenamefont {Clark}}]{luo2021gauge}%
  \BibitemOpen
  \bibfield  {author} {\bibinfo {author} {\bibfnamefont {Di}~\bibnamefont {Luo}}, \bibinfo {author} {\bibfnamefont {Zhuo}\ \bibnamefont {Chen}}, \bibinfo {author} {\bibfnamefont {Kaiwen}\ \bibnamefont {Hu}}, \bibinfo {author} {\bibfnamefont {Zhizhen}\ \bibnamefont {Zhao}}, \bibinfo {author} {\bibfnamefont {Vera~Mikyoung}\ \bibnamefont {Hur}}, \ and\ \bibinfo {author} {\bibfnamefont {Bryan~K.}\ \bibnamefont {Clark}},\ }\bibfield  {title} {\enquote {\bibinfo {title} {Gauge invariant autoregressive neural networks for quantum lattice models},}\ }\href@noop {} {\  (\bibinfo {year} {2021})},\ \Eprint {http://arxiv.org/abs/2101.07243} {arXiv:2101.07243 [cond-mat.str-el]} \BibitemShut {NoStop}%
\bibitem [{\citenamefont {Moss}\ \emph {et~al.}(2025{\natexlab{a}})\citenamefont {Moss}, \citenamefont {Wiersema}, \citenamefont {Hibat-Allah}, \citenamefont {Carrasquilla},\ and\ \citenamefont {Melko}}]{Schuyler}%
  \BibitemOpen
  \bibfield  {author} {\bibinfo {author} {\bibfnamefont {M.~Schuyler}\ \bibnamefont {Moss}}, \bibinfo {author} {\bibfnamefont {Roeland}\ \bibnamefont {Wiersema}}, \bibinfo {author} {\bibfnamefont {Mohamed}\ \bibnamefont {Hibat-Allah}}, \bibinfo {author} {\bibfnamefont {Juan}\ \bibnamefont {Carrasquilla}}, \ and\ \bibinfo {author} {\bibfnamefont {Roger~G.}\ \bibnamefont {Melko}},\ }\bibfield  {title} {\enquote {\bibinfo {title} {Leveraging recurrence in neural network wavefunctions for large-scale simulations of heisenberg antiferromagnets on the square lattice},}\ }\href {\doibase 10.1103/6ccd-wzhz} {\bibfield  {journal} {\bibinfo  {journal} {Phys. Rev. B}\ }\textbf {\bibinfo {volume} {112}},\ \bibinfo {pages} {134450} (\bibinfo {year} {2025}{\natexlab{a}})}\BibitemShut {NoStop}%
\bibitem [{\citenamefont {Moss}\ \emph {et~al.}(2025{\natexlab{b}})\citenamefont {Moss}, \citenamefont {Wiersema}, \citenamefont {Hibat-Allah}, \citenamefont {Carrasquilla},\ and\ \citenamefont {Melko}}]{nh89-6jmf}%
  \BibitemOpen
  \bibfield  {author} {\bibinfo {author} {\bibfnamefont {M.~Schuyler}\ \bibnamefont {Moss}}, \bibinfo {author} {\bibfnamefont {Roeland}\ \bibnamefont {Wiersema}}, \bibinfo {author} {\bibfnamefont {Mohamed}\ \bibnamefont {Hibat-Allah}}, \bibinfo {author} {\bibfnamefont {Juan}\ \bibnamefont {Carrasquilla}}, \ and\ \bibinfo {author} {\bibfnamefont {Roger~G.}\ \bibnamefont {Melko}},\ }\bibfield  {title} {\enquote {\bibinfo {title} {Leveraging recurrence in neural network wavefunctions for large-scale simulations of heisenberg antiferromagnets on the triangular lattice},}\ }\href {\doibase 10.1103/nh89-6jmf} {\bibfield  {journal} {\bibinfo  {journal} {Phys. Rev. B}\ }\textbf {\bibinfo {volume} {112}},\ \bibinfo {pages} {134449} (\bibinfo {year} {2025}{\natexlab{b}})}\BibitemShut {NoStop}%
\bibitem [{\citenamefont {Hibat-Allah}\ \emph {et~al.}(2025)\citenamefont {Hibat-Allah}, \citenamefont {Merali}, \citenamefont {Torlai}, \citenamefont {Melko},\ and\ \citenamefont {Carrasquilla}}]{Hibat_Allah_2025}%
  \BibitemOpen
  \bibfield  {author} {\bibinfo {author} {\bibfnamefont {Mohamed}\ \bibnamefont {Hibat-Allah}}, \bibinfo {author} {\bibfnamefont {Ejaaz}\ \bibnamefont {Merali}}, \bibinfo {author} {\bibfnamefont {Giacomo}\ \bibnamefont {Torlai}}, \bibinfo {author} {\bibfnamefont {Roger~G.}\ \bibnamefont {Melko}}, \ and\ \bibinfo {author} {\bibfnamefont {Juan}\ \bibnamefont {Carrasquilla}},\ }\bibfield  {title} {\enquote {\bibinfo {title} {Recurrent neural network wave functions for rydberg atom arrays on kagome lattice},}\ }\href {\doibase 10.1038/s42005-025-02226-7} {\bibfield  {journal} {\bibinfo  {journal} {Communications Physics}\ }\textbf {\bibinfo {volume} {8}} (\bibinfo {year} {2025}),\ 10.1038/s42005-025-02226-7}\BibitemShut {NoStop}%
\bibitem [{\citenamefont {Merali}\ \emph {et~al.}(2026)\citenamefont {Merali}, \citenamefont {Hibat-Allah}, \citenamefont {Kohandel}, \citenamefont {Scalettar},\ and\ \citenamefont {Khatami}}]{merali2026parallelscanrecurrentneural}%
  \BibitemOpen
  \bibfield  {author} {\bibinfo {author} {\bibfnamefont {Ejaaz}\ \bibnamefont {Merali}}, \bibinfo {author} {\bibfnamefont {Mohamed}\ \bibnamefont {Hibat-Allah}}, \bibinfo {author} {\bibfnamefont {Mohammad}\ \bibnamefont {Kohandel}}, \bibinfo {author} {\bibfnamefont {Richard~T.}\ \bibnamefont {Scalettar}}, \ and\ \bibinfo {author} {\bibfnamefont {Ehsan}\ \bibnamefont {Khatami}},\ }\href {https://arxiv.org/abs/2605.13807} {\enquote {\bibinfo {title} {Parallel scan recurrent neural quantum states for scalable variational monte carlo},}\ } (\bibinfo {year} {2026}),\ \Eprint {http://arxiv.org/abs/2605.13807} {arXiv:2605.13807 [cond-mat.str-el]} \BibitemShut {NoStop}%
\bibitem [{\citenamefont {Zhang}\ and\ \citenamefont {Di~Ventra}(2023)}]{Zhang_2023}%
  \BibitemOpen
  \bibfield  {author} {\bibinfo {author} {\bibfnamefont {Yuan-Hang}\ \bibnamefont {Zhang}}\ and\ \bibinfo {author} {\bibfnamefont {Massimiliano}\ \bibnamefont {Di~Ventra}},\ }\bibfield  {title} {\enquote {\bibinfo {title} {Transformer quantum state: A multipurpose model for quantum many-body problems},}\ }\href {\doibase 10.1103/physrevb.107.075147} {\bibfield  {journal} {\bibinfo  {journal} {Physical Review B}\ }\textbf {\bibinfo {volume} {107}} (\bibinfo {year} {2023}),\ 10.1103/physrevb.107.075147}\BibitemShut {NoStop}%
\bibitem [{\citenamefont {Viteritti}\ \emph {et~al.}(2023)\citenamefont {Viteritti}, \citenamefont {Rende},\ and\ \citenamefont {Becca}}]{PhysRevLett.130.236401}%
  \BibitemOpen
  \bibfield  {author} {\bibinfo {author} {\bibfnamefont {Luciano~Loris}\ \bibnamefont {Viteritti}}, \bibinfo {author} {\bibfnamefont {Riccardo}\ \bibnamefont {Rende}}, \ and\ \bibinfo {author} {\bibfnamefont {Federico}\ \bibnamefont {Becca}},\ }\bibfield  {title} {\enquote {\bibinfo {title} {Transformer variational wave functions for frustrated quantum spin systems},}\ }\href {\doibase 10.1103/PhysRevLett.130.236401} {\bibfield  {journal} {\bibinfo  {journal} {Phys. Rev. Lett.}\ }\textbf {\bibinfo {volume} {130}},\ \bibinfo {pages} {236401} (\bibinfo {year} {2023})}\BibitemShut {NoStop}%
\bibitem [{\citenamefont {Rende}\ \emph {et~al.}(2024)\citenamefont {Rende}, \citenamefont {Viteritti}, \citenamefont {Bardone}, \citenamefont {Becca},\ and\ \citenamefont {Goldt}}]{minSR-lianlg}%
  \BibitemOpen
  \bibfield  {author} {\bibinfo {author} {\bibfnamefont {Riccardo}\ \bibnamefont {Rende}}, \bibinfo {author} {\bibfnamefont {Luciano~Loris}\ \bibnamefont {Viteritti}}, \bibinfo {author} {\bibfnamefont {Lorenzo}\ \bibnamefont {Bardone}}, \bibinfo {author} {\bibfnamefont {Federico}\ \bibnamefont {Becca}}, \ and\ \bibinfo {author} {\bibfnamefont {Sebastian}\ \bibnamefont {Goldt}},\ }\bibfield  {title} {\enquote {\bibinfo {title} {A simple linear algebra identity to optimize large-scale neural network quantum states},}\ }\href {\doibase 10.1038/s42005-024-01732-4} {\bibfield  {journal} {\bibinfo  {journal} {Communications Physics}\ }\textbf {\bibinfo {volume} {7}},\ \bibinfo {pages} {260} (\bibinfo {year} {2024})}\BibitemShut {NoStop}%
\bibitem [{\citenamefont {Chen}\ \emph {et~al.}(2026)\citenamefont {Chen}, \citenamefont {Naik},\ and\ \citenamefont {Heyl}}]{7xwp-25y9}%
  \BibitemOpen
  \bibfield  {author} {\bibinfo {author} {\bibfnamefont {Ao}~\bibnamefont {Chen}}, \bibinfo {author} {\bibfnamefont {Vighnesh~Dattatraya}\ \bibnamefont {Naik}}, \ and\ \bibinfo {author} {\bibfnamefont {Markus}\ \bibnamefont {Heyl}},\ }\bibfield  {title} {\enquote {\bibinfo {title} {Convolutional transformer wave functions},}\ }\href {\doibase 10.1103/7xwp-25y9} {\bibfield  {journal} {\bibinfo  {journal} {Phys. Rev. Res.}\ }\textbf {\bibinfo {volume} {8}},\ \bibinfo {pages} {L022040} (\bibinfo {year} {2026})}\BibitemShut {NoStop}%
\bibitem [{\citenamefont {Sprague}\ and\ \citenamefont {Czischek}(2024)}]{Sprague_2024}%
  \BibitemOpen
  \bibfield  {author} {\bibinfo {author} {\bibfnamefont {Kyle}\ \bibnamefont {Sprague}}\ and\ \bibinfo {author} {\bibfnamefont {Stefanie}\ \bibnamefont {Czischek}},\ }\bibfield  {title} {\enquote {\bibinfo {title} {Variational monte carlo with large patched transformers},}\ }\href {\doibase 10.1038/s42005-024-01584-y} {\bibfield  {journal} {\bibinfo  {journal} {Communications Physics}\ }\textbf {\bibinfo {volume} {7}} (\bibinfo {year} {2024}),\ 10.1038/s42005-024-01584-y}\BibitemShut {NoStop}%
\bibitem [{\citenamefont {Schmitt}\ and\ \citenamefont {Heyl}(2020)}]{PhysRevLett.125.100503}%
  \BibitemOpen
  \bibfield  {author} {\bibinfo {author} {\bibfnamefont {Markus}\ \bibnamefont {Schmitt}}\ and\ \bibinfo {author} {\bibfnamefont {Markus}\ \bibnamefont {Heyl}},\ }\bibfield  {title} {\enquote {\bibinfo {title} {Quantum many-body dynamics in two dimensions with artificial neural networks},}\ }\href {\doibase 10.1103/PhysRevLett.125.100503} {\bibfield  {journal} {\bibinfo  {journal} {Phys. Rev. Lett.}\ }\textbf {\bibinfo {volume} {125}},\ \bibinfo {pages} {100503} (\bibinfo {year} {2020})}\BibitemShut {NoStop}%
\bibitem [{\citenamefont {Sinibaldi}\ \emph {et~al.}(2026)\citenamefont {Sinibaldi}, \citenamefont {Hendry}, \citenamefont {Vicentini},\ and\ \citenamefont {Carleo}}]{kqvx-dl54}%
  \BibitemOpen
  \bibfield  {author} {\bibinfo {author} {\bibfnamefont {Alessandro}\ \bibnamefont {Sinibaldi}}, \bibinfo {author} {\bibfnamefont {Douglas}\ \bibnamefont {Hendry}}, \bibinfo {author} {\bibfnamefont {Filippo}\ \bibnamefont {Vicentini}}, \ and\ \bibinfo {author} {\bibfnamefont {Giuseppe}\ \bibnamefont {Carleo}},\ }\bibfield  {title} {\enquote {\bibinfo {title} {Time-dependent neural galerkin method for quantum dynamics},}\ }\href {\doibase 10.1103/kqvx-dl54} {\bibfield  {journal} {\bibinfo  {journal} {Phys. Rev. Lett.}\ }\textbf {\bibinfo {volume} {136}},\ \bibinfo {pages} {120402} (\bibinfo {year} {2026})}\BibitemShut {NoStop}%
\bibitem [{\citenamefont {Van~de Walle}\ \emph {et~al.}(2025)\citenamefont {Van~de Walle}, \citenamefont {Schmitt},\ and\ \citenamefont {Bohrdt}}]{Van_de_Walle_2025}%
  \BibitemOpen
  \bibfield  {author} {\bibinfo {author} {\bibfnamefont {Anka}\ \bibnamefont {Van~de Walle}}, \bibinfo {author} {\bibfnamefont {Markus}\ \bibnamefont {Schmitt}}, \ and\ \bibinfo {author} {\bibfnamefont {Annabelle}\ \bibnamefont {Bohrdt}},\ }\bibfield  {title} {\enquote {\bibinfo {title} {Many-body dynamics with explicitly time-dependent neural quantum states},}\ }\href {\doibase 10.1088/2632-2153/ae0f39} {\bibfield  {journal} {\bibinfo  {journal} {Machine Learning: Science and Technology}\ }\textbf {\bibinfo {volume} {6}},\ \bibinfo {pages} {045011} (\bibinfo {year} {2025})}\BibitemShut {NoStop}%
\bibitem [{\citenamefont {Sinibaldi}\ \emph {et~al.}(2023)\citenamefont {Sinibaldi}, \citenamefont {Giuliani}, \citenamefont {Carleo},\ and\ \citenamefont {Vicentini}}]{Sinibaldi2023unbiasingtime}%
  \BibitemOpen
  \bibfield  {author} {\bibinfo {author} {\bibfnamefont {Alessandro}\ \bibnamefont {Sinibaldi}}, \bibinfo {author} {\bibfnamefont {Clemens}\ \bibnamefont {Giuliani}}, \bibinfo {author} {\bibfnamefont {Giuseppe}\ \bibnamefont {Carleo}}, \ and\ \bibinfo {author} {\bibfnamefont {Filippo}\ \bibnamefont {Vicentini}},\ }\bibfield  {title} {\enquote {\bibinfo {title} {Unbiasing time-dependent {V}ariational {M}onte {C}arlo by projected quantum evolution},}\ }\href {\doibase 10.22331/q-2023-10-10-1131} {\bibfield  {journal} {\bibinfo  {journal} {{Quantum}}\ }\textbf {\bibinfo {volume} {7}},\ \bibinfo {pages} {1131} (\bibinfo {year} {2023})}\BibitemShut {NoStop}%
\bibitem [{\citenamefont {Vicentini}\ \emph {et~al.}(2019)\citenamefont {Vicentini}, \citenamefont {Biella}, \citenamefont {Regnault},\ and\ \citenamefont {Ciuti}}]{PhysRevLett.122.250503}%
  \BibitemOpen
  \bibfield  {author} {\bibinfo {author} {\bibfnamefont {Filippo}\ \bibnamefont {Vicentini}}, \bibinfo {author} {\bibfnamefont {Alberto}\ \bibnamefont {Biella}}, \bibinfo {author} {\bibfnamefont {Nicolas}\ \bibnamefont {Regnault}}, \ and\ \bibinfo {author} {\bibfnamefont {Cristiano}\ \bibnamefont {Ciuti}},\ }\bibfield  {title} {\enquote {\bibinfo {title} {Variational neural-network ansatz for steady states in open quantum systems},}\ }\href {\doibase 10.1103/PhysRevLett.122.250503} {\bibfield  {journal} {\bibinfo  {journal} {Phys. Rev. Lett.}\ }\textbf {\bibinfo {volume} {122}},\ \bibinfo {pages} {250503} (\bibinfo {year} {2019})}\BibitemShut {NoStop}%
\bibitem [{\citenamefont {Luo}\ \emph {et~al.}(2022)\citenamefont {Luo}, \citenamefont {Chen}, \citenamefont {Carrasquilla},\ and\ \citenamefont {Clark}}]{PhysRevLett.128.090501}%
  \BibitemOpen
  \bibfield  {author} {\bibinfo {author} {\bibfnamefont {Di}~\bibnamefont {Luo}}, \bibinfo {author} {\bibfnamefont {Zhuo}\ \bibnamefont {Chen}}, \bibinfo {author} {\bibfnamefont {Juan}\ \bibnamefont {Carrasquilla}}, \ and\ \bibinfo {author} {\bibfnamefont {Bryan~K.}\ \bibnamefont {Clark}},\ }\bibfield  {title} {\enquote {\bibinfo {title} {Autoregressive neural network for simulating open quantum systems via a probabilistic formulation},}\ }\href {\doibase 10.1103/PhysRevLett.128.090501} {\bibfield  {journal} {\bibinfo  {journal} {Phys. Rev. Lett.}\ }\textbf {\bibinfo {volume} {128}},\ \bibinfo {pages} {090501} (\bibinfo {year} {2022})}\BibitemShut {NoStop}%
\bibitem [{\citenamefont {Reh}\ \emph {et~al.}(2021)\citenamefont {Reh}, \citenamefont {Schmitt},\ and\ \citenamefont {G\"arttner}}]{PhysRevLett.127.230501}%
  \BibitemOpen
  \bibfield  {author} {\bibinfo {author} {\bibfnamefont {Moritz}\ \bibnamefont {Reh}}, \bibinfo {author} {\bibfnamefont {Markus}\ \bibnamefont {Schmitt}}, \ and\ \bibinfo {author} {\bibfnamefont {Martin}\ \bibnamefont {G\"arttner}},\ }\bibfield  {title} {\enquote {\bibinfo {title} {Time-dependent variational principle for open quantum systems with artificial neural networks},}\ }\href {\doibase 10.1103/PhysRevLett.127.230501} {\bibfield  {journal} {\bibinfo  {journal} {Phys. Rev. Lett.}\ }\textbf {\bibinfo {volume} {127}},\ \bibinfo {pages} {230501} (\bibinfo {year} {2021})}\BibitemShut {NoStop}%
\bibitem [{\citenamefont {Nomura}\ \emph {et~al.}(2021)\citenamefont {Nomura}, \citenamefont {Yoshioka},\ and\ \citenamefont {Nori}}]{Nomura_2021}%
  \BibitemOpen
  \bibfield  {author} {\bibinfo {author} {\bibfnamefont {Yusuke}\ \bibnamefont {Nomura}}, \bibinfo {author} {\bibfnamefont {Nobuyuki}\ \bibnamefont {Yoshioka}}, \ and\ \bibinfo {author} {\bibfnamefont {Franco}\ \bibnamefont {Nori}},\ }\bibfield  {title} {\enquote {\bibinfo {title} {Purifying deep boltzmann machines for thermal quantum states},}\ }\href {\doibase 10.1103/physrevlett.127.060601} {\bibfield  {journal} {\bibinfo  {journal} {Physical Review Letters}\ }\textbf {\bibinfo {volume} {127}} (\bibinfo {year} {2021}),\ 10.1103/physrevlett.127.060601}\BibitemShut {NoStop}%
\bibitem [{\citenamefont {Nys}\ \emph {et~al.}(2024)\citenamefont {Nys}, \citenamefont {Denis},\ and\ \citenamefont {Carleo}}]{nys2024realtimequantumdynamicsthermal}%
  \BibitemOpen
  \bibfield  {author} {\bibinfo {author} {\bibfnamefont {Jannes}\ \bibnamefont {Nys}}, \bibinfo {author} {\bibfnamefont {Zakari}\ \bibnamefont {Denis}}, \ and\ \bibinfo {author} {\bibfnamefont {Giuseppe}\ \bibnamefont {Carleo}},\ }\href {https://arxiv.org/abs/2309.07063} {\enquote {\bibinfo {title} {Real-time quantum dynamics of thermal states with neural thermofields},}\ } (\bibinfo {year} {2024}),\ \Eprint {http://arxiv.org/abs/2309.07063} {arXiv:2309.07063 [quant-ph]} \BibitemShut {NoStop}%
\bibitem [{\citenamefont {Torlai}\ and\ \citenamefont {Melko}(2018)}]{PhysRevLett.120.240503}%
  \BibitemOpen
  \bibfield  {author} {\bibinfo {author} {\bibfnamefont {Giacomo}\ \bibnamefont {Torlai}}\ and\ \bibinfo {author} {\bibfnamefont {Roger~G.}\ \bibnamefont {Melko}},\ }\bibfield  {title} {\enquote {\bibinfo {title} {Latent space purification via neural density operators},}\ }\href {\doibase 10.1103/PhysRevLett.120.240503} {\bibfield  {journal} {\bibinfo  {journal} {Phys. Rev. Lett.}\ }\textbf {\bibinfo {volume} {120}},\ \bibinfo {pages} {240503} (\bibinfo {year} {2018})}\BibitemShut {NoStop}%
\bibitem [{\citenamefont {Hendry}\ \emph {et~al.}(2022)\citenamefont {Hendry}, \citenamefont {Chen},\ and\ \citenamefont {Feiguin}}]{PhysRevB.106.165111}%
  \BibitemOpen
  \bibfield  {author} {\bibinfo {author} {\bibfnamefont {Douglas}\ \bibnamefont {Hendry}}, \bibinfo {author} {\bibfnamefont {Hongwei}\ \bibnamefont {Chen}}, \ and\ \bibinfo {author} {\bibfnamefont {Adrian}\ \bibnamefont {Feiguin}},\ }\bibfield  {title} {\enquote {\bibinfo {title} {Neural network representation for minimally entangled typical thermal states},}\ }\href {\doibase 10.1103/PhysRevB.106.165111} {\bibfield  {journal} {\bibinfo  {journal} {Phys. Rev. B}\ }\textbf {\bibinfo {volume} {106}},\ \bibinfo {pages} {165111} (\bibinfo {year} {2022})}\BibitemShut {NoStop}%
\bibitem [{\citenamefont {Kumar}\ \emph {et~al.}(2026)\citenamefont {Kumar}, \citenamefont {Balents}, \citenamefont {Hsieh},\ and\ \citenamefont {Melko}}]{KUMAR2026170362}%
  \BibitemOpen
  \bibfield  {author} {\bibinfo {author} {\bibfnamefont {Tarun~Advaith}\ \bibnamefont {Kumar}}, \bibinfo {author} {\bibfnamefont {Leon}\ \bibnamefont {Balents}}, \bibinfo {author} {\bibfnamefont {Timothy~H.}\ \bibnamefont {Hsieh}}, \ and\ \bibinfo {author} {\bibfnamefont {Roger~G.}\ \bibnamefont {Melko}},\ }\bibfield  {title} {\enquote {\bibinfo {title} {Autoregressive typical thermal states},}\ }\href {\doibase https://doi.org/10.1016/j.aop.2026.170362} {\bibfield  {journal} {\bibinfo  {journal} {Annals of Physics}\ }\textbf {\bibinfo {volume} {487}},\ \bibinfo {pages} {170362} (\bibinfo {year} {2026})}\BibitemShut {NoStop}%
\bibitem [{\citenamefont {Becca}\ and\ \citenamefont {Sorella}(2017)}]{Becca_Sorella_2017}%
  \BibitemOpen
  \bibfield  {author} {\bibinfo {author} {\bibfnamefont {Federico}\ \bibnamefont {Becca}}\ and\ \bibinfo {author} {\bibfnamefont {Sandro}\ \bibnamefont {Sorella}},\ }\href@noop {} {\emph {\bibinfo {title} {Quantum Monte Carlo Approaches for Correlated Systems}}}\ (\bibinfo  {publisher} {Cambridge University Press},\ \bibinfo {year} {2017})\BibitemShut {NoStop}%
\bibitem [{\citenamefont {Sorella}\ and\ \citenamefont {Capriotti}(2000)}]{PhysRevB.61.2599}%
  \BibitemOpen
  \bibfield  {author} {\bibinfo {author} {\bibfnamefont {Sandro}\ \bibnamefont {Sorella}}\ and\ \bibinfo {author} {\bibfnamefont {Luca}\ \bibnamefont {Capriotti}},\ }\bibfield  {title} {\enquote {\bibinfo {title} {Green function monte carlo with stochastic reconfiguration: An effective remedy for the sign problem},}\ }\href {\doibase 10.1103/PhysRevB.61.2599} {\bibfield  {journal} {\bibinfo  {journal} {Phys. Rev. B}\ }\textbf {\bibinfo {volume} {61}},\ \bibinfo {pages} {2599--2612} (\bibinfo {year} {2000})}\BibitemShut {NoStop}%
\bibitem [{\citenamefont {Amari}(1998)}]{amariNG}%
  \BibitemOpen
  \bibfield  {author} {\bibinfo {author} {\bibfnamefont {Shun-ichi}\ \bibnamefont {Amari}},\ }\bibfield  {title} {\enquote {\bibinfo {title} {Natural gradient works efficiently in learning},}\ }\href {\doibase 10.1162/089976698300017746} {\bibfield  {journal} {\bibinfo  {journal} {Neural Computation}\ }\textbf {\bibinfo {volume} {10}},\ \bibinfo {pages} {251--276} (\bibinfo {year} {1998})}\BibitemShut {NoStop}%
\bibitem [{\citenamefont {Park}\ and\ \citenamefont {Kastoryano}(2020)}]{PhysRevResearch.2.023232}%
  \BibitemOpen
  \bibfield  {author} {\bibinfo {author} {\bibfnamefont {Chae-Yeun}\ \bibnamefont {Park}}\ and\ \bibinfo {author} {\bibfnamefont {Michael~J.}\ \bibnamefont {Kastoryano}},\ }\bibfield  {title} {\enquote {\bibinfo {title} {Geometry of learning neural quantum states},}\ }\href {\doibase 10.1103/PhysRevResearch.2.023232} {\bibfield  {journal} {\bibinfo  {journal} {Phys. Rev. Res.}\ }\textbf {\bibinfo {volume} {2}},\ \bibinfo {pages} {023232} (\bibinfo {year} {2020})}\BibitemShut {NoStop}%
\bibitem [{\citenamefont {Dash}\ \emph {et~al.}(2025)\citenamefont {Dash}, \citenamefont {Gravina}, \citenamefont {Vicentini}, \citenamefont {Ferrero},\ and\ \citenamefont {Georges}}]{Dash2025}%
  \BibitemOpen
  \bibfield  {author} {\bibinfo {author} {\bibfnamefont {Sidhartha}\ \bibnamefont {Dash}}, \bibinfo {author} {\bibfnamefont {Luca}\ \bibnamefont {Gravina}}, \bibinfo {author} {\bibfnamefont {Filippo}\ \bibnamefont {Vicentini}}, \bibinfo {author} {\bibfnamefont {Michel}\ \bibnamefont {Ferrero}}, \ and\ \bibinfo {author} {\bibfnamefont {Antoine}\ \bibnamefont {Georges}},\ }\bibfield  {title} {\enquote {\bibinfo {title} {Efficiency of neural quantum states in light of the quantum geometric tensor},}\ }\href {\doibase 10.1038/s42005-025-02005-4} {\bibfield  {journal} {\bibinfo  {journal} {Communications Physics}\ }\textbf {\bibinfo {volume} {8}},\ \bibinfo {pages} {92} (\bibinfo {year} {2025})}\BibitemShut {NoStop}%
\bibitem [{\citenamefont {Donatella}\ \emph {et~al.}(2023)\citenamefont {Donatella}, \citenamefont {Denis}, \citenamefont {Le~Boit\'e},\ and\ \citenamefont {Ciuti}}]{donatella23}%
  \BibitemOpen
  \bibfield  {author} {\bibinfo {author} {\bibfnamefont {Kaelan}\ \bibnamefont {Donatella}}, \bibinfo {author} {\bibfnamefont {Zakari}\ \bibnamefont {Denis}}, \bibinfo {author} {\bibfnamefont {Alexandre}\ \bibnamefont {Le~Boit\'e}}, \ and\ \bibinfo {author} {\bibfnamefont {Cristiano}\ \bibnamefont {Ciuti}},\ }\bibfield  {title} {\enquote {\bibinfo {title} {Dynamics with autoregressive neural quantum states: Application to critical quench dynamics},}\ }\href {\doibase 10.1103/PhysRevA.108.022210} {\bibfield  {journal} {\bibinfo  {journal} {Phys. Rev. A}\ }\textbf {\bibinfo {volume} {108}},\ \bibinfo {pages} {022210} (\bibinfo {year} {2023})}\BibitemShut {NoStop}%
\bibitem [{\citenamefont {Lange}\ \emph {et~al.}(2024{\natexlab{b}})\citenamefont {Lange}, \citenamefont {D{\"o}schl}, \citenamefont {Carrasquilla},\ and\ \citenamefont {Bohrdt}}]{Lange2024}%
  \BibitemOpen
  \bibfield  {author} {\bibinfo {author} {\bibfnamefont {Hannah}\ \bibnamefont {Lange}}, \bibinfo {author} {\bibfnamefont {Fabian}\ \bibnamefont {D{\"o}schl}}, \bibinfo {author} {\bibfnamefont {Juan}\ \bibnamefont {Carrasquilla}}, \ and\ \bibinfo {author} {\bibfnamefont {Annabelle}\ \bibnamefont {Bohrdt}},\ }\bibfield  {title} {\enquote {\bibinfo {title} {Neural network approach to quasiparticle dispersions in doped antiferromagnets},}\ }\href {\doibase 10.1038/s42005-024-01678-7} {\bibfield  {journal} {\bibinfo  {journal} {Communications Physics}\ }\textbf {\bibinfo {volume} {7}},\ \bibinfo {pages} {187} (\bibinfo {year} {2024}{\natexlab{b}})}\BibitemShut {NoStop}%
\bibitem [{\citenamefont {Duque}\ \emph {et~al.}(2026)\citenamefont {Duque}, \citenamefont {Heredia}, \citenamefont {Hernandes}, \citenamefont {Greplová}, \citenamefont {Spriggs}, \citenamefont {Courville},\ and\ \citenamefont {Dawid}}]{duque2026timerevisitingneuralquantum}%
  \BibitemOpen
  \bibfield  {author} {\bibinfo {author} {\bibfnamefont {Juan~Agustín}\ \bibnamefont {Duque}}, \bibinfo {author} {\bibfnamefont {Sergio~García}\ \bibnamefont {Heredia}}, \bibinfo {author} {\bibfnamefont {Vinicius}\ \bibnamefont {Hernandes}}, \bibinfo {author} {\bibfnamefont {Eliška}\ \bibnamefont {Greplová}}, \bibinfo {author} {\bibfnamefont {Thomas}\ \bibnamefont {Spriggs}}, \bibinfo {author} {\bibfnamefont {Aaron}\ \bibnamefont {Courville}}, \ and\ \bibinfo {author} {\bibfnamefont {Anna}\ \bibnamefont {Dawid}},\ }\href {https://arxiv.org/abs/2607.02292} {\enquote {\bibinfo {title} {One more time: Revisiting neural quantum states from a reinforcement learning perspective},}\ } (\bibinfo {year} {2026}),\ \Eprint {http://arxiv.org/abs/2607.02292} {arXiv:2607.02292 [cs.LG]} \BibitemShut {NoStop}%
\bibitem [{\citenamefont {Kingma}\ and\ \citenamefont {Ba}(2017)}]{AdamOriginal}%
  \BibitemOpen
  \bibfield  {author} {\bibinfo {author} {\bibfnamefont {Diederik~P.}\ \bibnamefont {Kingma}}\ and\ \bibinfo {author} {\bibfnamefont {Jimmy}\ \bibnamefont {Ba}},\ }\href {https://arxiv.org/abs/1412.6980} {\enquote {\bibinfo {title} {Adam: A method for stochastic optimization},}\ } (\bibinfo {year} {2017}),\ \Eprint {http://arxiv.org/abs/1412.6980} {arXiv:1412.6980 [cs.LG]} \BibitemShut {NoStop}%
\bibitem [{\citenamefont {Liu}\ \emph {et~al.}(2025)\citenamefont {Liu}, \citenamefont {Tang},\ and\ \citenamefont {Zhang}}]{Liu_2025}%
  \BibitemOpen
  \bibfield  {author} {\bibinfo {author} {\bibfnamefont {Jing}\ \bibnamefont {Liu}}, \bibinfo {author} {\bibfnamefont {Ying}\ \bibnamefont {Tang}}, \ and\ \bibinfo {author} {\bibfnamefont {Pan}\ \bibnamefont {Zhang}},\ }\bibfield  {title} {\enquote {\bibinfo {title} {Efficient optimization of variational autoregressive networks with natural gradient},}\ }\href {\doibase 10.1103/physreve.111.025304} {\bibfield  {journal} {\bibinfo  {journal} {Physical Review E}\ }\textbf {\bibinfo {volume} {111}} (\bibinfo {year} {2025}),\ 10.1103/physreve.111.025304}\BibitemShut {NoStop}%
\bibitem [{\citenamefont {Malyshev}(2024)}]{autoregressiveChem}%
  \BibitemOpen
  \bibfield  {author} {\bibinfo {author} {\bibfnamefont {Aleksei}\ \bibnamefont {Malyshev}},\ }\emph {\bibinfo {title} {Autoregressive neural quantum states for ab initio quantum chemistry}},\ \href {\doibase 10.5287/ORA-GAPQOBDRE} {Ph.D. thesis},\ \bibinfo  {school} {University of Oxford} (\bibinfo {year} {2024})\BibitemShut {NoStop}%
\bibitem [{\citenamefont {Sharir}\ \emph {et~al.}(2020)\citenamefont {Sharir}, \citenamefont {Levine}, \citenamefont {Wies}, \citenamefont {Carleo},\ and\ \citenamefont {Shashua}}]{Sharir_2020}%
  \BibitemOpen
  \bibfield  {author} {\bibinfo {author} {\bibfnamefont {Or}~\bibnamefont {Sharir}}, \bibinfo {author} {\bibfnamefont {Yoav}\ \bibnamefont {Levine}}, \bibinfo {author} {\bibfnamefont {Noam}\ \bibnamefont {Wies}}, \bibinfo {author} {\bibfnamefont {Giuseppe}\ \bibnamefont {Carleo}}, \ and\ \bibinfo {author} {\bibfnamefont {Amnon}\ \bibnamefont {Shashua}},\ }\bibfield  {title} {\enquote {\bibinfo {title} {Deep autoregressive models for the efficient variational simulation of many-body quantum systems},}\ }\href {\doibase 10.1103/physrevlett.124.020503} {\bibfield  {journal} {\bibinfo  {journal} {Physical Review Letters}\ }\textbf {\bibinfo {volume} {124}} (\bibinfo {year} {2020}),\ 10.1103/physrevlett.124.020503}\BibitemShut {NoStop}%
\bibitem [{\citenamefont {Lipton}\ \emph {et~al.}(2015)\citenamefont {Lipton}, \citenamefont {Berkowitz},\ and\ \citenamefont {Elkan}}]{Vanilla}%
  \BibitemOpen
  \bibfield  {author} {\bibinfo {author} {\bibfnamefont {Zachary~C.}\ \bibnamefont {Lipton}}, \bibinfo {author} {\bibfnamefont {John}\ \bibnamefont {Berkowitz}}, \ and\ \bibinfo {author} {\bibfnamefont {Charles}\ \bibnamefont {Elkan}},\ }\href {https://arxiv.org/abs/1506.00019} {\enquote {\bibinfo {title} {A critical review of recurrent neural networks for sequence learning},}\ } (\bibinfo {year} {2015}),\ \Eprint {http://arxiv.org/abs/1506.00019} {arXiv:1506.00019 [cs.LG]} \BibitemShut {NoStop}%
\bibitem [{\citenamefont {Ayub}\ \emph {et~al.}(2026)\citenamefont {Ayub}, \citenamefont {Aboussalah},\ and\ \citenamefont {Hibat-Allah}}]{DilatedRNN}%
  \BibitemOpen
  \bibfield  {author} {\bibinfo {author} {\bibfnamefont {Asif~Bin}\ \bibnamefont {Ayub}}, \bibinfo {author} {\bibfnamefont {Amine~Mohamed}\ \bibnamefont {Aboussalah}}, \ and\ \bibinfo {author} {\bibfnamefont {Mohamed}\ \bibnamefont {Hibat-Allah}},\ }\href {\doibase 10.48550/arXiv.2604.08661} {\enquote {\bibinfo {title} {Geometry-induced long-range correlations in recurrent neural network quantum states},}\ } (\bibinfo {year} {2026}),\ \Eprint {http://arxiv.org/abs/2604.08661} {arXiv:2604.08661 [quant-ph]} \BibitemShut {NoStop}%
\bibitem [{\citenamefont {Polyak}(1964)}]{POLYAK19641}%
  \BibitemOpen
  \bibfield  {author} {\bibinfo {author} {\bibfnamefont {B.T.}\ \bibnamefont {Polyak}},\ }\bibfield  {title} {\enquote {\bibinfo {title} {Some methods of speeding up the convergence of iteration methods},}\ }\href {\doibase https://doi.org/10.1016/0041-5553(64)90137-5} {\bibfield  {journal} {\bibinfo  {journal} {USSR Computational Mathematics and Mathematical Physics}\ }\textbf {\bibinfo {volume} {4}},\ \bibinfo {pages} {1--17} (\bibinfo {year} {1964})}\BibitemShut {NoStop}%
\bibitem [{\citenamefont {Amari}(2016)}]{amari_information_2016}%
  \BibitemOpen
  \bibfield  {author} {\bibinfo {author} {\bibfnamefont {Shun-ichi}\ \bibnamefont {Amari}},\ }\href {\doibase 10.1007/978-4-431-55978-8} {\emph {\bibinfo {title} {Information {Geometry} and {Its} {Applications}}}},\ \bibinfo {series} {Applied {Mathematical} {Sciences}}, Vol.\ \bibinfo {volume} {194}\ (\bibinfo  {publisher} {Springer Japan},\ \bibinfo {address} {Tokyo},\ \bibinfo {year} {2016})\BibitemShut {NoStop}%
\bibitem [{\citenamefont {Berner}\ \emph {et~al.}(2023)\citenamefont {Berner}, \citenamefont {Elbrächter},\ and\ \citenamefont {Grohs}}]{berner2023degenerateparametrizationneuralnetworks}%
  \BibitemOpen
  \bibfield  {author} {\bibinfo {author} {\bibfnamefont {Julius}\ \bibnamefont {Berner}}, \bibinfo {author} {\bibfnamefont {Dennis}\ \bibnamefont {Elbrächter}}, \ and\ \bibinfo {author} {\bibfnamefont {Philipp}\ \bibnamefont {Grohs}},\ }\href {https://arxiv.org/abs/1905.09803} {\enquote {\bibinfo {title} {How degenerate is the parametrization of neural networks with the relu activation function?}}\ } (\bibinfo {year} {2023}),\ \Eprint {http://arxiv.org/abs/1905.09803} {arXiv:1905.09803 [cs.LG]} \BibitemShut {NoStop}%
\bibitem [{\citenamefont {Nomura}(2023)}]{Nomura_2023}%
  \BibitemOpen
  \bibfield  {author} {\bibinfo {author} {\bibfnamefont {Yusuke}\ \bibnamefont {Nomura}},\ }\bibfield  {title} {\enquote {\bibinfo {title} {Boltzmann machines and quantum many-body problems},}\ }\href {\doibase 10.1088/1361-648x/ad0916} {\bibfield  {journal} {\bibinfo  {journal} {Journal of Physics: Condensed Matter}\ }\textbf {\bibinfo {volume} {36}},\ \bibinfo {pages} {073001} (\bibinfo {year} {2023})}\BibitemShut {NoStop}%
\bibitem [{\citenamefont {Ed-dib}\ \emph {et~al.}(2025)\citenamefont {Ed-dib}, \citenamefont {Datbayev},\ and\ \citenamefont {Aboussalah}}]{GeLoRA}%
  \BibitemOpen
  \bibfield  {author} {\bibinfo {author} {\bibfnamefont {Abdessalam}\ \bibnamefont {Ed-dib}}, \bibinfo {author} {\bibfnamefont {Zhanibek}\ \bibnamefont {Datbayev}}, \ and\ \bibinfo {author} {\bibfnamefont {Amine~Mohamed}\ \bibnamefont {Aboussalah}},\ }\bibfield  {title} {\enquote {\bibinfo {title} {{GeLoRA}: Geometric adaptive ranks for efficient {LoRA} fine-tuning},}\ }in\ \href {\doibase 10.18653/v1/2025.findings-emnlp.1372} {\emph {\bibinfo {booktitle} {Findings of the Association for Computational Linguistics: EMNLP 2025}}}\ (\bibinfo  {publisher} {Association for Computational Linguistics},\ \bibinfo {address} {Suzhou, China},\ \bibinfo {year} {2025})\ pp.\ \bibinfo {pages} {25174--25196}\BibitemShut {NoStop}%
\bibitem [{\citenamefont {Aboussalah}\ and\ \citenamefont {Ed-dib}(2025{\natexlab{a}})}]{GNNTopology}%
  \BibitemOpen
  \bibfield  {author} {\bibinfo {author} {\bibfnamefont {Amine~Mohamed}\ \bibnamefont {Aboussalah}}\ and\ \bibinfo {author} {\bibfnamefont {Abdessalam}\ \bibnamefont {Ed-dib}},\ }\href {\doibase 10.48550/arXiv.2502.17739} {\enquote {\bibinfo {title} {Are {GNNs} doomed by the topology of their input graph?}}\ } (\bibinfo {year} {2025}{\natexlab{a}}),\ \Eprint {http://arxiv.org/abs/2502.17739} {arXiv:2502.17739 [stat.ML]} \BibitemShut {NoStop}%
\bibitem [{\citenamefont {Sachdev}(2011)}]{Sachdev_2011}%
  \BibitemOpen
  \bibfield  {author} {\bibinfo {author} {\bibfnamefont {Subir}\ \bibnamefont {Sachdev}},\ }\href@noop {} {\emph {\bibinfo {title} {Quantum Phase Transitions}}},\ \bibinfo {edition} {2nd}\ ed.\ (\bibinfo  {publisher} {Cambridge University Press},\ \bibinfo {year} {2011})\BibitemShut {NoStop}%
\bibitem [{\citenamefont {Bravyi}\ \emph {et~al.}(2007)\citenamefont {Bravyi}, \citenamefont {DiVincenzo}, \citenamefont {Oliveira},\ and\ \citenamefont {Terhal}}]{bravyi2007complexitystoquasticlocalhamiltonian}%
  \BibitemOpen
  \bibfield  {author} {\bibinfo {author} {\bibfnamefont {Sergey}\ \bibnamefont {Bravyi}}, \bibinfo {author} {\bibfnamefont {David~P.}\ \bibnamefont {DiVincenzo}}, \bibinfo {author} {\bibfnamefont {Roberto~I.}\ \bibnamefont {Oliveira}}, \ and\ \bibinfo {author} {\bibfnamefont {Barbara~M.}\ \bibnamefont {Terhal}},\ }\href {https://arxiv.org/abs/quant-ph/0606140} {\enquote {\bibinfo {title} {The complexity of stoquastic local hamiltonian problems},}\ } (\bibinfo {year} {2007}),\ \Eprint {http://arxiv.org/abs/quant-ph/0606140} {arXiv:quant-ph/0606140 [quant-ph]} \BibitemShut {NoStop}%
\bibitem [{\citenamefont {Pfeuty}(1970)}]{ising}%
  \BibitemOpen
  \bibfield  {author} {\bibinfo {author} {\bibfnamefont {Pierre}\ \bibnamefont {Pfeuty}},\ }\bibfield  {title} {\enquote {\bibinfo {title} {The one-dimensional ising model with a transverse field},}\ }\href {\doibase https://doi.org/10.1016/0003-4916(70)90270-8} {\bibfield  {journal} {\bibinfo  {journal} {Annals of Physics}\ }\textbf {\bibinfo {volume} {57}},\ \bibinfo {pages} {79--90} (\bibinfo {year} {1970})}\BibitemShut {NoStop}%
\bibitem [{\citenamefont {Yang}\ \emph {et~al.}(2024)\citenamefont {Yang}, \citenamefont {Soleimanifar}, \citenamefont {Bergamaschi},\ and\ \citenamefont {Preskill}}]{cluster-preskill}%
  \BibitemOpen
  \bibfield  {author} {\bibinfo {author} {\bibfnamefont {Tai-Hsuan}\ \bibnamefont {Yang}}, \bibinfo {author} {\bibfnamefont {Mehdi}\ \bibnamefont {Soleimanifar}}, \bibinfo {author} {\bibfnamefont {Thiago}\ \bibnamefont {Bergamaschi}}, \ and\ \bibinfo {author} {\bibfnamefont {John}\ \bibnamefont {Preskill}},\ }\href {https://arxiv.org/abs/2410.23152} {\enquote {\bibinfo {title} {When can classical neural networks represent quantum states?}}\ } (\bibinfo {year} {2024}),\ \Eprint {http://arxiv.org/abs/2410.23152} {arXiv:2410.23152 [quant-ph]} \BibitemShut {NoStop}%
\bibitem [{\citenamefont {McNaughton}\ and\ \citenamefont {Hibat-Allah}(2025)}]{adaptive}%
  \BibitemOpen
  \bibfield  {author} {\bibinfo {author} {\bibfnamefont {Jake}\ \bibnamefont {McNaughton}}\ and\ \bibinfo {author} {\bibfnamefont {Mohamed}\ \bibnamefont {Hibat-Allah}},\ }\href {https://arxiv.org/abs/2507.18700} {\enquote {\bibinfo {title} {Adaptive neural quantum states: A recurrent neural network perspective},}\ } (\bibinfo {year} {2025}),\ \Eprint {http://arxiv.org/abs/2507.18700} {arXiv:2507.18700 [cond-mat.dis-nn]} \BibitemShut {NoStop}%
\bibitem [{\citenamefont {Marshall}(1955)}]{Marshall1955}%
  \BibitemOpen
  \bibfield  {author} {\bibinfo {author} {\bibfnamefont {W}~\bibnamefont {Marshall}},\ }\bibfield  {title} {\enquote {\bibinfo {title} {{Antiferromagnetism}},}\ }\href {\doibase 10.1098/rspa.1955.0200} {\bibfield  {journal} {\bibinfo  {journal} {Proceedings of the Royal Society of London. Series A. Mathematical and Physical Sciences}\ }\textbf {\bibinfo {volume} {232}},\ \bibinfo {pages} {48--68} (\bibinfo {year} {1955})}\BibitemShut {NoStop}%
\bibitem [{\citenamefont {Shamim}\ \emph {et~al.}(2026)\citenamefont {Shamim}, \citenamefont {Raj}, \citenamefont {Hibat-Allah},\ and\ \citenamefont {Araujo}}]{shamim2026graphtheoreticanalysisphaseoptimization}%
  \BibitemOpen
  \bibfield  {author} {\bibinfo {author} {\bibfnamefont {Mahmud~Ashraf}\ \bibnamefont {Shamim}}, \bibinfo {author} {\bibfnamefont {Md~Moshiur~Rahman}\ \bibnamefont {Raj}}, \bibinfo {author} {\bibfnamefont {Mohamed}\ \bibnamefont {Hibat-Allah}}, \ and\ \bibinfo {author} {\bibfnamefont {Paulo~T}\ \bibnamefont {Araujo}},\ }\href {https://arxiv.org/abs/2602.04943} {\enquote {\bibinfo {title} {Graph-theoretic analysis of phase optimization complexity in variational wave functions for heisenberg antiferromagnets},}\ } (\bibinfo {year} {2026}),\ \Eprint {http://arxiv.org/abs/2602.04943} {arXiv:2602.04943 [cond-mat.str-el]} \BibitemShut {NoStop}%
\bibitem [{\citenamefont {Wu}\ \emph {et~al.}(2024)\citenamefont {Wu}, \citenamefont {Rossi}, \citenamefont {Vicentini}, \citenamefont {Astrakhantsev}, \citenamefont {Becca}, \citenamefont {Cao}, \citenamefont {Carrasquilla}, \citenamefont {Ferrari}, \citenamefont {Georges}, \citenamefont {Hibat-Allah}, \citenamefont {Imada}, \citenamefont {Läuchli}, \citenamefont {Mazzola}, \citenamefont {Mezzacapo}, \citenamefont {Millis}, \citenamefont {Moreno}, \citenamefont {Neupert}, \citenamefont {Nomura}, \citenamefont {Nys}, \citenamefont {Parcollet}, \citenamefont {Pohle}, \citenamefont {Romero}, \citenamefont {Schmid}, \citenamefont {Silvester}, \citenamefont {Sorella}, \citenamefont {Tocchio}, \citenamefont {Wang}, \citenamefont {White}, \citenamefont {Wietek}, \citenamefont {Yang}, \citenamefont {Yang}, \citenamefont {Zhang},\ and\ \citenamefont {Carleo}}]{vscore}%
  \BibitemOpen
  \bibfield  {author} {\bibinfo {author} {\bibfnamefont {Dian}\ \bibnamefont {Wu}}, \bibinfo {author} {\bibfnamefont {Riccardo}\ \bibnamefont {Rossi}}, \bibinfo {author} {\bibfnamefont {Filippo}\ \bibnamefont {Vicentini}}, \bibinfo {author} {\bibfnamefont {Nikita}\ \bibnamefont {Astrakhantsev}}, \bibinfo {author} {\bibfnamefont {Federico}\ \bibnamefont {Becca}}, \bibinfo {author} {\bibfnamefont {Xiaodong}\ \bibnamefont {Cao}}, \bibinfo {author} {\bibfnamefont {Juan}\ \bibnamefont {Carrasquilla}}, \bibinfo {author} {\bibfnamefont {Francesco}\ \bibnamefont {Ferrari}}, \bibinfo {author} {\bibfnamefont {Antoine}\ \bibnamefont {Georges}}, \bibinfo {author} {\bibfnamefont {Mohamed}\ \bibnamefont {Hibat-Allah}}, \bibinfo {author} {\bibfnamefont {Masatoshi}\ \bibnamefont {Imada}}, \bibinfo {author} {\bibfnamefont {Andreas~M.}\ \bibnamefont {Läuchli}}, \bibinfo {author} {\bibfnamefont {Guglielmo}\ \bibnamefont {Mazzola}}, \bibinfo {author} {\bibfnamefont {Antonio}\ \bibnamefont {Mezzacapo}}, \bibinfo {author}
  {\bibfnamefont {Andrew}\ \bibnamefont {Millis}}, \bibinfo {author} {\bibfnamefont {Javier~Robledo}\ \bibnamefont {Moreno}}, \bibinfo {author} {\bibfnamefont {Titus}\ \bibnamefont {Neupert}}, \bibinfo {author} {\bibfnamefont {Yusuke}\ \bibnamefont {Nomura}}, \bibinfo {author} {\bibfnamefont {Jannes}\ \bibnamefont {Nys}}, \bibinfo {author} {\bibfnamefont {Olivier}\ \bibnamefont {Parcollet}}, \bibinfo {author} {\bibfnamefont {Rico}\ \bibnamefont {Pohle}}, \bibinfo {author} {\bibfnamefont {Imelda}\ \bibnamefont {Romero}}, \bibinfo {author} {\bibfnamefont {Michael}\ \bibnamefont {Schmid}}, \bibinfo {author} {\bibfnamefont {J.~Maxwell}\ \bibnamefont {Silvester}}, \bibinfo {author} {\bibfnamefont {Sandro}\ \bibnamefont {Sorella}}, \bibinfo {author} {\bibfnamefont {Luca~F.}\ \bibnamefont {Tocchio}}, \bibinfo {author} {\bibfnamefont {Lei}\ \bibnamefont {Wang}}, \bibinfo {author} {\bibfnamefont {Steven~R.}\ \bibnamefont {White}}, \bibinfo {author} {\bibfnamefont {Alexander}\ \bibnamefont {Wietek}}, \bibinfo {author}
  {\bibfnamefont {Qi}~\bibnamefont {Yang}}, \bibinfo {author} {\bibfnamefont {Yiqi}\ \bibnamefont {Yang}}, \bibinfo {author} {\bibfnamefont {Shiwei}\ \bibnamefont {Zhang}}, \ and\ \bibinfo {author} {\bibfnamefont {Giuseppe}\ \bibnamefont {Carleo}},\ }\bibfield  {title} {\enquote {\bibinfo {title} {Variational benchmarks for quantum many-body problems},}\ }\href {\doibase 10.1126/science.adg9774} {\bibfield  {journal} {\bibinfo  {journal} {Science}\ }\textbf {\bibinfo {volume} {386}},\ \bibinfo {pages} {296--301} (\bibinfo {year} {2024})}\BibitemShut {NoStop}%
\bibitem [{\citenamefont {Martens}\ and\ \citenamefont {Sutskever}(2011)}]{hf-rnn}%
  \BibitemOpen
  \bibfield  {author} {\bibinfo {author} {\bibfnamefont {James}\ \bibnamefont {Martens}}\ and\ \bibinfo {author} {\bibfnamefont {Ilya}\ \bibnamefont {Sutskever}},\ }\bibfield  {title} {\enquote {\bibinfo {title} {Learning recurrent neural networks with hessian-free optimization},}\ }in\ \href@noop {} {\emph {\bibinfo {booktitle} {Proceedings of the 28th international conference on machine learning (ICML-11)}}}\ (\bibinfo {year} {2011})\ pp.\ \bibinfo {pages} {1033--1040}\BibitemShut {NoStop}%
\bibitem [{\citenamefont {Graves}\ \emph {et~al.}(2007)\citenamefont {Graves}, \citenamefont {Fern{\'a}ndez},\ and\ \citenamefont {Schmidhuber}}]{mdrnn}%
  \BibitemOpen
  \bibfield  {author} {\bibinfo {author} {\bibfnamefont {Alex}\ \bibnamefont {Graves}}, \bibinfo {author} {\bibfnamefont {Santiago}\ \bibnamefont {Fern{\'a}ndez}}, \ and\ \bibinfo {author} {\bibfnamefont {J{\"u}rgen}\ \bibnamefont {Schmidhuber}},\ }\bibfield  {title} {\enquote {\bibinfo {title} {Multi-dimensional recurrent neural networks},}\ }in\ \href@noop {} {\emph {\bibinfo {booktitle} {International conference on artificial neural networks}}}\ (\bibinfo {organization} {Springer},\ \bibinfo {year} {2007})\ pp.\ \bibinfo {pages} {549--558}\BibitemShut {NoStop}%
\bibitem [{\citenamefont {Kalchbrenner}\ \emph {et~al.}(2015)\citenamefont {Kalchbrenner}, \citenamefont {Danihelka},\ and\ \citenamefont {Graves}}]{GridLSTM}%
  \BibitemOpen
  \bibfield  {author} {\bibinfo {author} {\bibfnamefont {Nal}\ \bibnamefont {Kalchbrenner}}, \bibinfo {author} {\bibfnamefont {Ivo}\ \bibnamefont {Danihelka}}, \ and\ \bibinfo {author} {\bibfnamefont {Alex}\ \bibnamefont {Graves}},\ }\bibfield  {title} {\enquote {\bibinfo {title} {Grid long short-term memory},}\ }\href@noop {} {\bibfield  {journal} {\bibinfo  {journal} {arXiv preprint arXiv:1507.01526}\ } (\bibinfo {year} {2015})}\BibitemShut {NoStop}%
\bibitem [{\citenamefont {Leifert}\ \emph {et~al.}(2016)\citenamefont {Leifert}, \citenamefont {Strau{\ss}}, \citenamefont {Gr{\"u}ning}, \citenamefont {Wustlich},\ and\ \citenamefont {Labahn}}]{leakyMDRNN}%
  \BibitemOpen
  \bibfield  {author} {\bibinfo {author} {\bibfnamefont {Gundram}\ \bibnamefont {Leifert}}, \bibinfo {author} {\bibfnamefont {Tobias}\ \bibnamefont {Strau{\ss}}}, \bibinfo {author} {\bibfnamefont {Tobias}\ \bibnamefont {Gr{\"u}ning}}, \bibinfo {author} {\bibfnamefont {Welf}\ \bibnamefont {Wustlich}}, \ and\ \bibinfo {author} {\bibfnamefont {Roger}\ \bibnamefont {Labahn}},\ }\bibfield  {title} {\enquote {\bibinfo {title} {Cells in multidimensional recurrent neural networks},}\ }\href@noop {} {\bibfield  {journal} {\bibinfo  {journal} {Journal of Machine Learning Research}\ }\textbf {\bibinfo {volume} {17}},\ \bibinfo {pages} {1--37} (\bibinfo {year} {2016})}\BibitemShut {NoStop}%
\bibitem [{\citenamefont {Li}\ \emph {et~al.}(2018)\citenamefont {Li}, \citenamefont {Xu}, \citenamefont {Taylor}, \citenamefont {Studer},\ and\ \citenamefont {Goldstein}}]{VizualizingRNNs}%
  \BibitemOpen
  \bibfield  {author} {\bibinfo {author} {\bibfnamefont {Hao}\ \bibnamefont {Li}}, \bibinfo {author} {\bibfnamefont {Zheng}\ \bibnamefont {Xu}}, \bibinfo {author} {\bibfnamefont {Gavin}\ \bibnamefont {Taylor}}, \bibinfo {author} {\bibfnamefont {Christoph}\ \bibnamefont {Studer}}, \ and\ \bibinfo {author} {\bibfnamefont {Tom}\ \bibnamefont {Goldstein}},\ }\bibfield  {title} {\enquote {\bibinfo {title} {Visualizing the loss landscape of neural nets},}\ }\href@noop {} {\bibfield  {journal} {\bibinfo  {journal} {Advances in neural information processing systems}\ }\textbf {\bibinfo {volume} {31}} (\bibinfo {year} {2018})}\BibitemShut {NoStop}%
\bibitem [{\citenamefont {B{\"o}ttcher}\ and\ \citenamefont {Wheeler}(2024)}]{bottcher2024visualizing}%
  \BibitemOpen
  \bibfield  {author} {\bibinfo {author} {\bibfnamefont {Lucas}\ \bibnamefont {B{\"o}ttcher}}\ and\ \bibinfo {author} {\bibfnamefont {Gregory}\ \bibnamefont {Wheeler}},\ }\bibfield  {title} {\enquote {\bibinfo {title} {Visualizing high-dimensional loss landscapes with hessian directions},}\ }\href@noop {} {\bibfield  {journal} {\bibinfo  {journal} {Journal of Statistical Mechanics: Theory and Experiment}\ }\textbf {\bibinfo {volume} {2024}},\ \bibinfo {pages} {023401} (\bibinfo {year} {2024})}\BibitemShut {NoStop}%
\bibitem [{\citenamefont {Chen}(2011)}]{chebyshev-inequality}%
  \BibitemOpen
  \bibfield  {author} {\bibinfo {author} {\bibfnamefont {Xinjia}\ \bibnamefont {Chen}},\ }\href {https://arxiv.org/abs/0707.0805} {\enquote {\bibinfo {title} {A new generalization of chebyshev inequality for random vectors},}\ } (\bibinfo {year} {2011}),\ \Eprint {http://arxiv.org/abs/0707.0805} {arXiv:0707.0805 [math.ST]} \BibitemShut {NoStop}%
\bibitem [{\citenamefont {Soen}\ and\ \citenamefont {Sun}(2021)}]{fisher-var}%
  \BibitemOpen
  \bibfield  {author} {\bibinfo {author} {\bibfnamefont {Alexander}\ \bibnamefont {Soen}}\ and\ \bibinfo {author} {\bibfnamefont {Ke}~\bibnamefont {Sun}},\ }\bibfield  {title} {\enquote {\bibinfo {title} {On the variance of the fisher information for deep learning},}\ }in\ \href {https://proceedings.neurips.cc/paper_files/paper/2021/file/2d290e496d16c9dcaa9b4ded5cac10cc-Paper.pdf} {\emph {\bibinfo {booktitle} {Advances in Neural Information Processing Systems}}},\ Vol.~\bibinfo {volume} {34},\ \bibinfo {editor} {edited by\ \bibinfo {editor} {\bibfnamefont {M.}~\bibnamefont {Ranzato}}, \bibinfo {editor} {\bibfnamefont {A.}~\bibnamefont {Beygelzimer}}, \bibinfo {editor} {\bibfnamefont {Y.}~\bibnamefont {Dauphin}}, \bibinfo {editor} {\bibfnamefont {P.S.}\ \bibnamefont {Liang}}, \ and\ \bibinfo {editor} {\bibfnamefont {J.~Wortman}\ \bibnamefont {Vaughan}}}\ (\bibinfo  {publisher} {Curran Associates, Inc.},\ \bibinfo {year} {2021})\ pp.\ \bibinfo {pages} {5708--5719}\BibitemShut {NoStop}%
\bibitem [{\citenamefont {Aboussalah}\ and\ \citenamefont {Ed-dib}(2025{\natexlab{b}})}]{GeoHNN}%
  \BibitemOpen
  \bibfield  {author} {\bibinfo {author} {\bibfnamefont {Amine~Mohamed}\ \bibnamefont {Aboussalah}}\ and\ \bibinfo {author} {\bibfnamefont {Abdessalam}\ \bibnamefont {Ed-dib}},\ }\href {\doibase 10.48550/arXiv.2507.15678} {\enquote {\bibinfo {title} {{GeoHNNs}: Geometric hamiltonian neural networks},}\ } (\bibinfo {year} {2025}{\natexlab{b}}),\ \Eprint {http://arxiv.org/abs/2507.15678} {arXiv:2507.15678 [cs.LG]} \BibitemShut {NoStop}%
\bibitem [{\citenamefont {Abbahaddou}\ and\ \citenamefont {Aboussalah}(2025)}]{DMDGEN}%
  \BibitemOpen
  \bibfield  {author} {\bibinfo {author} {\bibfnamefont {Yassine}\ \bibnamefont {Abbahaddou}}\ and\ \bibinfo {author} {\bibfnamefont {Amine~Mohamed}\ \bibnamefont {Aboussalah}},\ }\bibfield  {title} {\enquote {\bibinfo {title} {A geometry-aware metric for mode collapse in time series generative models},}\ }in\ \href {\doibase 10.52202/085713-5065} {\emph {\bibinfo {booktitle} {Advances in Neural Information Processing Systems 38}}}\ (\bibinfo {year} {2025})\ pp.\ \bibinfo {pages} {168149--168179}\BibitemShut {NoStop}%
\bibitem [{\citenamefont {Abbahaddou}\ \emph {et~al.}(2025)\citenamefont {Abbahaddou}, \citenamefont {Malliaros}, \citenamefont {Lutzeyer}, \citenamefont {Aboussalah},\ and\ \citenamefont {Vazirgiannis}}]{GRATIN}%
  \BibitemOpen
  \bibfield  {author} {\bibinfo {author} {\bibfnamefont {Yassine}\ \bibnamefont {Abbahaddou}}, \bibinfo {author} {\bibfnamefont {Fragkiskos~D.}\ \bibnamefont {Malliaros}}, \bibinfo {author} {\bibfnamefont {Johannes~F.}\ \bibnamefont {Lutzeyer}}, \bibinfo {author} {\bibfnamefont {Amine~M.}\ \bibnamefont {Aboussalah}}, \ and\ \bibinfo {author} {\bibfnamefont {Michalis}\ \bibnamefont {Vazirgiannis}},\ }\bibfield  {title} {\enquote {\bibinfo {title} {Graph neural network generalization with gaussian mixture model based augmentation},}\ }in\ \href {https://proceedings.mlr.press/v267/abbahaddou25a.html} {\emph {\bibinfo {booktitle} {Proceedings of the 42nd International Conference on Machine Learning}}},\ \bibinfo {series} {Proceedings of Machine Learning Research}, Vol.\ \bibinfo {volume} {267}\ (\bibinfo  {publisher} {PMLR},\ \bibinfo {year} {2025})\ pp.\ \bibinfo {pages} {102--125}\BibitemShut {NoStop}%
\bibitem [{\citenamefont {Cho}\ \emph {et~al.}(2014)\citenamefont {Cho}, \citenamefont {van Merri{\"e}nboer}, \citenamefont {Gulcehre}, \citenamefont {Bahdanau}, \citenamefont {Bougares}, \citenamefont {Schwenk},\ and\ \citenamefont {Bengio}}]{GRU}%
  \BibitemOpen
  \bibfield  {author} {\bibinfo {author} {\bibfnamefont {Kyunghyun}\ \bibnamefont {Cho}}, \bibinfo {author} {\bibfnamefont {Bart}\ \bibnamefont {van Merri{\"e}nboer}}, \bibinfo {author} {\bibfnamefont {Caglar}\ \bibnamefont {Gulcehre}}, \bibinfo {author} {\bibfnamefont {Dzmitry}\ \bibnamefont {Bahdanau}}, \bibinfo {author} {\bibfnamefont {Fethi}\ \bibnamefont {Bougares}}, \bibinfo {author} {\bibfnamefont {Holger}\ \bibnamefont {Schwenk}}, \ and\ \bibinfo {author} {\bibfnamefont {Yoshua}\ \bibnamefont {Bengio}},\ }\bibfield  {title} {\enquote {\bibinfo {title} {Learning phrase representations using {RNN} encoder{--}decoder for statistical machine translation},}\ }in\ \href {\doibase 10.3115/v1/D14-1179} {\emph {\bibinfo {booktitle} {Proceedings of the 2014 Conference on Empirical Methods in Natural Language Processing ({EMNLP})}}},\ \bibinfo {editor} {edited by\ \bibinfo {editor} {\bibfnamefont {Alessandro}\ \bibnamefont {Moschitti}}, \bibinfo {editor} {\bibfnamefont {Bo}~\bibnamefont {Pang}}, \ and\ \bibinfo
  {editor} {\bibfnamefont {Walter}\ \bibnamefont {Daelemans}}}\ (\bibinfo  {publisher} {Association for Computational Linguistics},\ \bibinfo {address} {Doha, Qatar},\ \bibinfo {year} {2014})\ pp.\ \bibinfo {pages} {1724--1734}\BibitemShut {NoStop}%
\bibitem [{\citenamefont {Goldshlager}\ \emph {et~al.}(2024)\citenamefont {Goldshlager}, \citenamefont {Abrahamsen},\ and\ \citenamefont {Lin}}]{SPRING}%
  \BibitemOpen
  \bibfield  {author} {\bibinfo {author} {\bibfnamefont {Gil}\ \bibnamefont {Goldshlager}}, \bibinfo {author} {\bibfnamefont {Nilin}\ \bibnamefont {Abrahamsen}}, \ and\ \bibinfo {author} {\bibfnamefont {Lin}\ \bibnamefont {Lin}},\ }\bibfield  {title} {\enquote {\bibinfo {title} {A kaczmarz-inspired approach to accelerate the optimization of neural network wavefunctions},}\ }\href {\doibase https://doi.org/10.1016/j.jcp.2024.113351} {\bibfield  {journal} {\bibinfo  {journal} {Journal of Computational Physics}\ }\textbf {\bibinfo {volume} {516}},\ \bibinfo {pages} {113351} (\bibinfo {year} {2024})}\BibitemShut {NoStop}%
\bibitem [{\citenamefont {Liu}\ \emph {et~al.}(2017)\citenamefont {Liu}, \citenamefont {Dong}, \citenamefont {Han}, \citenamefont {Guo},\ and\ \citenamefont {He}}]{QMCH2D}%
  \BibitemOpen
  \bibfield  {author} {\bibinfo {author} {\bibfnamefont {Wen-Yuan}\ \bibnamefont {Liu}}, \bibinfo {author} {\bibfnamefont {Shao-Jun}\ \bibnamefont {Dong}}, \bibinfo {author} {\bibfnamefont {Yong-Jian}\ \bibnamefont {Han}}, \bibinfo {author} {\bibfnamefont {Guang-Can}\ \bibnamefont {Guo}}, \ and\ \bibinfo {author} {\bibfnamefont {Lixin}\ \bibnamefont {He}},\ }\bibfield  {title} {\enquote {\bibinfo {title} {Gradient optimization of finite projected entangled pair states},}\ }\href {\doibase 10.1103/PhysRevB.95.195154} {\bibfield  {journal} {\bibinfo  {journal} {Phys. Rev. B}\ }\textbf {\bibinfo {volume} {95}},\ \bibinfo {pages} {195154} (\bibinfo {year} {2017})}\BibitemShut {NoStop}%
\bibitem [{\citenamefont {Liu}\ \emph {et~al.}(2018)\citenamefont {Liu}, \citenamefont {Dong}, \citenamefont {Wang}, \citenamefont {Han}, \citenamefont {An}, \citenamefont {Guo},\ and\ \citenamefont {He}}]{PEPSbench}%
  \BibitemOpen
  \bibfield  {author} {\bibinfo {author} {\bibfnamefont {Wen-Yuan}\ \bibnamefont {Liu}}, \bibinfo {author} {\bibfnamefont {Shaojun}\ \bibnamefont {Dong}}, \bibinfo {author} {\bibfnamefont {Chao}\ \bibnamefont {Wang}}, \bibinfo {author} {\bibfnamefont {Yongjian}\ \bibnamefont {Han}}, \bibinfo {author} {\bibfnamefont {Hong}\ \bibnamefont {An}}, \bibinfo {author} {\bibfnamefont {Guang-Can}\ \bibnamefont {Guo}}, \ and\ \bibinfo {author} {\bibfnamefont {Lixin}\ \bibnamefont {He}},\ }\bibfield  {title} {\enquote {\bibinfo {title} {Gapless spin liquid ground state of the spin-$\frac{1}{2}$ ${J}_{1}\ensuremath{-}{J}_{2}$ heisenberg model on square lattices},}\ }\href {\doibase 10.1103/PhysRevB.98.241109} {\bibfield  {journal} {\bibinfo  {journal} {Phys. Rev. B}\ }\textbf {\bibinfo {volume} {98}},\ \bibinfo {pages} {241109(R)} (\bibinfo {year} {2018})}\BibitemShut {NoStop}%
\bibitem [{\citenamefont {Zhang}\ \emph {et~al.}(2026)\citenamefont {Zhang}, \citenamefont {Armegioiu}, \citenamefont {Carrasquilla}, \citenamefont {Mishra}, \citenamefont {M{\"u}ller}, \citenamefont {Nys},\ and\ \citenamefont {Zeinhofer}}]{PII}%
  \BibitemOpen
  \bibfield  {author} {\bibinfo {author} {\bibfnamefont {Hang}\ \bibnamefont {Zhang}}, \bibinfo {author} {\bibfnamefont {Victor}\ \bibnamefont {Armegioiu}}, \bibinfo {author} {\bibfnamefont {Juan}\ \bibnamefont {Carrasquilla}}, \bibinfo {author} {\bibfnamefont {Siddhartha}\ \bibnamefont {Mishra}}, \bibinfo {author} {\bibfnamefont {Johannes}\ \bibnamefont {M{\"u}ller}}, \bibinfo {author} {\bibfnamefont {Jannes}\ \bibnamefont {Nys}}, \ and\ \bibinfo {author} {\bibfnamefont {Marius}\ \bibnamefont {Zeinhofer}},\ }\bibfield  {title} {\enquote {\bibinfo {title} {Projected inverse iteration: An eigenvalue approach to ground-state computation with neural quantum states},}\ }\href@noop {} {\bibfield  {journal} {\bibinfo  {journal} {arXiv preprint arXiv:2606.07825}\ } (\bibinfo {year} {2026})}\BibitemShut {NoStop}%
\bibitem [{\citenamefont {Solinas}\ \emph {et~al.}(2026)\citenamefont {Solinas}, \citenamefont {Valenti}, \citenamefont {Bou-Rabee},\ and\ \citenamefont {Wiersema}}]{mixed-precision}%
  \BibitemOpen
  \bibfield  {author} {\bibinfo {author} {\bibfnamefont {Massimo}\ \bibnamefont {Solinas}}, \bibinfo {author} {\bibfnamefont {Agnes}\ \bibnamefont {Valenti}}, \bibinfo {author} {\bibfnamefont {Nawaf}\ \bibnamefont {Bou-Rabee}}, \ and\ \bibinfo {author} {\bibfnamefont {Roeland}\ \bibnamefont {Wiersema}},\ }\bibfield  {title} {\enquote {\bibinfo {title} {Neural quantum states in mixed precision},}\ }\href@noop {} {\bibfield  {journal} {\bibinfo  {journal} {arXiv preprint arXiv:2601.20782}\ } (\bibinfo {year} {2026})}\BibitemShut {NoStop}%
\bibitem [{\citenamefont {Higham}(2002)}]{Accuracy}%
  \BibitemOpen
  \bibfield  {author} {\bibinfo {author} {\bibfnamefont {Nicholas~J.}\ \bibnamefont {Higham}},\ }\href {\doibase 10.1137/1.9780898718027} {\emph {\bibinfo {title} {Accuracy and Stability of Numerical Algorithms}}},\ \bibinfo {edition} {2nd}\ ed.\ (\bibinfo  {publisher} {Society for Industrial and Applied Mathematics},\ \bibinfo {year} {2002})\BibitemShut {NoStop}%
\bibitem [{\citenamefont {Darken}\ and\ \citenamefont {Moody}(1990)}]{exploration-phase}%
  \BibitemOpen
  \bibfield  {author} {\bibinfo {author} {\bibfnamefont {Christian}\ \bibnamefont {Darken}}\ and\ \bibinfo {author} {\bibfnamefont {John}\ \bibnamefont {Moody}},\ }\bibfield  {title} {\enquote {\bibinfo {title} {Note on learning rate schedules for stochastic optimization},}\ }\href@noop {} {\bibfield  {journal} {\bibinfo  {journal} {Advances in neural information processing systems}\ }\textbf {\bibinfo {volume} {3}} (\bibinfo {year} {1990})}\BibitemShut {NoStop}%
\bibitem [{\citenamefont {Martens}\ and\ \citenamefont {Grosse}(2020)}]{kfac}%
  \BibitemOpen
  \bibfield  {author} {\bibinfo {author} {\bibfnamefont {James}\ \bibnamefont {Martens}}\ and\ \bibinfo {author} {\bibfnamefont {Roger}\ \bibnamefont {Grosse}},\ }\href {https://arxiv.org/abs/1503.05671} {\enquote {\bibinfo {title} {Optimizing neural networks with kronecker-factored approximate curvature},}\ } (\bibinfo {year} {2020}),\ \Eprint {http://arxiv.org/abs/1503.05671} {arXiv:1503.05671 [cs.LG]} \BibitemShut {NoStop}%
\bibitem [{\citenamefont {Martens}(2020)}]{NG-martens}%
  \BibitemOpen
  \bibfield  {author} {\bibinfo {author} {\bibfnamefont {James}\ \bibnamefont {Martens}},\ }\bibfield  {title} {\enquote {\bibinfo {title} {New insights and perspectives on the natural gradient method},}\ }\href {http://jmlr.org/papers/v21/17-678.html} {\bibfield  {journal} {\bibinfo  {journal} {Journal of Machine Learning Research}\ }\textbf {\bibinfo {volume} {21}},\ \bibinfo {pages} {1--76} (\bibinfo {year} {2020})}\BibitemShut {NoStop}%
\bibitem [{\citenamefont {Mor{\'e}}(2006)}]{more2006levenberg}%
  \BibitemOpen
  \bibfield  {author} {\bibinfo {author} {\bibfnamefont {Jorge~J}\ \bibnamefont {Mor{\'e}}},\ }\bibfield  {title} {\enquote {\bibinfo {title} {The levenberg-marquardt algorithm: implementation and theory},}\ }in\ \href@noop {} {\emph {\bibinfo {booktitle} {Numerical analysis: proceedings of the biennial Conference held at Dundee, June 28--July 1, 1977}}}\ (\bibinfo {organization} {Springer},\ \bibinfo {year} {2006})\ pp.\ \bibinfo {pages} {105--116}\BibitemShut {NoStop}%
\bibitem [{\citenamefont {Wright}\ \emph {et~al.}(1999)\citenamefont {Wright}, \citenamefont {Nocedal} \emph {et~al.}}]{numerical.opt.textbook}%
  \BibitemOpen
  \bibfield  {author} {\bibinfo {author} {\bibfnamefont {Stephen}\ \bibnamefont {Wright}}, \bibinfo {author} {\bibfnamefont {Jorge}\ \bibnamefont {Nocedal}},  \emph {et~al.},\ }\bibfield  {title} {\enquote {\bibinfo {title} {Numerical optimization},}\ }\href@noop {} {\bibfield  {journal} {\bibinfo  {journal} {Springer Science}\ }\textbf {\bibinfo {volume} {35}},\ \bibinfo {pages} {7} (\bibinfo {year} {1999})}\BibitemShut {NoStop}%
\bibitem [{\citenamefont {Drissi}\ \emph {et~al.}(2024)\citenamefont {Drissi}, \citenamefont {Keeble}, \citenamefont {Rozal{\'e}n~Sarmiento},\ and\ \citenamefont {Rios}}]{nqs.kfac}%
  \BibitemOpen
  \bibfield  {author} {\bibinfo {author} {\bibfnamefont {M}~\bibnamefont {Drissi}}, \bibinfo {author} {\bibfnamefont {JW~T}\ \bibnamefont {Keeble}}, \bibinfo {author} {\bibfnamefont {J}~\bibnamefont {Rozal{\'e}n~Sarmiento}}, \ and\ \bibinfo {author} {\bibfnamefont {A}~\bibnamefont {Rios}},\ }\bibfield  {title} {\enquote {\bibinfo {title} {Second-order optimization strategies for neural network quantum states},}\ }\href@noop {} {\bibfield  {journal} {\bibinfo  {journal} {Philosophical Transactions A}\ }\textbf {\bibinfo {volume} {382}},\ \bibinfo {pages} {20240057} (\bibinfo {year} {2024})}\BibitemShut {NoStop}%
\end{thebibliography}%
\appendix

\section{Gated Recurrent Unit Cells}
\label{app:GRU}
To mitigate the vanishing gradient problem present in Vanilla RNNs, GRUs introduce three additional vectors $\tilde{\bm{h}}_i,\bm{r}_i, \bm{u}_i$ to capture memory in addition to the hidden state $\bm{h}_i$. Here we present the GRU architecture introduced in Ref.~\cite{GRU}. Its recursion update rules are as follows:
\begin{align}
    \bm{u}_i &= \operatorname{Sigmoid}(W_u[\bm{h}_{i-1};\bm{\sigma}_{i-1}]+\bm{b}_u),\\
    \bm{r}_i &=\operatorname{Sigmoid}(W_r[\bm{h}_{i-1};\bm{\sigma}_{i-1}]+\bm{b}_r),\\
    \bm{\tilde{h}}_i &= \tanh\left( \bm{r}_i \odot (W_{\rm h} \bm{h}_{i-1} + \bm{b}_{\rm h}) + W_{\rm in} \bm{\sigma}_{i-1} + \bm{b}_{\rm in} \right),\\
    \bm{h}_i &= (1-\bm{u}_i)\odot\bm{h}_{i-1} +\bm{u}_i\odot \tilde{\bm{h}}_i.
\end{align}
Here $\odot$ is the Hadamard product, which multiplies two vectors element-wise. The hidden state $\bm{h}_i$ is updated based on the interpolation of the previous hidden state $\bm{h}_{i-1}$ and the candidate hidden state $\tilde{\bm{h}}_i$. The update gate $\bm{u}_i$ controls this interpolation based on the importance value of $\bm{\sigma}_{i-1}$. The reset gate $\bm{r}_i$ modifies the candidate hidden state $\tilde{\bm{h}}_i$ by setting certain components of $\bm{h}_{i-1}$ to zero when the corresponding components of $\bm{r}_i$ are close to zero. The reset gate allows the GRU to `forget' irrelevant information, which helps mitigate the vanishing gradient problem.

The two-dimensional version of the GRU wave function is implemented as follows:
\begin{align}
    \bm{u}_{i,j} &= \operatorname{Sigmoid}(W_u[\bm{h}'_{i,j};\bm{\sigma}'_{i,j}]+\bm{b}_u),\\
    \bm{r}_{i,j} &=\operatorname{Sigmoid}(W_r[\bm{h}'_{i,j};\bm{\sigma}'_{i,j}]+\bm{b}_r),\\
    \bm{\tilde{h}}_{i,j} &= \tanh\left( \bm{r}_{i,j} \odot (W_{\rm h} \bm{h}'_{i,j} + \bm{b}_{\rm h}) + W_{\rm in} \bm{\sigma}'_{i,j} + \bm{b}_{\rm in} \right), \\
    \bm{h}_{i,j} &= (1-\bm{u}_{i,j})\odot\bm{h}'_{i,j} +\bm{u}_{i,j}\odot \tilde{\bm{h}}_{i,j}\textcolor{purple}{,}
\end{align}
where
\begin{eqnarray}
    \bm{h}'_{i,j} &= U_1 \bm{h}_{i-(-1)^j,j} + U_2 \bm{h}_{i,j-1}, \\
    \bm{\sigma}'_{i,j} &= [\bm{\sigma}_{i-(-1)^j,j};  \bm{\sigma}_{i,j-1}].
\end{eqnarray}
Here $U_1$ and $U_2$ are additional trainable weights.

\section{minSR with Momentum}
\label{app:momentum}
We add momentum to our minSR scheme, which has previously been observed to increase performance in SR~\cite{SPRING}. We can write our SR parameter updates as $\bm{\theta}_{t+1}=\bm{\theta}_{t}-\delta\bm{\theta}_t$. The corresponding momentum-based parameter update is given by
\begin{align}
    \bm{m}_{t+1} &= \mu \bm{m}_{t}+(1-\mu)\delta \bm{\theta}_t ,\label{eq:momentum}\\
    \bm{\theta}_{t+1} &= \bm{\theta}_t-\bm{m}_{t+1}.\notag
\end{align}
Here $\mu\in [0,1)$ is the momentum hyperparameter, with  $\mu=0$  corresponding to regular SGD. $\bm{m}$ is the momentum vector which stores the previous history of updates. We present the momentum values for training in  App.~\ref{app:hyperparameter}. We attempted using the SPRING algorithm~\cite{SPRING}, and found that it performs slightly worse than minSR.

\section{Derivation of stochastic-reconfiguration ill-conditioning in RNN wave functions}
\label{app:illcond_detailed}

The SR matrix $S=\bar{\mathcal O}\bar{\mathcal O}^{\dagger}$ of an autoregressive RNN wave function is ill-conditioned from both ends of its spectrum, with exact gauge zero modes at the bottom and recurrently amplified modes scaling as $N^3$ (marginal) or exponentially in $N$ (expanding) at the top.
This appendix proves Proposition~\ref{prop:gauge_zero_modes} and Proposition~\ref{prop:spectral_divergence}, stated in Sec.~\ref{subsec:rnn_fisher_spectrum} of the main text. Subsection~C.1 fixes the model and exhibits its hidden-coordinate gauge map. Subsection~C.2 proves Proposition~\ref{prop:gauge_zero_modes}. Subsections~C.3 and~C.4 rewrite the SR matrix $S$ in terms of the logits $\mathbf z_i$ rather than the raw network parameters. This logit-coordinate form is the starting point for the eigenvalue bounds proved in Subsections~C.5 and~C.6. Subsections~C.5 and~C.6 prove Proposition~\ref{prop:spectral_divergence}. The derivation is deliberately granular: each hypothesis is made explicit at the step where it is used.

\subsection*{C.1 Setup and hidden-coordinate gauge freedom}

We first fix the model and exhibit its parametrization redundancy. The intuition is geometric: the RNN cell communicates with the outputs only through the logits, so any change of the internal hidden coordinates that the logits cannot see is physically invisible and must leave the state invariant.

We work with a real, positive, autoregressive wave function whose Born probability factorizes over sites,
\begin{equation}
|\Psi_{\boldsymbol{\theta}}(\boldsymbol{\sigma})|^2
=
P_{\boldsymbol{\theta}}(\boldsymbol{\sigma})
=
\prod_{i=1}^{N}
p_{\boldsymbol{\theta}}(\sigma_i\,|\,\sigma_{<i}) ,
\label{eq:born_factorization_H}
\end{equation}
where $\boldsymbol{\sigma}=(\sigma_1,\dots,\sigma_N)$ is a spin configuration, $\sigma_{<i}=(\sigma_1,\dots,\sigma_{i-1})$ is its history, and $\boldsymbol{\theta}$ collects all variational parameters. As we assume $\Psi_{\boldsymbol{\theta}}>0$, taking the logarithm of Eq.~\eqref{eq:born_factorization_H} and using $\log\Psi=\tfrac12\log P$ gives
\begin{equation}
\log \Psi_{\boldsymbol{\theta}}(\boldsymbol{\sigma})
=
\frac{1}{2}\log P_{\boldsymbol{\theta}}(\boldsymbol{\sigma})
=
\frac{1}{2}
\sum_{i=1}^{N}
\log p_{\boldsymbol{\theta}}(\sigma_i\,|\,\sigma_{<i}) ,
\label{eq:logpsi_sum_H}
\end{equation}
where the second equality is the logarithm of the product in Eq.~\eqref{eq:born_factorization_H}.

The conditionals are produced by a linear recurrent cell. At site $i$ the hidden state $\mathbf h_i\in\mathbb R^{d_h}$ and logit vector $\mathbf z_i\in\mathbb R^{q}$ are
\begin{equation}
\mathbf h_i
=
W\mathbf h_{i-1}+V\boldsymbol{\sigma}_i+\mathbf b,
\qquad
\mathbf z_i
=
U\mathbf h_{i-1}+\mathbf c ,
\label{eq:linear_rnn_H}
\end{equation}
where $\boldsymbol{\sigma}_i\in\mathbb R^{d_\sigma}$ is the one-hot encoding of the local spin $\sigma_i$ over its $d_\sigma$ possible values, $W\in\mathbb R^{d_h\times d_h}$ is the recurrent matrix, $V\in\mathbb R^{d_h\times d_{\sigma}}$ the input matrix, $U\in\mathbb R^{d_\sigma\times d_h}$ the readout matrix, and $\mathbf b\in\mathbb R^{d_h}$, $\mathbf c\in\mathbb R^{d_\sigma}$ the biases. Throughout, we take the standing convention that the initial hidden state is fixed at $\mathbf h_0=\mathbf 0$; this both gives the recursion a starting point and, as we use below, makes it invariant under the hidden-coordinate map. The conditional distribution is the softmax of the logits,
\begin{equation}
\mathbf p_i
=
\operatorname{softmax}(\mathbf z_i),
\qquad
p_{\boldsymbol{\theta}}(\sigma_i\,|\,\sigma_{<i})
=
\mathbf p_i^{\mathsf T}\boldsymbol{\sigma}_i ,
\label{eq:softmax_def_H}
\end{equation}
where the last equality selects the component of $\mathbf p_i$ indexed by the realized one-hot $\boldsymbol{\sigma}_i$.

This parametrization carries a continuous hidden-coordinate gauge symmetry. For any invertible $A\in GL(d_h)$, apply the change of internal basis $\mathbf h_i\mapsto A\mathbf h_i$ (consistently for all $i\ge 0$, so that in particular $\mathbf h_0=\mathbf 0\mapsto A\mathbf 0=\mathbf 0$) together with
\begin{equation}
\begin{gathered}
W\mapsto AWA^{-1},
\qquad
V\mapsto AV,
\qquad
\mathbf b\mapsto A\mathbf b,
\\[2pt]
U\mapsto UA^{-1},
\qquad
\mathbf c\mapsto \mathbf c .
\end{gathered}
\label{eq:gauge_map_H}
\end{equation}
We check by induction on $i$ that the recursion in Eq.~\eqref{eq:linear_rnn_H} maps $\mathbf h_i\mapsto A\mathbf h_i$. \emph{Base case} $i=0$. The fixed initial state obeys $\mathbf h_0=\mathbf 0\mapsto A\mathbf 0=\mathbf 0$, so the claim holds trivially. This is where the choice $\mathbf h_0=\mathbf 0$ is needed. \emph{Inductive step}. Assume $\mathbf h_{i-1}\mapsto A\mathbf h_{i-1}$. Substituting the transformed quantities into the hidden update gives
\begin{align}
\mathbf h_i'
&=
(AWA^{-1})(A\mathbf h_{i-1})+(AV)\boldsymbol{\sigma}_i+A\mathbf b
\nonumber\\
&=
A\big(W\mathbf h_{i-1}+V\boldsymbol{\sigma}_i+\mathbf b\big)
=
A\mathbf h_i ,
\label{eq:gauge_hidden_H}
\end{align}
where the first line inserts the maps of Eq.~\eqref{eq:gauge_map_H} together with the induction hypothesis. The second line uses $A^{-1}A=\mathbb I$ and factors out $A$, which closes the induction. The logits are then strictly unchanged (invariant) because $A^{-1}A=\mathbb I$ cancels the change of internal basis.
\begin{equation}
\mathbf z_i
\mapsto
(UA^{-1})(A\mathbf h_{i-1})+\mathbf c
=
U\mathbf h_{i-1}+\mathbf c
=
\mathbf z_i .
\label{eq:gauge_logit_H}
\end{equation}

The invariance now propagates forward. First, the softmax output $\mathbf p_i$ of Eq.~\eqref{eq:softmax_def_H} is unchanged, since it sees the parameters only through $\mathbf z_i$. Second, so is the product $P_{\boldsymbol{\theta}}$ of Eq.~\eqref{eq:born_factorization_H}, and hence $\Psi_{\boldsymbol{\theta}}$ itself. Moving along the orbit generated by $A\in GL(d_h)$ therefore changes the parameters but leaves the physical state identical.

\subsection*{C.2 Exact Fisher zero modes from the gauge orbit}

\begin{proof}[Proof of Proposition~\ref{prop:gauge_zero_modes}]
The gauge orbit has infinitesimal generators: the directions in which the
parameters move when the change of basis is first switched on. We now show that these are \emph{exact} zero modes of $S$, not simply approximate ones. Along these directions the log-amplitude does not move for any single configuration, so the log-derivative vector that builds $S$ vanishes identically.

So far we have considered a finite change of basis $A$. To identify
\emph{directions} in parameter space we need infinitesimal ones, so we take $A$ to be a small departure from doing nothing,
\begin{equation}
A = \mathbb I + t\,X ,
\end{equation}
where $X$ is an arbitrary fixed $d_h\times d_h$ matrix that selects which change of basis we perform, and the small number $t$ controls its size. Setting $t=0$ gives $A=\mathbb I$, i.e.\ no change at all.

As $t$ varies, Eq.~\eqref{eq:gauge_map_H} turns this into a smooth path $\boldsymbol{\theta}(t)$ through parameter space. Every point on this path describes the \emph{same} wave function, since each value of $t$ corresponds to a valid change of basis. 


To find the explicit components of $d\boldsymbol{\theta}$, we substitute $A=\mathbb I+tX$ into the transformation rules of Eq.~\eqref{eq:gauge_map_H} and keep only the terms proportional to $t$. We also need the inverse to the same order,
\begin{equation}
A^{-1}=\mathbb I-tX+\mathcal O(t^2) ,
\end{equation}
which follows from
$(\mathbb I+tX)(\mathbb I-tX)=\mathbb I-t^2X^2=\mathbb I+\mathcal O(t^2)$.

Take the recurrent matrix as a worked example:
\begin{align}
W\;\mapsto\;AWA^{-1}
&=(\mathbb I+tX)\,W\,(\mathbb I-tX)+\mathcal O(t^2)
\nonumber\\
&=W+t\,(XW-WX)+\mathcal O(t^2) .
\end{align}
The two terms carry opposite signs, one coming from each side of $W$. Their difference is the \emph{commutator} of $X$ and $W$,
\begin{equation}
[X,W]\equiv XW-WX ,
\label{eq:commutator_H}
\end{equation}
also known as the Lie bracket. It measures the extent to which two matrices fail to commute, and vanishes precisely when $XW=WX$.

The commutator appears only for $W$. The reason is that $W$ is the only matrix that takes a hidden state and returns another hidden state, so both of its sides feel the change of basis: $A$ multiplies it from the left and $A^{-1}$ from the right, $W\mapsto AWA^{-1}$. Each side then contributes one term, and the two carry opposite signs, so their difference is taken. Every other parameter meets the hidden space on one side only, for example $V\mapsto AV$, so it contributes a single term and no difference is formed. Repeating the same substitution for each remaining parameter gives
\begin{equation}
\begin{gathered}
dW=[X,W]t,
\qquad
dV=XVt,
\qquad
d\mathbf b=X\mathbf bt,
\\[2pt]
dU=-UXt,
\qquad
d\mathbf c=\mathbf 0.
\end{gathered}
\label{eq:gauge_components_H}
\end{equation}
Note that $d\mathbf c=\mathbf 0$ because $\mathbf c$ never appears with an $A$ at all.

These are our candidate zero directions, and we can indeed confirm them immediately. Eq.~\eqref{eq:gauge_logit_H} showed that the logits are \emph{exactly} unchanged for every $A$, hence for every value of $t$. A quantity that never changes has zero rate of change, so
\begin{equation}
d\mathbf z_i=\mathbf 0
\qquad\text{for all } i ,
\label{eq:dz_zero_H}
\end{equation}
since $\mathbf z_i$ is constant along the whole orbit. Because $\log\Psi_{\boldsymbol{\theta}}$ depends on $\boldsymbol{\theta}$ only through the logits $\{\mathbf z_i\}$ (Eqs.~\eqref{eq:logpsi_sum_H}--\eqref{eq:softmax_def_H}), the chain rule and Eq.~\eqref{eq:dz_zero_H} yield
\begin{equation}
d\log \Psi_{\boldsymbol{\theta}}(\boldsymbol{\sigma})
=
\sum_{i=1}^{N}
\frac{\partial \log\Psi_{\boldsymbol{\theta}}}{\partial \mathbf z_i}\!\cdot d\mathbf z_i
=
0
\qquad\text{for every } \boldsymbol{\sigma} .
\label{eq:dlogpsi_zero_H}
\end{equation}
We now feed this result into $S$. Recall from the main text how $S$ is
assembled from a sample of $N_s$ configurations $\boldsymbol{\sigma}^{(1)},\dots,\boldsymbol{\sigma}^{(N_s)}$. The
(unnormalized) log-derivative matrix has entries
\begin{equation}
\mathcal O_{ks}
=
\tfrac{1}{\sqrt{N_s}}\,\partial_{\theta_k}\log\Psi(\boldsymbol{\sigma}^{(s)}) ,
\end{equation}
where $k$ labels the parameter and $s$ labels the sample. Subtracting from each row its own sample mean gives the centered matrix $\bar{\mathcal O}$,
\begin{equation}
\bar{\mathcal O}_{ks}
=
\mathcal O_{ks}-\tfrac{1}{N_s}\sum_{s'}\mathcal O_{ks'} ,
\end{equation}
and $S=Q=\bar{\mathcal O}\bar{\mathcal O}^{\dagger}$ for real-valued wavefunctions. To test whether a gauge direction $d \bm{\theta}$ is a zero mode, we evaluate the quadratic
form $d\boldsymbol{\theta}^{\mathsf T} S\, d\boldsymbol{\theta}$, that is, we
contract $S$ with $d\boldsymbol{\theta}$ on both of its indices, multiplying
it by $d\boldsymbol{\theta}$ on the left and on the right,
\begin{equation}
d\boldsymbol{\theta}^{\mathsf T} S\, d\boldsymbol{\theta}
=
\big\|\,\bar{\mathcal O}^{\dagger} d\boldsymbol{\theta}\,\big\|_2^2
=
\sum_{s=1}^{N_s}
\Big|\big(\bar{\mathcal O}^{\dagger} d\boldsymbol{\theta}\big)_s\Big|^2 ,
\label{eq:zero_mode_quadform_H}
\end{equation}
where the first equality writes the quadratic form as a norm using $S=\bar{\mathcal O}\bar{\mathcal O}^{\dagger}$. Each entry $\big(\bar{\mathcal O}^{\dagger} d\boldsymbol{\theta}\big)_s$ is the centered directional log-derivative $\tfrac{1}{\sqrt{N_s}}\big(d\log\Psi(\boldsymbol{\sigma}^{(s)})-\overline{d\log\Psi}\,\big)$, which vanishes because Eq.~\eqref{eq:dlogpsi_zero_H} holds configuration by configuration. Hence
\begin{equation}
d\boldsymbol{\theta}^{\mathsf T} S\, d\boldsymbol{\theta}=0 ,
\label{eq:exact_zero_mode_H}
\end{equation}
for \emph{any} sample. Nothing in the argument required many samples,
since each square vanished on its own. The gauge directions are therefore exact
zero modes of $S$ in two senses. First, in the infinite-sample limit. Second,
for every finite-sample estimate $\bar{\mathcal O}\bar{\mathcal O}^{\dagger}$. These are the Fisher-information singularities: null directions of $S$ along which the parameters change but the wave function does not. It remains to show that $d\boldsymbol{\theta}$ lies in $\ker S$, in both the finite-sample and the population case. For a finite sample, Eq.~\eqref{eq:zero_mode_quadform_H} makes every entry of $\bar{\mathcal O}^{\dagger} d\boldsymbol{\theta}$ vanish, so $S\,d\boldsymbol{\theta}=\bar{\mathcal O}\big(\bar{\mathcal O}^{\dagger} d\boldsymbol{\theta}\big)=\mathbf 0$. For the population matrix, $S$ is positive semidefinite, and its quadratic form vanishes along $d\boldsymbol{\theta}$; since a positive semidefinite matrix annihilates any vector on which its quadratic form is zero, $S\,d\boldsymbol{\theta}=\mathbf 0$ there too. The generator $X$ was arbitrary, so every gauge direction lies in $\ker S$: the tangent space of the $GL(d_h)$ gauge orbit is contained in $\ker S$.
\end{proof}

\subsection*{C.3 The Fisher quadratic form in logit coordinates}

In these two subsections, we derive the logit-coordinate form of the population SR matrix, Eq.~\eqref{eq:fisher_logit_metric_H}, which the proof of Proposition~\ref{prop:spectral_divergence} (Secs.~C.5 and C.6) takes as its starting point. Having established that the gauge directions are exact zero modes, we now turn to the remaining directions. We express the Fisher form in logit coordinates, where its structure is most transparent. The autoregressive factorization then reduces it to a \emph{sum of per-site contributions}: the cross terms between different sites vanish because the score of each conditional distribution has zero mean.

It is easier to compute the Fisher form of the probability distribution first and convert to the wave function at the end. We therefore define the \emph{classical} Fisher quadratic form,
\begin{equation}
\mathcal S_P(d\boldsymbol{\theta},d\boldsymbol{\theta})
=
\Big\langle
\big[
d\log P_{\boldsymbol{\theta}}(\boldsymbol{\sigma})
\big]^2
\Big\rangle ,
\label{eq:classical_fisher_H}
\end{equation}
where $\langle\cdot\rangle$ denotes the expectation over configurations
$\boldsymbol{\sigma}$ drawn from $P_{\boldsymbol{\theta}}$, and
$d\log P_{\boldsymbol{\theta}}$ is the change in $\log P_{\boldsymbol{\theta}}$
produced by the parameter displacement $d\boldsymbol{\theta}$. It is a
\emph{quadratic form} because $d\boldsymbol{\theta}$ enters twice, through the square; both arguments are therefore written out. It is \emph{classical} because it is built from the probability distribution
$P_{\boldsymbol{\theta}}$ rather than from the amplitude
$\Psi_{\boldsymbol{\theta}}$.

To evaluate it we need $d\log P_{\boldsymbol{\theta}}$. Multiplying
Eq.~\eqref{eq:logpsi_sum_H} by two removes the factor of $1/2$ and leaves
\begin{equation}
\log P_{\boldsymbol{\theta}}(\boldsymbol{\sigma})
=
\sum_{i=1}^{N}
\log p_{\boldsymbol{\theta}}(\sigma_i\,|\,\sigma_{<i}) ,
\label{eq:logP_sum_H}
\end{equation}
and differentiating this sum term by term gives
\begin{equation}
d\log P_{\boldsymbol{\theta}}(\boldsymbol{\sigma})
=
\sum_{i=1}^{N}
d\log p_{\boldsymbol{\theta}}(\sigma_i\,|\,\sigma_{<i}) ,
\label{eq:dlogP_sum_H}
\end{equation}
To build $\mathcal S_P$ we must square this sum. It is convenient to
abbreviate the site-$i$ score as
\begin{equation}
u_i \equiv d\log p_{\boldsymbol{\theta}}(\sigma_i\,|\,\sigma_{<i}) ,
\label{eq:site_score_H}
\end{equation}
so that $d\log P_{\boldsymbol{\theta}}=\sum_i u_i$. The square of a sum
is a double sum,
\begin{equation}
\Big[\sum_{i=1}^{N}u_i\Big]^{2}
=
\sum_{i=1}^{N}\sum_{j=1}^{N}u_i u_j ,
\end{equation}
and taking the expectation term by term gives
\begin{equation}
\mathcal S_P
=
\sum_{i,j=1}^{N}
\big\langle u_i u_j \big\rangle ,
\label{eq:cross_terms_H}
\end{equation}
which contains diagonal terms ($i=j$) and cross terms ($i\neq j$). We now show that every cross term vanishes, leaving only the diagonal.

Consider a cross term with $i<j$. The key observation is that $u_i$ depends only on the sites $\sigma_1,\dots,\sigma_i$. Since $i<j$, every one of
those sites already belongs to the history
$\sigma_{<j}=(\sigma_1,\dots,\sigma_{j-1})$. Once that history is known, $u_i$ is therefore a fixed number, not a random one.

This lets us average in two stages: first over $\sigma_j$ with the history held fixed, then over the history itself. This regrouping is the law of total expectation,
\begin{equation}
\big\langle u_i u_j \big\rangle
=
\Big\langle
\big\langle u_i u_j \,\big|\, \sigma_{<j}\big\rangle
\Big\rangle .
\end{equation}
In the inner average, the history is fixed, so $u_i$ is a constant and can be taken outside it,
\begin{equation}
\big\langle u_i u_j \big\rangle
=
\Big\langle
u_i\,
\big\langle u_j \,\big|\, \sigma_{<j}\big\rangle
\Big\rangle .
\label{eq:tower_H}
\end{equation}
Everything now depends on the inner conditional average $\langle u_j \,|\, \sigma_{<j}\rangle$, which vanishes because probabilities are normalized,
\begin{align}
\big\langle
d\log p_{\boldsymbol{\theta}}(\sigma_j|\sigma_{<j})
\,\big|\, \sigma_{<j}
\big\rangle
&=
\sum_{\sigma_j}
p_{\boldsymbol{\theta}}(\sigma_j|\sigma_{<j})\,
d\log p_{\boldsymbol{\theta}}(\sigma_j|\sigma_{<j})
\nonumber\\
&=
\sum_{\sigma_j}
d\,p_{\boldsymbol{\theta}}(\sigma_j|\sigma_{<j})
\nonumber\\
&=
d\!\sum_{\sigma_j}
p_{\boldsymbol{\theta}}(\sigma_j|\sigma_{<j})
\nonumber\\
&=
d(1)=0 ,
\label{eq:score_zero_mean_H}
\end{align}
Here the first line is the definition of the conditional average; the
second uses the identity $p\,d\log p=dp$; the third interchanges the
differential and the finite sum over $\sigma_j$; and the last uses the fact
that these conditional probabilities sum to one for every history, so that
their differential vanishes.

Substituting Eq.~\eqref{eq:score_zero_mean_H} into
Eq.~\eqref{eq:tower_H} shows that every term with $i<j$ vanishes. The case
$i>j$ follows immediately: the summand in Eq.~\eqref{eq:cross_terms_H} is
symmetric under $i\leftrightarrow j$, so exchanging the two indices reduces it
to the case just treated. Every off-diagonal term therefore vanishes, and only the diagonal survives,
\begin{equation}
\mathcal S_P(d\boldsymbol{\theta},d\boldsymbol{\theta})
=
\sum_{i=1}^{N}
\Big\langle
\big[
d\log p_{\boldsymbol{\theta}}(\sigma_i|\sigma_{<i})
\big]^2
\Big\rangle .
\label{eq:fisher_diagonal_H}
\end{equation}

\subsection*{C.4 Softmax covariance}

We now evaluate each per-site term of Eq.~\eqref{eq:fisher_diagonal_H} and connect the classical Fisher form to the wave function SR matrix. Two facts do the work. First, the softmax score is a linear map of the logit displacement, so its second moment is the covariance of the one-hot variable. Second, the amplitude carries a factor $\tfrac12$ relative to the probability, and squaring turns that $\tfrac12$ into a $\tfrac14$.

We first derive the softmax score explicitly. Here $\mathbf z_i$ is the logit vector at site $i$, $\mathbf p_i=\operatorname{softmax}(\mathbf z_i)$ is the conditional probability vector, and $\boldsymbol{\sigma}_i$ is the one-hot encoding of the realized spin $\sigma_i$. For that realized one-hot, Eq.~\eqref{eq:softmax_def_H} reads
\begin{equation}
p_{\boldsymbol{\theta}}(\sigma_i\,|\,\sigma_{<i})
=
\mathbf p_i^{\mathsf T}\boldsymbol{\sigma}_i
=
\frac{e^{z_{i,\sigma_i}}}{\sum_{k}e^{z_{i,k}}} ,
\end{equation}
and taking its logarithm gives
\begin{equation}
\log p_{\boldsymbol{\theta}}(\sigma_i\,|\,\sigma_{<i})
=
\mathbf z_i^{\mathsf T}\boldsymbol{\sigma}_i
-
\log\!\sum_{k} e^{z_{i,k}} ,
\label{eq:logsoftmax_H}
\end{equation}
where the first term is the selected logit and the second is the log-partition function. Differentiating the log-partition term with respect to the logits,
\begin{equation}
d\log\!\sum_{k} e^{z_{i,k}}
=
\frac{\sum_k e^{z_{i,k}}\,dz_{i,k}}{\sum_{k'}e^{z_{i,k'}}}
=
\sum_k p_{i,k}\,dz_{i,k}
=
\mathbf p_i^{\mathsf T}d\mathbf z_i ,
\label{eq:dlogpartition_H}
\end{equation}
which uses the softmax normalization
$p_{i,k}=e^{z_{i,k}}/\sum_{k'}e^{z_{i,k'}}$. Now differentiate the first term of Eq.~\eqref{eq:logsoftmax_H}. Once $\sigma_i$ is realized, $\boldsymbol{\sigma}_i$ is a fixed vector, not a random one, so $d(\mathbf z_i^{\mathsf T}\boldsymbol{\sigma}_i)=\boldsymbol{\sigma}_i^{\mathsf T}d\mathbf z_i$.
Subtracting Eq.~\eqref{eq:dlogpartition_H} from this gives the score,
\begin{equation}
d\log p_{\boldsymbol{\theta}}(\sigma_i\,|\,\sigma_{<i})
=
(\boldsymbol{\sigma}_i-\mathbf p_i)^{\mathsf T}\,d\mathbf z_i ,
\label{eq:softmax_score_H}
\end{equation}
where $\boldsymbol{\sigma}_i$ is the realized one-hot and $\mathbf p_i$ its mean. Squaring and taking the conditional expectation over $\sigma_i$ at fixed history $\sigma_{<i}$ (so that $\mathbf p_i$ and $d\mathbf z_i$ are fixed),
\begin{align}
&\Big\langle
\big[d\log p_{\boldsymbol{\theta}}(\sigma_i|\sigma_{<i})\big]^{2}
\,\big|\,\sigma_{<i}
\Big\rangle
\nonumber\\
&\quad=
d\mathbf z_i^{\mathsf T}
\Big\langle
(\boldsymbol{\sigma}_i-\mathbf p_i)
(\boldsymbol{\sigma}_i-\mathbf p_i)^{\mathsf T}
\,\Big|\,\sigma_{<i}
\Big\rangle
d\mathbf z_i
\nonumber\\
&\quad=
d\mathbf z_i^{\mathsf T}\,C_i\,d\mathbf z_i ,
\label{eq:cond_second_moment_H}
\end{align}
where the first line pulls the deterministic $d\mathbf z_i$ out of the expectation, and the second identifies the conditional covariance $C_i$ of the one-hot variable. To evaluate $C_i$, expand the outer product and use linearity of the conditional expectation,
\begin{align}
C_i
&=
\big\langle
(\boldsymbol{\sigma}_i-\mathbf p_i)(\boldsymbol{\sigma}_i-\mathbf p_i)^{\mathsf T}
\,\big|\,\sigma_{<i}
\big\rangle
\nonumber\\
&=
\big\langle\boldsymbol{\sigma}_i\boldsymbol{\sigma}_i^{\mathsf T}\big\rangle
-\big\langle\boldsymbol{\sigma}_i\big\rangle\mathbf p_i^{\mathsf T}
-\mathbf p_i\big\langle\boldsymbol{\sigma}_i\big\rangle^{\mathsf T}
+\mathbf p_i\mathbf p_i^{\mathsf T}
\nonumber\\
&=
\operatorname{diag}(\mathbf p_i)
-\mathbf p_i\mathbf p_i^{\mathsf T}
-\mathbf p_i\mathbf p_i^{\mathsf T}
+\mathbf p_i\mathbf p_i^{\mathsf T}
\nonumber\\
&=
\operatorname{diag}(\mathbf p_i)-\mathbf p_i\mathbf p_i^{\mathsf T} ,
\label{eq:softmax_cov_H}
\end{align}
where the conditioning on $\sigma_{<i}$ is left implicit inside $\langle\cdot\rangle$. The third line uses two properties of a one-hot vector. First, its outer product with itself is the diagonal matrix built from its own entries,
$\boldsymbol{\sigma}_i\boldsymbol{\sigma}_i^{\mathsf T}=\operatorname{diag}(\boldsymbol{\sigma}_i)$,
so averaging gives
$\langle\boldsymbol{\sigma}_i\boldsymbol{\sigma}_i^{\mathsf T}\rangle=\operatorname{diag}(\mathbf p_i)$.
Second, its mean is the softmax vector itself,
$\langle\boldsymbol{\sigma}_i\rangle=\mathbf p_i$. The resulting matrix $C_i$
is symmetric and positive semidefinite, and is the Fisher matrix of the softmax logits.

Two steps connect $\mathcal S_P$ to the wave function SR matrix
$S=\bar{\mathcal O}\bar{\mathcal O}^{\dagger}$. The first is the relation
$\log\Psi_{\boldsymbol{\theta}}=\tfrac12\log P_{\boldsymbol{\theta}}$ between
the amplitude and the probability: differentiating and then squaring converts the factor $\tfrac12$ into $\tfrac14$,
\begin{equation}
\begin{gathered}
d\log\Psi=\tfrac12\,d\log P,
\\[2pt]
(d\log\Psi)^2=\tfrac14(d\log P)^2 .
\end{gathered}
\end{equation}
Second, the total score has zero mean,
\begin{align}
\big\langle d\log P_{\boldsymbol{\theta}}\big\rangle
&=
\sum_{\boldsymbol{\sigma}} P_{\boldsymbol{\theta}}(\boldsymbol{\sigma})\,
d\log P_{\boldsymbol{\theta}}(\boldsymbol{\sigma})
\nonumber\\
&=
d\!\sum_{\boldsymbol{\sigma}} P_{\boldsymbol{\theta}}(\boldsymbol{\sigma})
=
d(1)=0 ,
\label{eq:total_score_zero_H}
\end{align}
by the same normalization argument as Eq.~\eqref{eq:score_zero_mean_H} applied to the joint distribution.

The matrix $S=\bar{\mathcal O}\bar{\mathcal O}^{\dagger}$ is built from a finite sample and from \emph{centered} log-derivatives. The quadratic form it defines is therefore a sample variance,
\begin{equation}
d\boldsymbol{\theta}^{\mathsf T} S\, d\boldsymbol{\theta}
=
\frac{1}{N_s}\sum_{s=1}^{N_s}
\Big(d\log\Psi(\boldsymbol{\sigma}^{(s)})-\overline{d\log\Psi}\Big)^{2} ,
\end{equation}
where $d\log\Psi(\boldsymbol{\sigma}^{(s)})$ is the change in $\log\Psi$
along $d\boldsymbol{\theta}$ at sample $s$, and $\overline{d\log\Psi}$ is its average over the $N_s$ samples.

Two things happen as $N_s\to\infty$. First, the sample average converges to the population expectation, so $\tfrac{1}{N_s}\sum_s(\cdot)\to\langle\cdot
\rangle$. Second, the subtracted mean vanishes,
\begin{equation}
\overline{d\log\Psi}\;\longrightarrow\;
\langle d\log\Psi\rangle
=
\tfrac{1}{2}\big\langle d\log P_{\boldsymbol{\theta}}\big\rangle
=
0,
\end{equation}
by Eq.~\eqref{eq:total_score_zero_H}. The centering therefore drops out,
and the quadratic form reduces to
\begin{align}
d\boldsymbol{\theta}^{\mathsf T} S\, d\boldsymbol{\theta}
&\;\xrightarrow[N_s\to\infty]{}\;
\big\langle (d\log\Psi_{\boldsymbol{\theta}})^2\big\rangle
\nonumber\\
&=
\frac{1}{4}\,\big\langle (d\log P_{\boldsymbol{\theta}})^2\big\rangle
=
\frac{1}{4}\,\mathcal S_P(d\boldsymbol{\theta},d\boldsymbol{\theta}) ,
\label{eq:quarter_H}
\end{align}
where the first arrow uses the two limits just described (centering removed via Eq.~\eqref{eq:total_score_zero_H}), and the second equality uses $(d\log\Psi)^2=\tfrac14(d\log P)^2$. Inserting Eqs.~\eqref{eq:fisher_diagonal_H} and~\eqref{eq:cond_second_moment_H} (and averaging the conditional result over histories) gives the logit-space form of the population SR matrix,
\begin{equation}
d\boldsymbol{\theta}^{\mathsf T} S\, d\boldsymbol{\theta}
=
\frac{1}{4}
\sum_{i=1}^{N}
\big\langle
d\mathbf z_i^{\mathsf T}C_i\,d\mathbf z_i
\big\rangle .
\label{eq:fisher_logit_metric_H}
\end{equation}
Introducing the logit Jacobian $J_i\equiv\partial\mathbf z_i/\partial\boldsymbol{\theta}$, so that $d\mathbf z_i=J_i\,d\boldsymbol{\theta}$, and stripping the arbitrary displacement $d\boldsymbol{\theta}$ from both sides,
\begin{equation}
S
=
\frac{1}{4}
\sum_{i=1}^{N}
\big\langle
J_i^{\mathsf T}C_iJ_i
\big\rangle .
\label{eq:fisher_jacobian_form_H}
\end{equation}
Eq.~\eqref{eq:fisher_logit_metric_H} makes both ends of the
spectrum visible. Gauge directions leave the logits unchanged,
$d\mathbf z_i=\mathbf 0$ (Sec.~C.2), and therefore have zero Fisher norm. Directions that produce a large logit response, and with nonvanishing softmax covariance, give large Fisher eigenvalues, which we bound in the next subsection.

\subsection*{C.5 A lower bound on the largest eigenvalue}

\begin{proof}[Proof of Proposition~\ref{prop:spectral_divergence}]
To reach the top of the spectrum, we do not need to diagonalize $S$. It is enough to probe a single, carefully chosen direction. By the Rayleigh quotient, every diagonal element of $S$ in an orthonormal parameter basis is a lower bound on $\lambda_{\max}(S)$. It therefore suffices to exhibit one recurrent-bias direction whose diagonal element is large.

Throughout this subsection, $S$ denotes the population Fisher matrix of Eq.~\eqref{eq:fisher_jacobian_form_H}, that is, the $N_s\to\infty$ limit of Eq.~\eqref{eq:quarter_H}. Proposition~\ref{prop:spectral_divergence} is therefore a population statement. A finite-sample estimate fluctuates around this population matrix; obtaining a finite-sample analogue of the lower bound would require additional assumptions.

We take the recurrent matrix to be diagonal,
\begin{equation}
W=\operatorname{diag}(r_1,\dots,r_{d_h}),
\label{eq:diag_W_H}
\end{equation}
so that the hidden coordinates no longer mix. Writing the hidden update of Eq.~\eqref{eq:linear_rnn_H} component by component,
\begin{equation}
(\mathbf h_i)_a = r_a\,(\mathbf h_{i-1})_a + (V\boldsymbol{\sigma}_i)_a + b_a ,
\end{equation}
each coordinate evolves on its own. We call the $a$-th coordinate \emph{hidden mode $a$}, and $r_a$ its \emph{multiplier}: a perturbation of that mode is rescaled by $r_a$ at every site, so after $k$ sites it carries a factor $r_a^{\,k}$. This single number controls how fast the response grows.
We then perturb only the recurrent bias along one such mode,
\begin{equation}
d\boldsymbol{\theta}=d\mathbf b=\delta b_a\,\mathbf e_a ,
\label{eq:bias_perturbation_H}
\end{equation}
where $\mathbf e_a$ is the $a$-th standard basis vector of $\mathbb R^{d_h}$. We now linearize the hidden update of Eq.~\eqref{eq:linear_rnn_H}. Under a variation of the bias alone, both $V\boldsymbol{\sigma}_i$ and $W$ are held fixed, so the response obeys $d\mathbf h_i=W\,d\mathbf h_{i-1}+d\mathbf b$. The initial hidden state is fixed at $\mathbf h_0=\mathbf 0$, hence $d\mathbf h_0=\mathbf 0$. Unrolling the recursion from this starting point gives
\begin{equation}
d\mathbf h_{i-1}
=
\sum_{k=0}^{i-2} W^{k}\,d\mathbf b ,
\label{eq:unroll_H}
\end{equation}
a geometric sum in which the term $W^{k}\,d\mathbf b$ is the bias
variation after it has propagated through $k$ sites.

We now insert this response into the logit map
$\mathbf z_i=U\mathbf h_{i-1}+\mathbf c$. Because $W$ is diagonal, Eq.~\eqref{eq:diag_W_H} gives $W^{k}\mathbf e_a=r_a^{k}\mathbf e_a$, so every term of the sum stays on mode $a$, and the logit response reduces to
\begin{equation}
\begin{gathered}
d\mathbf z_i
=
U\,d\mathbf h_{i-1}
=
\mathbf u_a\,m_{i-1}(r_a)\,\delta b_a,
\\[2pt]
\mathbf u_a\equiv U\mathbf e_a ,
\end{gathered}
\label{eq:dz_mode_H}
\end{equation}
where $\mathbf u_a$ is the readout of hidden mode $a$, that is, the logit direction this mode excites. The geometric factor collects the powers of $r_a$ picked up along the way,
\begin{equation}
m_{i-1}(r_a)
=
\sum_{k=0}^{i-2} r_a^{\,k} .
\label{eq:geometric_factor_H}
\end{equation}
Because the displacement of Eq.~\eqref{eq:bias_perturbation_H} has a
single nonzero component, the quadratic form reduces to one diagonal element of $S$,
\begin{equation}
d\boldsymbol{\theta}^{\mathsf T} S\, d\boldsymbol{\theta}
=
|\delta b_a|^2\,S_{b_a b_a} ,
\label{eq:oneD_quadform_H}
\end{equation}
We identify that element by evaluating the right-hand side of Eq.~\eqref{eq:fisher_logit_metric_H} on the logit response of Eq.~\eqref{eq:dz_mode_H}. Both sides then carry the same factor $|\delta b_a|^2$, which cancels and leaves
\begin{equation}
\begin{gathered}
S_{b_a b_a}
=
\frac{1}{4}
\sum_{i=1}^{N}
|m_{i-1}(r_a)|^2\,\chi_{a,i},
\\[2pt]
\chi_{a,i}
\equiv
\big\langle
\mathbf u_a^{\mathsf T}C_i\mathbf u_a
\big\rangle ,
\end{gathered}
\label{eq:Sbaba_H}
\end{equation}
Here $\chi_{a,i}$ is the softmax covariance $C_i$ of Eq.~\eqref{eq:cond_second_moment_H}, projected onto the readout direction $\mathbf u_a$ and averaged over histories. It is nonnegative, $\chi_{a,i}\ge0$, because $C_i$ is positive semidefinite.

We can now apply the Rayleigh-quotient argument stated at the start of this subsection: for a symmetric $S$, the largest eigenvalue is at least as large as any diagonal element in an orthonormal basis. Let $\mathbf e_{b_a}$ denote the unit parameter-basis vector along the coordinate $b_a$; testing $S$ along this direction gives
\begin{align}
\lambda_{\max}(S)
&\ge
\mathbf e_{b_a}^{\mathsf T} S\,\mathbf e_{b_a}
=
S_{b_a b_a}
\nonumber\\
&=
\frac{1}{4}
\sum_{i=1}^{N}
|m_{i-1}(r_a)|^2\,\chi_{a,i} .
\label{eq:large_eig_bound_H}
\end{align}

\subsection*{C.6 Asymptotic scaling of the bound}

We now determine how the bound of Eq.~\eqref{eq:large_eig_bound_H} grows with the system size $N$. The growth is controlled entirely by the geometric factor $m_{i-1}(r_a)$, and we treat two regimes. A marginally stable recurrent mode already makes $\lambda_{\max}(S)$ diverge polynomially in $N$, and an expanding mode makes it diverge exponentially.

\emph{Marginal mode, $r_a=1$.} Every term of the geometric factor of
Eq.~\eqref{eq:geometric_factor_H} equals $1^{k}=1$, so the sum counts its $i-1$ terms,
\begin{equation}
m_{i-1}(1)=\sum_{k=0}^{i-2}1=i-1 ,
\label{eq:marginal_m_H}
\end{equation}
and Eq.~\eqref{eq:Sbaba_H} becomes
\begin{equation}
S_{b_a b_a}
=
\frac{1}{4}
\sum_{i=1}^{N}
(i-1)^2\,\chi_{a,i} .
\label{eq:marginal_S_H}
\end{equation}
The $N^3$ growth derived below rests on one assumption. The quantity
$\chi_{a,i}$ must be bounded below uniformly, i.e.\ it must not become arbitrarily small: there must be a constant
$\chi_0>0$ such that $\chi_{a,i}\ge\chi_0$ at every site $i$ and for every sequence length $N$. Since $\chi_{a,i}$ measures the softmax covariance along the readout direction $\mathbf u_a$, this says that the softmax of the unperturbed model must not saturate in that direction.

This assumption is imposed on the unperturbed model; it does not follow from the recurrence multiplier alone. When $r_a=1$, the derivative of the logits with respect to the recurrent bias accumulates linearly along the chain, as Eq.~\eqref{eq:marginal_m_H} shows. This derivative growth does not by itself imply that the unperturbed logits grow with $i$. Uniform non-saturation is therefore an independent hypothesis. If the unperturbed softmax does saturate, then $C_i=\operatorname{diag}(\mathbf p_i)-\mathbf p_i\mathbf p_i^{\mathsf T}\to0$ and the lower bound can weaken because $\chi_{a,i}\to0$.

Under this assumption we may replace $\chi_{a,i}$ in Eq.~\eqref{eq:marginal_S_H} by its lower bound $\chi_0$, which gives
$S_{b_a b_a}\ge\tfrac{\chi_0}{4}\sum_{i=1}^{N}(i-1)^2$. The remaining sum has a closed form,
\begin{equation}
\sum_{i=1}^{N}(i-1)^2
=
\frac{(N-1)N(2N-1)}{6}
\sim
\frac{N^3}{3} ,
\label{eq:cubic_sum_H}
\end{equation}
and the two prefactors $\tfrac14$ and $\tfrac13$ combine to give
\begin{equation}
\lambda_{\max}(S)
\ge
S_{b_a b_a}
\gtrsim
\frac{\chi_0}{12}\,N^3
\sim
N^3 ,
\label{eq:cubic_growth_H}
\end{equation}
so the largest eigenvalue grows polynomially with the sequence length.

\emph{Expanding mode, $|r_a|>1$.} Here $r_a\neq1$, so the finite geometric series of Eq.~\eqref{eq:geometric_factor_H} can be summed in closed form. Its $i-1$ terms give
\begin{equation}
m_{i-1}(r_a)
=
\sum_{k=0}^{i-2} r_a^{\,k}
=
\frac{1-r_a^{\,i-1}}{1-r_a} .
\label{eq:expanding_m_H}
\end{equation}
We substitute this into Eq.~\eqref{eq:Sbaba_H} and again assume that the softmax covariance does not saturate along $\mathbf u_a$, so that $\chi_{a,i}\ge\chi_0>0$ uniformly. Since every term in Eq.~\eqref{eq:Sbaba_H} is nonnegative, retaining only the $i=N$ term gives the rigorous lower bound
\begin{align}
S_{b_a b_a}
&\ge
\frac{1}{4}\,|m_{N-1}(r_a)|^2\,\chi_{a,N}
\nonumber\\
&=\frac{1}{4}\,
\frac{|1-r_a^{\,N-1}|^2}{|1-r_a|^2}\,\chi_{a,N}
=\Omega\!\left(|r_a|^{2(N-1)}\right),
\label{eq:expanding_S_H}
\end{align}
where the last statement uses $\chi_{a,N}\ge\chi_0$ and fixed $|r_a|>1$. Thus $S_{b_a b_a}$ , and therefore $\lambda_{\max}(S)$, grows at least exponentially in $N$.

The asymptotic form $m_{i-1}\sim r_a^{\,i-1}/(r_a-1)$ is needed only to read off the growth of this final retained term. It need not be substituted term by term inside the full sum; at $i=1$, for instance, the exact value is $m_0=0$.
Two regimes remain, one for each part of the proposition. The first is the marginal case $r_a=1$: replacing every weight $\chi_{a,i}$ in Eq.~\eqref{eq:marginal_S_H} by its uniform lower bound $\chi_0$ and evaluating the sum of squares with Eq.~\eqref{eq:cubic_sum_H} gives the bound of Eq.~\eqref{eq:prop_marginal_bound}, which grows as $N^3$ (Eq.~\eqref{eq:cubic_growth_H}). The second is the expanding case $|r_a|>1$: the sum in Eq.~\eqref{eq:large_eig_bound_H} has only nonnegative terms, so keeping only the $i=N$ term already gives a valid lower bound at every $N$, Eq.~\eqref{eq:prop_expanding_bound}, and it grows exponentially in $N$ (Eq.~\eqref{eq:expanding_S_H}).
\end{proof}

\subsection*{C.7 Summary}

We now collect the two ends of the spectrum. Together they are what
makes $S$ ill-conditioned. The bottom of the spectrum comes from Sec.~C.2. A change of hidden-coordinate basis $A\in GL(d_h)$ moves the parameters but leaves the wave function unchanged. Along the directions generated by such a change, $d\log\Psi=0$, and this holds configuration by configuration, not only on average. Since $S=\bar{\mathcal O}\bar{\mathcal O}^{\dagger}$ is built from
the sample-centered log-derivative matrix $\bar{\mathcal O}$, its quadratic form vanishes along those directions. They are therefore exact zero eigenvalues of $S$, at any sample size. These are the Fisher-information singularities.

The top of the spectrum comes from Secs.~C.5--C.6. There we probed a single recurrent-bias direction and applied the Rayleigh bound of
Eq.~\eqref{eq:large_eig_bound_H}, which bounds $\lambda_{\max}(S)$ below by a single diagonal element of $S$. How fast that element grows is set by the multiplier $r_a$ of the hidden mode being probed. A marginal mode ($r_a=1$) gives $\lambda_{\max}(S)\gtrsim N^3$, polynomial growth in the sequence length
$N$; an expanding mode ($|r_a|>1$) gives growth that is exponential in $N$. Both statements assume that $\chi_{a,i}$ stays bounded below by a positive constant, uniformly in the site index $i$ and in $N$.

For the population matrix of the simplified model, an exact gauge zero mode together with a diverging top makes the condition number unbounded. Further small eigenvalues can arise from softmax saturation ($C_i\to0$) and sampling fluctuations. In information-geometry language, the Fisher spectrum encodes the local curvature of the variational manifold. The gauge zero modes are flat, redundant directions; the recurrently amplified modes are strongly curved ones. A broad spectrum therefore reflects a highly anisotropic geometry, precisely the regime that motivates the minSR conditioning strategy used in this work.

\section{Hyperparameter tuning}
\label{app:hyperparameter}
We tune our hyperparameters with the Tree-Structured Parzen Estimator (TPE) using the Optuna Python library. We list the results in Tab.~\ref{tab:combined_hyperparams}.

\begin{table*}[ht]
    \centering
    \renewcommand{\arraystretch}{1.4}
    \begin{tabular}{|l|l|c|c|c|c|}
        \hline
        \multicolumn{2}{|c|}{\textbf{Hyperparameter / Model}} & \textbf{1D TFIM} & \textbf{1D Cluster State} & \textbf{2D Heisenberg} & \textbf{2D $J_1-J_2$} \\ \hline\hline
        \multicolumn{2}{|c|}{\textbf{Figure Ref.}} & Fig.~\ref{fig:TFIM}(a) & Fig.~\ref{fig:TFIM}(b) & Fig.~\ref{fig:J1J2}(a) & Fig.~\ref{fig:J1J2}(b) \\ \hline
        \multicolumn{2}{|c|}{\textbf{GPU}} & H100 & H100 & A100 & A100 \\ \hline
        \multicolumn{2}{|c|}{\textbf{Hidden Dim ($d_h$)}} & 32 & 256 & 200 & 200 \\ \hline
        \multicolumn{2}{|c|}{\textbf{Samples ($N_s$)}} & 100 & 100 & 200 & 200 \\ \hline\hline

        & Learning Rate & $5\times10^{-1}$ & $5\times10^{-3}$ & $6\times10^{-2}$ & $1\times10^{-1}\cdot \|\bm\tau\|^{-1}$ \\ \cline{2-6}
        & Momentum & 0.7 & 0.1 & 0.75 & 0.1 \\ \cline{2-6}
        & Regularization ($\lambda$) & $10^{-3}$ & $10^{-3}$ & $3\times10^{-4}\cdot\|\mathcal{T}\|^{2/3}$ & $10^{-1}$ \\ \cline{2-6}
        & Decay Time &$10^3 $&$5\times10^3$ &$10^4$ &$10^4$ \\ \cline{2-6}
        \multirow{-5}{*}{\textbf{minSR}} & Iterations &$10^4$ &$10^4$ &$10^5$ &  $10^5$\\ \hline\hline

        & Learning Rate & $5\times10^{-3}$ & $5\times10^{-4}$ & $5\times10^{-4}$ & $3\times10^{-3}$ \\ \cline{2-6}
        & Decay Time $T$ &$\infty$ &$\infty$ &$10^4$ &$5\times10^3$ \\ \cline{2-6}
        \multirow{-3}{*}{\textbf{Adam}} & Iterations &$10^4$ & $10^4$ & $2\times10^5$&$2\times10^5$ \\ \hline
    \end{tabular}
    \caption{\textbf{Optimization Hyperparameters.} Combined layout displaying structural constants followed by algorithm-specific parameters for minSR and Adam. The learning rate is sometimes decayed as $\eta = \eta_0(1+\frac{i}{T})^{-1}$, where $\eta$ and $\eta_0$ are the learning rate and the initial learning rate, respectively, $i$ is the training iteration, and $T$ is the learning rate decay time. We note that $T=\infty$ corresponds to keeping a constant learning rate.}
    \label{tab:combined_hyperparams}
\end{table*}

\begin{table*}[tbp]
    \centering
    \renewcommand{\arraystretch}{1.4}
    \begin{tabular}{|c|c|c|c|c|}
    \hline\textbf{Hamiltonian}&\textbf{$E_{\textrm{Adam}}/N$}&\textbf{$E_{\textrm{minSR}}/N$} & $E_{\textrm{ref}}/N$\\ \hline
        \textbf{1D TFIM}& -1.271208(7)& -1.2714254(3)&-1.27142590795~\cite{ising}\\ \hline
        \textbf{1D Cluster} &-0.8760(2) & -0.999947(4)& -1~\cite{cluster-preskill}\\ \hline
        \textbf{2D Heisenberg} & -0.62840(2)&-0.628486(6) &-0.628656~\cite{QMCH2D}\\ \hline
        \textbf{2D $J_1-J_2$}&$-0.48553(2)$& -0.48514(4) &-0.48654~\cite{PEPSbench}\\ \hline
    \end{tabular}
    \caption{ \textbf{Final variational energy results.} Final energies of Adam and minSR for each hamiltonian evaluated at the end of training.}
    \label{tab:energies}
\end{table*}

\section{Training Instabilities}
\label{app:instabilities}
SR and minSR can suffer from instability during training, which manifests as transient spikes in the average local energy~\cite{PII} or as failures due to NaNs. We find this instability is particularly pronounced at low floating-point precision and when $\lambda$ is too small. To understand this, we must look at the condition number $\kappa$ of our minSR step, which describes the sensitivity of our solution to noise. For the parameter-space system, define
\begin{equation}
    S_\lambda \equiv S+\lambda\mathbb I .
\end{equation}
Because $S$ is positive semidefinite,
\begin{equation}
    \kappa_2(S_\lambda)
    =\frac{\lambda_{\max}(S)+\lambda}{\lambda_{\min}(S)+\lambda}.
\end{equation}
$S$ is ill-conditioned, with $\lambda_{\min}(S)=0$ ( see App.~\ref{app:illcond_detailed}) so we have
\begin{equation}
    \kappa_2(S_\lambda)=1+\frac{\lambda_{\max}(S)}{\lambda}
    \approx \frac{\lambda_{\max}(S)}{\lambda}
\end{equation}
when $\lambda\ll\lambda_{\max}(S)$. The sample-space matrix $\mathcal T=\bar{\mathcal O}^{'\mathsf T}\bar{\mathcal O}'$ has the same nonzero eigenvalues as $S=\bar{\mathcal O}'\bar{\mathcal O}^{'\mathsf T}$; hence the same large-eigenvalue scale controls the conditioning of the damped minSR solve.

For a fixed coefficient matrix, perturbation theory for a linear solve gives, to first order in a perturbation of the right-hand side~\cite{mixed-precision,Accuracy}
\begin{align}
    \frac{\|\Delta(\delta\bm\theta)\|_2}{\|\delta\bm\theta\|_2}
    \lesssim \kappa_2(S_\lambda)\frac{\|\Delta g\|_2}{\|g\|_2}.
    \label{eq:errors}
\end{align}
Thus a standard floating-point estimate for the attainable relative forward accuracy is of order $\kappa_2(S_\lambda)\varepsilon_{\mathrm{machine}}$, up to algorithm-dependent constants and additional sampling/model error. When this product is much smaller than one, the corresponding rule-of-thumb number of reliable decimal digits is
\begin{align}
    d  \approx -\log_{10} \mh{\!}\big( \kappa _2 ( S_\lambda)\varepsilon_{\mathrm{machine}} \big) .
\end{align}
Consequently, in the regime $\lambda_{\min}(S)=0$ and $\lambda\ll\lambda_{\max}(S)$, decreasing $\lambda$ by one decade costs roughly one decimal digit of numerical accuracy. Additionally, in single floating-point precision, $\varepsilon_{\mathrm{machine}}\approx 10^{-7}$, which may result in a very small value of $d$, leading to instability. This conditioning estimate is consistent with the instability observed in Fig.~\ref{fig:instability}; it should be interpreted as an order-of-magnitude numerical diagnostic rather than as a strict universal error bound. 

\begin{figure}
    \centering
    \includegraphics[width=\linewidth]{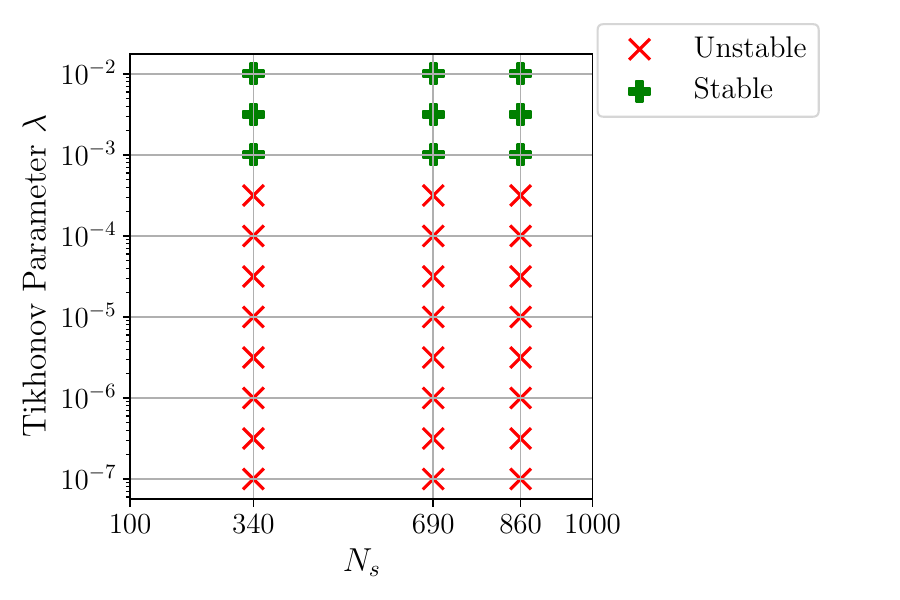}
    \caption{\textbf{Instability of minSR training.} We plot the training results of a 1D RNN wave function on the TFIM with $d_h=32$ for various Tikhonov parameters $\lambda$ and number of samples $N_s$. We run each experiment for $100$ training steps and mark the plot with an 'x' marker if training produced NaNs. Otherwise, a '+' marker is used if training is stable.}
    \label{fig:instability}
\end{figure}

\section{Variance-controlled learning rate for SR}

\label{app:trust_region}
Within the SR setup, a natural (infinitesimal) distance between quantum states is the Fubini--Study distance:
$$
d_{\mathrm{FS}}^2(\theta,\theta+\delta\theta)\equiv 1-\big|\langle\Psi_\theta|\Psi_{\theta+\delta\theta}\rangle\big|^2.
$$
For small $\delta\theta$, its second-order expansion is
$$
d_{\mathrm{FS}}^2(\theta,\theta+\delta\theta)\approx \delta\theta^{\mathsf T} S\,\delta\theta,
$$
where $S=\operatorname{Re}Q$ is the SR metric defined in Eq.~\eqref{eq:SR}. This yields the trust region constraint
$$
\delta\theta^{\mathsf T} S\,\delta\theta\le \varepsilon,
$$
given a small trust-region radius $\varepsilon$. Taylor expanding the energy to first-order,
$$
E(\theta+\delta\theta)\approx E(\theta)+g^{\mathsf T}\delta\theta,
\qquad
g_k \equiv \frac{\partial E}{\partial \theta_k},
$$
SR corresponds to the trust-region problem:
$$
\min_{\delta\theta}\; g^{\mathsf T}\delta\theta
\qquad \text{s.t.}\qquad
\delta\theta^{\mathsf T} S\,\delta\theta\le \varepsilon.
$$
The latter, augmented with a Euclidean trust region $\delta\theta^{\mathsf T}\delta\theta\le \varepsilon'$ that controls the step along the flat directions of $S$, can be formulated using the Lagrangian
$$
\mathcal L(\delta\theta,\mu, \nu
)=g^{\mathsf T}\delta\theta+\mu\big(\delta\theta^{\mathsf T}S\delta\theta-\varepsilon\big)+\nu\big(\delta\theta^{\mathsf T}\delta\theta-\varepsilon'\big),
\quad \nu,\mu\ge 0.
$$
Here the first-order condition of optimality (stationarity) in $\delta\theta$ gives
$$
\nabla_{\delta\theta}\mathcal L =g+2\mu \left(S+\frac{\nu}{\mu}\mathbb{I}\right)\delta\theta = 0.
$$
Writing $\eta\equiv \frac{1}{2\mu}$ for the learning rate and $\lambda \equiv \frac{\nu}{\mu}$, yields the SR step
$$
\delta\theta=-\eta\,(S+\lambda\mathbb{I})^{-1}g.
$$
To get the largest step that remains inside the trust region, we impose the constraint:
$$
\delta\theta^{\mathsf T}S\delta\theta=\varepsilon.
$$
For $\lambda=0$, the finite-sample $S$ is generally singular. Since the SR gradient direction lies in $\operatorname{range}(S)$, the appropriate inverse is the Moore--Penrose pseudoinverse $S^{+}$. Substituting $\delta\theta=-\eta_*S^{+}g$ into the active constraint gives
$$
(-\eta_* S^{+}g)^{\mathsf T}S(-\eta_* S^{+}g)=\varepsilon,
$$
and hence $\eta_*=\sqrt{\varepsilon/(g^{\mathsf T}S^{+}g)}$ for $g\neq0$. We now repeat this computation for general $\lambda>0$. We use the real-valued system of equations in Eq.~\eqref{eq:real-valued}. We follow the conventions from Ref.~\cite{minSR-lianlg}
\begin{align}
    g = \bar{\mathcal{O}}'\bar{\boldsymbol\epsilon}', \quad S = \bar{\mathcal{O}}'\bar{\mathcal{O}}'^\top.
\end{align}
By substituting into the constraint
$$
(-\eta_* (S+\lambda\mathbb{I})^{-1}g)^{\mathsf T}S(-\eta_* (S+\lambda\mathbb{I})^{-1}g)=\varepsilon,
$$
\begin{align*}
\varepsilon&=\eta_* ^2g^\mathsf T(S+\lambda\mathbb{I})^{-\mathsf{T}}S(S+\lambda\mathbb{I})^{-1}g,\\
&=\eta_* ^2\bar{\bm\epsilon}^{'\mathsf T}\bar{\mathcal{O}}^{'\mathsf T}(S+\lambda\mathbb{I})^{-1}S(S+\lambda\mathbb{I})^{-1}\bar{\mathcal{O}}'\bar{\bm\epsilon}'.\end{align*}
We can apply the linear algebra trick
\begin{align*}
=\eta_* ^2\bar{\bm\epsilon}^{'\mathsf T}(\mathcal T+\lambda\mathbb{I})^{-1}\bar{\mathcal{O}}^{'\mathsf T}S\bar{\mathcal{O}}'(\mathcal T+\lambda\mathbb{I})^{-1}\bar{\bm\epsilon}'\\
=\eta_* ^2\bar{\bm\epsilon}^{'\mathsf T}(\mathcal T+\lambda\mathbb{I})^{-1}\mathcal{T}^{2}(\mathcal T+\lambda\mathbb{I})^{-1}\bar{\bm\epsilon}'
\end{align*}
Using the equation
\begin{equation}
    \mathcal{T}(\mathcal T+\lambda\mathbb{I})^{-1}=\mathbb{I}-\lambda(\mathcal T+\lambda\mathbb{I})^{-1}
\end{equation}
we get
\begin{align*}
    =\eta_* ^2\bar{\bm\epsilon}^{'\mathsf T}\Big[\mathbb{I}-\lambda(\mathcal T+\lambda\mathbb{I})^{-1}\Big]^2\bar{\bm\epsilon}'.
\end{align*}
Thus our ideal learning rate is
\begin{equation*}
    \frac{\eta_*}{\sqrt{\varepsilon}} = \Big\|\big(\mathbb{I}-\lambda(\mathcal T+\lambda\mathbb{I})^{-1}\big)\bar{\bm\epsilon}'\Big\|^{-1}.\label{eq:tr2}
\end{equation*}
Treating the trust-region radius as a tunable base learning rate $\eta\equiv\sqrt{\varepsilon}$, we have
\begin{align}
    \eta_* =\frac{\eta}{\|\bm\tau\|}, \qquad \bm\tau = \big(\mathbb{I}-\lambda(\mathcal T+\lambda\mathbb{I})^{-1}\big)\bar{\bm\epsilon}'
    \label{eq:trust-region-lr}
\end{align}

\section{Damping Scheme}
\label{app:levenberg_marquardt}
The appropriate choice of the damping parameter $\lambda$ is essential for stable training. In the derivation of App.~\ref{app:trust_region}, $\lambda=\nu/\mu$ is the ratio of the Lagrange multipliers associated with the Euclidean and Fubini--Study trust-region constraints. Larger $\lambda$ therefore suppresses motion in poorly constrained parameter directions more strongly, but $\lambda^{-1}$ should not be identified with an exact universal trust-region radius. Although a fixed $\lambda$ is sufficient for our one-dimensional benchmarks, adaptive damping can be useful more generally, especially in the more ill-conditioned two-dimensional problems. Common schedules use larger damping early in training and smaller damping closer to convergence~\cite{exploration-phase,kfac}.

Previously introduced methods include a decaying schedule for $\lambda$~\cite{Carleo_original_2017}, and an adaptive cutoff based on the signal-to-noise ratio~\cite{PhysRevLett.125.100503}.
The LM heuristic is commonly used in Natural Gradients and Gauss-Newton optimization~\cite{NG-martens} to adjust $\lambda$ based on the accuracy of our quadratic model~\cite{more2006levenberg,numerical.opt.textbook}. When a step underperforms, $\lambda$ is increased by a factor $\omega$, and when it overperforms, $\lambda$ is reduced by a factor of $\omega^{-1}$, signaling a higher degree of `trust' in the least-squares solution. We present an implementation of the LM heuristic above.

The LM heuristic interpolates between gradient descent and SR by adjusting $\lambda$ at every training step. The ratio $\rho$ is given by~\cite{nqs.kfac}
\begin{align*}
    \rho = \frac{E_{\bm\theta} -E_{\bm\theta+\delta\bm\theta}}{E_{\bm\theta}-\Big[E_{\bm\theta}+\nabla E_{\bm\theta}^{\mathsf T}\delta\bm\theta+\frac{1}{2}\delta\bm\theta^{\mathsf T}S\delta\bm\theta\Big]}.
\end{align*}
The numerator involves the differences in expected energy as defined in Eq.~\eqref{eq:energy}. The denominator is an approximation of the numerator with $E_{\bm\theta+\delta\bm\theta}$ replaced by a quadratic model about $\bm\theta$, in which the QGT $S$ would play the role of the Hessian of the energy. There is an alternative linear model approach to $\rho$ applied in \cite{numerical.opt.textbook,more2006levenberg}. The denominator can be simplified as 
\begin{align*}
    \rho = \frac{E_{\bm\theta} -E_{\bm\theta+\delta\bm\theta}}{-\bar{\bm\epsilon}^{'\mathsf T}\bar{\mathcal{O}}^{'\mathsf T}\delta\bm\theta-\frac{1}{2}\eta^2}.
\end{align*}
The quadratic term simplifies due to our trust region: the step saturates the trust-region constraint of App.~\ref{app:trust_region}, so $\delta\bm\theta^{\mathsf T}S\,\delta\bm\theta=\varepsilon=\eta^{2}$. This further simplifies to
\begin{align*}
    \rho = \frac{E_{\bm\theta} -E_{\bm\theta+\delta\bm\theta}}{\eta_*\bar{\bm\epsilon}^{'\mathsf T}\bm\tau-\frac{1}{2}\eta^2}.
\end{align*}
Then $\lambda$ is updated as follows \cite{nqs.kfac}
\begin{align}
    \lambda =\begin{cases}
        \omega\lambda , & \rho<\frac{1}{4},\\
        \omega^{-1}\lambda , & \rho>\frac{3}{4},\\
        \lambda\textcolor{purple}{,} & \textrm{otherwise}.
    \end{cases}
\end{align}
Here $\omega>1$ is a hyperparameter often chosen to be $\frac{3}{2}$.

The LM heuristic requires an additional forward pass to calculate $E_{\bm\theta+\delta\bm\theta}$, though this does not change the runtime significantly since the backward passes dominate the computational cost.

\section{Proof of the subsampled-QGT error bound}
\label{app:proof_subsampled}

This appendix proves Proposition~\ref{prop:subsampled_fisher}. The configurations $\bm\sigma^{(1)},\dots,\bm\sigma^{(N_s)}$ are drawn i.i.d.\ from the Born distribution $|\Psi_{\bm\theta}|^2$; exact autoregressive sampling guarantees this independence. For each parameter $i=1,\dots,N_p$ we write $O_i(\bm\sigma)=\partial_{\theta_i}\log\Psi_{\bm\theta}(\bm\sigma)$ for the log-derivative, which is complex when the ansatz is complex. In the notation of the main text, $\mathcal O_{is}=O_i(\bm\sigma^{(s)})/\sqrt{N_s}$. We define
\begin{equation}
\mu_i=\mathbb{E}[O_i],
\qquad
X_i=O_i-\mu_i.
\end{equation}
Throughout, $\mathbb{E}[\cdot]$ denotes the expectation and $\Pr[\cdot]$ the probability of an event, both taken over the sample draw at fixed $\bm\theta$. The two matrices we compare are the population QGT
\begin{equation}
Q_{ij}=\mathbb{E}[X_iX_j^{*}]
\end{equation}
and its population-centered estimator $\widetilde Q$ of Eq.~\eqref{eq:Qtilde_est}; both are $N_p\times N_p$ and Hermitian.

For a complex random variable $Z$ we use $\operatorname{Var}(Z)\equiv\mathbb{E}[|Z-\mathbb{E}[Z]|^2]$. This variance is additive over independent summands, because it equals $\operatorname{Var}(\operatorname{Re}Z)+\operatorname{Var}(\operatorname{Im}Z)$ and each real part is additive.

The proof uses exactly two hypotheses: (i) i.i.d.\ sampling from $|\Psi_{\bm\theta}|^2$, and (ii) finiteness of $V_{\mathrm{avg}}^{(c)}$ in Eq.~\eqref{eq:vars1}. Hypothesis (ii) is not restrictive. Since $\operatorname{Var}(X_iX_j^{*})\le\mathbb{E}[|X_i|^2|X_j|^2]$ and, by the Cauchy--Schwarz inequality, $\mathbb{E}[|X_i|^2|X_j|^2]\le\sqrt{\mathbb{E}[|X_i|^4]\,\mathbb{E}[|X_j|^4]}$, it holds whenever the centered log-derivatives have finite fourth moments under $|\Psi_{\bm\theta}|^2$.

\begin{proof}[Proof of Proposition~\ref{prop:subsampled_fisher}]
The estimator is an average of one term per sample, so its error is pure sampling noise. We follow that noise from a single entry, to the whole matrix, and finally into a probability.

We start by establishing unbiasedness. Define the per-sample product
\begin{equation}
Y_{ij}^{(s)}
=
X_i\big(\bm\sigma^{(s)}\big)\,
X_j\big(\bm\sigma^{(s)}\big)^{*},
\label{eq:Xdef_proof}
\end{equation}
so that $\widetilde Q_{ij}=\frac{1}{N_s}\sum_{s}Y_{ij}^{(s)}$ by Eq.~\eqref{eq:Qtilde_est}. The samples are identically distributed, so $\mathbb{E}[Y_{ij}^{(s)}]=Q_{ij}$ for every $s$. Linearity of expectation then gives $\mathbb{E}[\widetilde Q_{ij}]=Q_{ij}$: the estimator is unbiased, and its error is exactly the fluctuation of the average about its mean.

Because the estimator is unbiased, that fluctuation is measured by the variance of each entry. Fix $(i,j)$. The terms $Y_{ij}^{(1)},\dots,Y_{ij}^{(N_s)}$ are functions of independent samples, hence independent, and each has the single-sample variance $\operatorname{Var}(X_iX_j^{*})$, finite by hypothesis (ii). Additivity of the variance together with the scaling $\operatorname{Var}(cZ)=|c|^2\operatorname{Var}(Z)$ gives
\begin{align}
\operatorname{Var}\big(\widetilde Q_{ij}\big)
&=\frac{1}{N_s^{2}}\sum_{s=1}^{N_s}\operatorname{Var}\big(Y_{ij}^{(s)}\big)
\nonumber\\
&=\frac{1}{N_s}\,\operatorname{Var}\big(X_iX_j^{*}\big).
\label{eq:elem_var_proof}
\end{align}
Averaging $N_s$ independent copies divides the variance by $N_s$; this single factor $1/N_s$ is the only source of $N_s$-dependence in the bound.

Summing these entrywise variances produces the expected squared error of the whole matrix. By definition $\|A\|_F^2=\sum_{i,j}|A_{ij}|^2$, and unbiasedness turns each summand into a variance, $\mathbb{E}[|\widetilde Q_{ij}-Q_{ij}|^2]=\operatorname{Var}(\widetilde Q_{ij})$. Summing Eq.~\eqref{eq:elem_var_proof} over all $N_p^2$ pairs and inserting Eq.~\eqref{eq:vars1} yields
\begin{align}
\mathbb{E}[\|\widetilde Q-Q\|_F^{2}]
&=\frac{1}{N_s}\sum_{i,j=1}^{N_p}\operatorname{Var}\big(X_iX_j^{*}\big)
\nonumber\\
&=\frac{N_p^{2}}{N_s}\,V_{\mathrm{avg}}^{(c)}.
\label{eq:frob_second_moment}
\end{align}
The factor $N_p^{2}$ counts the entries, and the $1/N_s$ carries over unchanged from each of them.

A high-probability bound now follows from Markov's inequality. Since $\|\widetilde Q-Q\|_F^{2}$ is nonnegative, for any $t>0$
\begin{align}
\Pr\big[\|\widetilde Q-Q\|_F\ge t\big]
&=\Pr\big[\|\widetilde Q-Q\|_F^{2}\ge t^{2}\big]
\nonumber\\
&\le\frac{\mathbb{E}[\|\widetilde Q-Q\|_F^{2}]}{t^{2}}
=\frac{N_p^{2}\,V_{\mathrm{avg}}^{(c)}}{N_s\,t^{2}},
\label{eq:markov_proof}
\end{align}
where the first equality holds because squaring is monotone on nonnegative reals. Choosing
\begin{equation}
t=\frac{N_p}{\sqrt{N_s}}\sqrt{\frac{V_{\mathrm{avg}}^{(c)}}{\epsilon}}
\label{eq:t_choice_proof}
\end{equation}
makes the right-hand side of Eq.~\eqref{eq:markov_proof} equal to $\epsilon$. The reverse inequality $\|\widetilde Q-Q\|_F<t$ is exactly Eq.~\eqref{eq:UB}, so it holds with probability at least $1-\epsilon$. The square root taken here is what turns the entrywise factor $1/N_s$ into the $N_s^{-1/2}$ rate carried by the norm.

The spectral norm inherits the same bound. For any matrix $A$, $\|A\|_2^2$ is the largest squared singular value, while $\|A\|_F^2$ is the sum of all of them; hence $\|A\|_2\le\|A\|_F$, and Eq.~\eqref{eq:UB} also bounds $\|Q-\widetilde Q\|_2$.

For real variational parameters, the population SR metric is
$S=\operatorname{Re}Q$. Defining
$\widetilde S=\operatorname{Re}\widetilde Q$, we have
\begin{equation}
\|S-\widetilde S\|_F
=
\|\operatorname{Re}(Q-\widetilde Q)\|_F
\le
\|Q-\widetilde Q\|_F,
\end{equation}
so the same bound applies to the SR metric.

The proof never assumed that $\Psi_{\bm\theta}$ is real or positive. Every step relied only on linearity of $\mathbb{E}$, on the additivity of $\operatorname{Var}(Z)=\mathbb{E}[|Z-\mathbb{E}[Z]|^2]$ over independent terms, and on hypothesis (ii). The same QGT bound therefore applies to complex RNN wave functions (cRNN) just as it does to positive ones (pRNN).
\end{proof}


\emph{Sample-mean centering.} Proposition~\ref{prop:subsampled_fisher} uses the population-centered estimator $\widetilde Q$, in which the exact population mean $\mu_i=\mathbb{E}[O_i]$ is subtracted. In practice, we use the sample-centered matrix $\widehat Q=\bar{\mathcal O}\bar{\mathcal O}^{\dagger}$, in which each log-derivative has its sample mean
\begin{equation}
\hat\mu_i=\frac{1}{N_s}\sum_sO_i(\bm\sigma^{(s)})
\end{equation}
subtracted. We now check that this sample centering does not change the $N_s^{-1/2}$ rate. Since
\begin{equation}
\hat{\bm\mu}-\bm\mu
=
\frac{1}{N_s}
\sum_s
\bm X\big(\bm\sigma^{(s)}\big),
\end{equation}
multiplying out the sample-centered products gives
\begin{equation}
\widehat Q
=
\widetilde Q
-
(\hat{\bm\mu}-\bm\mu)
(\hat{\bm\mu}-\bm\mu)^{\dagger}.
\label{eq:centered_defs_proof}
\end{equation}
Subtracting $Q$ gives
\begin{equation}
\widehat Q-Q
=
\big(\widetilde Q-Q\big)
-
(\hat{\bm\mu}-\bm\mu)
(\hat{\bm\mu}-\bm\mu)^{\dagger}.
\label{eq:centered_split_proof}
\end{equation}
The first difference is controlled by Proposition~\ref{prop:subsampled_fisher}. For the second, the rank-one identity
$\|\bm x\bm y^{\dagger}\|_F=\|\bm x\|_2\|\bm y\|_2$
gives
\begin{equation}
\big\|
(\hat{\bm\mu}-\bm\mu)
(\hat{\bm\mu}-\bm\mu)^{\dagger}
\big\|_F
=
\|\hat{\bm\mu}-\bm\mu\|_2^2.
\label{eq:mean_term_proof}
\end{equation}
The mean correction is therefore governed by
$\|\hat{\bm\mu}-\bm\mu\|_2^2$, and
$\hat{\bm\mu}$ is itself an i.i.d.\ average:
\begin{equation}
\mathbb{E}
\left[
\|\hat{\bm\mu}-\bm\mu\|_2^2
\right]
=
\frac{1}{N_s}
\sum_{i = 1}^{N_p}\operatorname{Var}(O_i),
\end{equation}
where the last sum runs over all parameter indices. Thus, if the average variance per parameter remains $O(1)$ as $N_p$ varies, this term scales linearly with $N_p$. The correction due to sample-mean centering is therefore of order $N_s^{-1}$ in expectation and is subleading to the
$N_s^{-1/2}$ fluctuation bounded in Proposition~\ref{prop:subsampled_fisher}. Thus sample centering does not change the leading rate.

A few further remarks complete the picture. First, at finite $N_s$ the sample-centered estimator has a small, non-random bias. Using
\begin{equation}
\mathbb{E}
\left[
(\hat{\bm\mu}-\bm\mu)
(\hat{\bm\mu}-\bm\mu)^{\dagger}
\right]
=
\frac{Q}{N_s}
\end{equation}
gives
\begin{equation}
\mathbb{E}[\widehat Q]
=
\frac{N_s-1}{N_s}Q,
\end{equation}
a bias of $-Q/N_s$; because it scales as $1/N_s$, it is smaller than the $N_s^{-1/2}$ random fluctuation.

Second, for a normalized ansatz,
$\sum_{\bm\sigma}|\Psi_{\bm\theta}(\bm\sigma)|^2=1$,
the mean log-derivative has zero real part,
$\operatorname{Re}\mathbb{E}[\partial_{\theta_i}\log\Psi_{\bm\theta}]=0$;
for a positive ansatz (pRNN) the mean vanishes entirely,
$\bm\mu=\bm 0$.

Finally, in Fig.~\ref{fig:scaling}(b) the reference matrix is itself an estimate from a finite sample with $N_s\gtrsim N_p$. By the triangle inequality, the measured error contains contributions from both sample sizes; when the reference sample size is sufficiently large, its contribution is smaller, and the leading $N_s^{-1/2}$ rate is unchanged.

\end{document}